\documentclass[11pt]{article}
\usepackage[margin=1.1in]{geometry}
\usepackage{amsmath,amssymb,amsthm,mathtools}
\usepackage{booktabs,array,graphicx}
\usepackage[numbers,sort&compress]{natbib}
\usepackage{microtype}
\usepackage{xcolor}
\definecolor{linkblue}{rgb}{0.1,0.2,0.55}
\usepackage[colorlinks=true,linkcolor=linkblue,citecolor=linkblue,urlcolor=linkblue]{hyperref}
\usepackage{authblk}
\theoremstyle{plain}
\newtheorem{theorem}{Theorem}[section]
\newtheorem{lemma}[theorem]{Lemma}
\newtheorem{proposition}[theorem]{Proposition}
\newtheorem{corollary}[theorem]{Corollary}
\theoremstyle{definition}
\newtheorem{definition}[theorem]{Definition}
\newtheorem{example}[theorem]{Example}
\theoremstyle{remark}
\newtheorem{rmk}[theorem]{Remark}
\newenvironment{thm}{\begin{theorem}}{\end{theorem}}
\newenvironment{lem}{\begin{lemma}}{\end{lemma}}
\newenvironment{prop}{\begin{proposition}}{\end{proposition}}
\newenvironment{cor}{\begin{corollary}}{\end{corollary}}

\newenvironment{exmp}{\begin{example}}{\end{example}}
\providecommand{\Description}[1]{}

\newcommand{\E}{\mathbb{E}}
\newcommand{\indic}{\mathbf{1}}
\newcommand{\inn}[2]{\langle #1,#2\rangle}
\newcommand{\Iset}{\mathcal{I}}
\newcommand{\Gset}{\mathcal{G}}
\newcommand{\Sset}{\mathcal{S}}
\newcommand{\Pset}{\mathcal{P}}
\newcommand{\PhiI}{\Phi_{\mathrm I}}
\newcommand{\PhiM}{\Phi_{\pi}}
\newcommand{\Fcost}{\mathcal{F}_{\mathrm{cost}}}
\newcommand{\Ftwo}{\mathcal{F}_{\mathrm{two}}}
\newcommand{\Pio}{\Pi_{\mathrm{F}}}

\newcommand{\com}[1]{}
\newcommand{\comr}[1]{}
\newcommand{\resp}[1]{}

\title{Strategic Classification Has a Missing Lever: Audit Risk}
\author[1]{Raman Ebrahimi}
\author[1]{Massimo Franceschetti}
\affil[1]{Department of Electrical and Computer Engineering, University of California, San Diego\\ \texttt{\{raman,mfranceschetti\}@ucsd.edu}}
\date{\today}

\begin{document}
\maketitle

\begin{abstract}
Strategic classification studies how a decision maker should choose a classifier when the agents being classified can adjust their features in response to it. In existing models, the classifier is the only instrument available to the decision maker, and therefore a feature that is predictive but easy to fake can only be down-weighted or discarded. However, in many settings the decision maker can also verify: lenders verify income, admissions offices check documents, and tax authorities audit returns. In this paper, we propose a model of strategic classification in which the firm jointly designs a linear classifier and an \emph{audit profile}, which assigns to each fakeable feature a probability of detection and a penalty when caught. We show that under linear costs, the classifier affects the audit problem only through the distribution of gaming rents it induces, so that the joint design problem decomposes into the choice of a score rule and an audit allocation problem. We use this decomposition to characterize the optimal audit allocation, to identify when the allocation problem is tractable and when it is NP-hard (namely, when agents can game through overlapping features under an inspection cap), and to bound the regret of a firm that has to learn the rents by auditing. We further show that audit intensity is a quantity to be tuned rather than maximized: welfare is single-peaked in it, and a firm and a social planner disagree on the mix of detection and penalty that delivers a given level of deterrence. Notably, two populations with identical costs and causal structure can game in one domain and improve in the other, a difference that a cost-only model cannot account for. Together, our findings highlight that whether a feature is ``gameable'' depends on the institution's verification policy as much as on the feature itself.

\end{abstract}

\section{Introduction}\label{sec:intro}
As machine learning systems are increasingly used to make decisions about people, in lending, hiring, and admissions, the people being classified have begun to respond to them strategically. The literature on strategic classification models such responses game-theoretically: the firm anticipates that agents will move their features and chooses a classifier that is robust to such movement \cite{hardt2016strategic,kleinberg2019induce,miller2020causal,jagadeesan2021microfoundations}. Consider, for instance, a lender who scores applicants on income, a credential, and a self-reported asset figure, where income can be inflated, the credential can be earned but not forged, and the asset figure is informative in the honest population but easy to misreport. When the score rule is the only instrument available to the lender, the prescription of existing models for a fakeable feature is to down-weight or discard it, at a cost borne by the honest applicants whom the feature would have served.

However, in practice, decision makers rarely rely on the score rule alone. Lenders verify income, admissions offices check documents, exam boards proctor, and tax authorities audit returns. A few recent works have added auditing to strategic classification, by choosing which agents to inspect under a budget \cite{estornell2021audits}, announcing audits and fines to steer agents toward genuine recourse \cite{estornell2023recourse}, or letting the firm verify agents' reports \cite{cezar2020adversarial}; see Section~\ref{sec:related}. Motivated by this, in this paper, we study the \emph{joint} design of a classifier and an audit policy. Our model differs from these works in what the firm designs: an \emph{audit profile} over \emph{feature channels}, which specifies which claims are checked, at what rate $p_j$, and with what penalty when caught. We model the audit per channel because a linear classifier creates a gaming rent on each fakeable feature, which the audit can price one at a time, and we keep detection and penalty separate because agents respond only to their product, so that the choice between them is a welfare question.

\textit{Paper overview and contributions.}
We begin by characterizing the agent's best response to a classifier and an audit profile (Section~\ref{sec:agent}). We show that under linear costs, the audit acts as a per-channel price: the agent games on channel $j$ if and only if the risk-adjusted price of doing so is below the price of honest recourse (Theorem~\ref{thm:portfolio} and Lemma~\ref{lem:deterrence-frontier}). We then turn to the firm's problem (Section~\ref{sec:firm}). Our main structural result is that the classifier affects the audit problem only through the distribution of gaming rents it induces (Theorem~\ref{thm:rentsuff}), so that the joint design decomposes into an outer choice of a score rule and an inner audit allocation problem. For the inner problem, we show that the optimal allocation across channels is a water-filling rule over the density of rents at the deterrence frontier (Theorem~\ref{thm:auditdesign}), and that even a uniform audit rate is a targeting instrument: against a population that is heterogeneous in what it loses when caught, the optimal uniform audit deters a quantile of the risk distribution and tolerates the rest (Theorem~\ref{thm:quantile-audit}). In other words, residual gaming at the optimum is not a sign that the audit has failed.

We next study the computational and statistical cost of the second instrument (Sections~\ref{sec:computation} and~\ref{sec:learning}). Audit allocation can be solved exactly when few channels are audited (Proposition~\ref{prop:vertex}), and the joint design is polynomial in fixed feature dimension (Theorem~\ref{thm:fixeddim}). Perhaps surprisingly, the difficulty does not come from the number of agents or from the interaction between the two instruments, but from the combination of two ingredients: agents who can game through overlapping sets of channels, and a hard cap on inspection. With both, the problem is as hard as Densest-$k$-Subgraph (Theorem~\ref{thm:hardness}); replacing the cap with a per-unit inspection cost makes it polynomial again (Proposition~\ref{prop:softbudget}). When the rents are unknown and the firm learns them by auditing, we show that the regret rate is determined by the shape of the rent distribution: $\tilde\Theta(\sqrt T)$ for smooth single-route rents and $\Theta(T^{2/3})$ once rents have atoms (Theorem~\ref{thm:learning}).

Finally, we ask what the second instrument is worth and how it should be used (Sections~\ref{sec:welfare} and~\ref{sec:risk-tuning}). We show that the welfare gain from auditing is bounded by the informational value of the fakeable features, and that this bound is attained, with an unbounded ratio, when a bounded signal leaves the classifier no room for Spence-style separation (Theorems~\ref{thm:dominance} and~\ref{thm:price}). Since audits have false positives, welfare is single-peaked in the audit intensity, and is maximized at the level that just deters gaming (Proposition~\ref{prop:overpen}). Moreover, while agents respond only to the product of detection and penalty, a firm that does not bear the harm of false accusations prefers severe penalties with little detection, whereas a social planner prefers the reverse (Proposition~\ref{prop:composition}). Finally, two populations with identical costs, classifier, and causal structure can game in one domain and improve in the other, a difference that no cost-only model with these primitives can produce, and that our model attributes to a single inequality on the deterrence product (Theorem~\ref{thm:domainsep}).

\textit{A key takeaway.}
A recurring theme in our results is that the audit is best understood as a price on the rents the classifier creates, rather than as a separate anti-fraud measure. This is what allows a predictive but fakeable feature to re-enter the optimal classifier once its rent can be priced, what makes the allocation problem separable across channels (and hard only when channels overlap), and why the right amount of verification is the least that prices a rent out rather than the most the firm can afford.

\textit{Summary of contributions.}
(i) We propose a model of strategic classification with a per-channel audit profile in which detection and penalty are separate instruments, and prove that the classifier enters the audit problem only through the rents it induces. (ii) We characterize the optimal audit allocation, identify when it is tractable and when it is NP-hard, and give regret rates, tight for one channel, for learning the audit from censored feedback. (iii) We characterize the welfare value of auditing, show that audit intensity should be tuned rather than maximized, and show that the firm and the planner disagree on the detection-penalty mix.

The remainder of the paper is organized as follows. Section~\ref{sec:related} reviews related work and Section~\ref{sec:model} presents the model. Sections~\ref{sec:agent}--\ref{sec:risk-tuning} present our results, and Section~\ref{sec:disc} concludes. Proofs that are only sketched in the main text, the assumption ledger, and the details of our numerical experiments are provided in Appendices~\ref{app:assumptions}--\ref{app:experiments}. A running example with three features, income (improvable and fakeable), a credential (improvable only), and self-reported assets (fakeable only), is used throughout (Example~\ref{ex:setup}).

\section{Related work}\label{sec:related}
\textit{Strategic classification.}
Our work is closely related to the literature on strategic classification, in which a firm chooses a classifier while agents move their features at a cost \cite{hardt2016strategic,ebrahimi2025double}. This literature has studied learning from revealed strategic responses \cite{dong2018revealed}, strategy-aware linear classifiers \cite{chen2020strategyaware,levanon2021practical}, the social cost of strategic rules \cite{milli2019social} and their disparate effects \cite{hu2019disparate}, and learning-theoretic aspects such as the order of play \cite{zrnic2021leads}, imperfect knowledge of the rule \cite{ghalme2021dark}, sample complexity \cite{sundaram2021pac}, and the geometry of the best response \cite{ahmadi2021perceptron}. Randomized rules have been shown to shrink the manipulation advantage \cite{braverman2020randomness}, and incentive-compatible rules to cope with strategically withheld features \cite{krishnaswamy2021withheld}. A causal strand of this literature distinguishes genuine improvement from gaming \cite{kleinberg2019induce,miller2020causal}, uses strategic responses as a source of causal information \cite{shavit2020causal,harris2022strategic,bechavod2021gaming}, and designs evaluation rules that reward investment over imitation \cite{haghtalab2020welfare,alon2020multiagent}. While our model of agents' responses is similar to those in these works, we differ in that in all of them, the classifier is the only deterrent available to the firm. We keep the causal distinction between improvement and gaming, and add an instrument that changes the payoff to faking rather than whether a feature is causal.

\textit{Audits in strategic classification.}
Three prior works are closest to ours in that they give the firm a verification instrument. \cite{estornell2021audits} let a principal audit a budgeted subset of \emph{agents}, where a detected misreport is corrected, and ask which agents to inspect so that the resulting scores are truthful. \cite{estornell2023recourse} announce an audit policy and a fine and ask how to steer agents from misreporting toward genuine recourse; they characterize optimal policies for utility-maximizing and recourse-maximizing principals, and show that subsidizing recourse can be worth part of the audit budget. \cite{cezar2020adversarial} study a firm that designs its classifier and may verify agents' information, with two agent types whose faking costs differ. That said, our design object differs from theirs in two ways, and these differences drive our results. First, our audit is a profile over \emph{feature channels}, each with its own rate; this is what allows us to ask how audit should be allocated across channels (Theorems~\ref{thm:auditdesign} and~\ref{thm:quantile-audit}), and what makes the hardness of Theorem~\ref{thm:hardness}, which arises only when agents can game through overlapping channels, possible to state. Second, detection and penalty are separate variables in our model, which is what makes the comparison between the firm and the planner in Proposition~\ref{prop:composition} meaningful; when an audit means ``inspect and correct,'' there is no penalty to choose. The coupling between the two instruments (Theorem~\ref{thm:rentsuff}) is also not studied in these works.

\textit{Verification and enforcement.}
More broadly, our work relates to inspection games, audit games, tax compliance, and costly state verification, which study how a principal deters hidden misrepresentation \cite{avenhaus2002inspection,blocki2013audit,townsend1979csv,reinganum1985compliance,mookherjee1989auditing,benporath2014verification}, to security games that allocate scarce inspection against a strategic adversary \cite{sinha2018security}, and to the economics of public enforcement, which characterizes the optimal mix of detection probability and sanction severity \cite{becker1968crime,polinsky2000enforcement}. In mechanism design, verification has been modeled as partial verification \cite{green1986verifiable}, probabilistic detection of lies \cite{caragiannis2012probabilistic,ball2019probabilistic}, selective verification of a few agents \cite{fotakis2016selective}, reporting costs \cite{kephart2016reporting}, ex post verification with bounded penalties \cite{mylovanov2017expost}, hard evidence \cite{hart2017evidence}, and adaptive costly audits in dynamic allocation \cite{dai2025audits}; strategic responses to the audit itself have also recently been considered \cite{burnat2026dpaudit,ceppi2019partial}. In all of these, the set of violations worth verifying is exogenous. In contrast, in our model it is created by the classifier: raising the weight on a fakeable feature adds predictive power in the honest population and, at the same time, opens a rent for would-be gamers. This feedback is the content of Theorem~\ref{thm:rentsuff}, and it is the reason why audit allocation across channels inherits the structure of Densest-$k$-Subgraph \cite{khot2006ptas,manurangsi2017dks}. Our detection-versus-penalty result can be viewed as the Becker-Polinsky-Shavell trade-off specialized to the false accusation of honest applicants flagged by a classifier. Finally, performative prediction \cite{perdomo2020performative} studies deployment-induced distribution shift without separating the choice of classifier from enforcement, and the documented harms of enforcement errors in automated eligibility systems \cite{eubanks2018automating} and the uneven incidence of audit risk across groups \cite{elzayn2025disparities} motivate our treatment of audit intensity as a quantity to be tuned rather than maximized.
\section{Model}\label{sec:model}
\textit{Agents and movement.}
A firm commits to a linear classifier $(w,\theta)$ and accepts an agent with observed features $x$ when $\inn{w}{x}\ge\theta$. An agent starts at true features $x^{0}$ with score deficit $g_{0}:=\theta-\inn{w}{x^{0}}$; if $g_0\le0$ the agent is already accepted and does not move. A rejected agent can use honest improvement $\delta\ge0$, which changes true features and is supported on $\Iset$, and gaming $\gamma\ge0$, which changes only observed features and is supported on $\Gset$; the two supports may overlap. Costs are separable, $C_{I}(\delta)=\sum_{j\in\Iset}c^{I}_{j}\phi_{j}(\delta_{j})$ and $C_{G}(\gamma)=\sum_{j\in\Gset}c^{G}_{j}\psi_{j}(\gamma_{j})$, with positive unit costs and convex nondecreasing $\phi_j,\psi_j$ vanishing at $0$. A move succeeds when $\inn{w}{\delta+\gamma}\ge g_{0}$.

\textit{The audit profile.}
The firm's second instrument is an audit profile $\pi=\big((p_j)_{j\in\Gset},u_{-}\big)$: $p_j\in[0,1]$ is the probability that a fake on coordinate $j$ is detected, and $u_-$ is the agent's utility when caught. Let $u_{+}>u_{0}=0\ge u_{-}$ be the accepted, rejected, and caught utilities, and let $\Delta:=u_{+}-u_{-}$ be the \emph{stakes-risk differential}: a caught fake forfeits the acceptance the agent was reaching for and, if $u_-<0$, pays a penalty on top. Both $p$ and $u_-$ are chosen by the firm. Severity is capped, $\Delta\le\Delta_{\max}$, by limited liability or statute; without a cap a firm could drive $p\to0$ and $\Delta\to\infty$ at fixed deterrence, and Section~\ref{sec:risk-tuning} shows that only the harm of false accusation breaks that tie. The profile is committed and observed before the agent acts, so detection enters as a known expected penalty rather than as the unknown inspection probability of an inspection game \cite{avenhaus2002inspection,blocki2013audit}. Commitment describes settings where verification policy is disclosed or institutionally fixed (published audit criteria, mandated proctoring, statutory credential checks), and it is also a benchmark for concealed audits (Section~\ref{sec:disc}).

Detection is \emph{support-based}: if the agent fakes on $S(\gamma)=\{j:\gamma_j>0\}$, detection occurs with probability $P(\gamma)=1-\prod_{j\in S(\gamma)}(1-p_j)$, with $P(0)=0$. Detection therefore depends on which coordinates are faked, not on how much; Proposition~\ref{prop:magnitude} in Appendix~\ref{app:agent} treats magnitude-dependent detection and shows what this choice buys. The agent's expected utility from a successful move is
\begin{equation}\label{eq:agent-eu}
EU(\delta,\gamma)=u_{+}-C_I(\delta)-C_G(\gamma)-P(\gamma)\Delta ,
\end{equation}
and the cost-only model is the boundary $p\equiv0$. Two quantities will play a central role in our analysis:
\begin{align}
\PhiI&:=\min\{C_I(\delta):\inn{w}{\delta}\ge g_0,\ \delta\ge0,\ \mathrm{supp}(\delta)\subseteq\Iset\},\label{eq:PhiI}\\
\PhiM(\pi)&:=\min\{C_I(\delta)+C_G(\gamma)+P(\gamma)\Delta:\nonumber\\
&\qquad\qquad\inn{w}{\delta+\gamma}\ge g_0,\ \delta,\gamma\ge0\}.\label{eq:PhiM}
\end{align}
$\PhiI$ is the price of honest recourse; $\PhiM(\pi)$ is the price of reaching the bar when faking is charged both its movement cost and its detection risk. The agent games exactly when a gaming-using minimizer of \eqref{eq:PhiM} beats both honest improvement and abstention.

\textit{The firm's objective.}
The firm is a Stackelberg leader choosing $(w,\theta,\pi)$ against a population of types $t$ with law $\rho$. Each agent ends in one of four outcomes: accepted honestly, rejected (abstained or caught), accepted with a successful fake, or, in the welfare sections, wrongly flagged while honest. Let $F(\mathrm{out}_t,y_t)$ score the outcome against the agent's true post-move label $y_t$ in accuracy units: a correct decision (a qualified agent accepted, an unqualified agent rejected or caught) scores $1$, and an incorrect one (a gamer accepted, a qualified agent rejected) scores $0$. Inspection is costly: running audit at rates $p$ costs $K(p)=\sum_{j\in\Gset}k_j(p_j)$ with each $k_j$ increasing, convex, and $k_j(0)=0$; several sections use the linear form $k_j(p_j)=\kappa_jp_j$ or a budget $\sum_j\kappa_jp_j\le B$. The firm maximizes
\begin{equation}\label{eq:firm-objective}
\Pio(w,\theta,\pi)\;=\;\int \E\big[F(\mathrm{out}_t,y_t)\,\big|\,t\big]\,d\rho(t)\;-\;K(p),
\end{equation}
where the expectation is over the audit coin and $\mathrm{out}_t$ is the agent's best response to $(w,\theta,\pi)$. The cost term is what rules out $p\equiv1$; the welfare analysis in Sections~\ref{sec:welfare} and~\ref{sec:risk-tuning} additionally charges a false-accusation harm $\kappa\varphi(\Delta)$ per honest agent wrongly flagged. We note two features of \eqref{eq:firm-objective} that our analysis relies on: the objective is additively separable across agents, and the audit enters each agent's contribution only through $\pi$. Section~\ref{sec:firm} studies two special cases of \eqref{eq:firm-objective}: a \emph{deterrence objective} that credits only agents who stop gaming, and an \emph{outcome objective} that also credits gamers who are caught.

\textit{Standing conditions.}
We will refer to the following four conditions throughout. \textnormal{(Lin)}: $\phi_j(t)=\psi_j(t)=t$. \textnormal{(SC)}: strictly convex differentiable costs with zero marginal cost at $0$. \textnormal{(W$>$0)}: every coordinate used for movement has $w_j>0$. \textnormal{(PG)}: a purely gameable coordinate $j\in\Pset:=\Gset\setminus\Iset$ has no causal effect on true qualification. Table~\ref{tab:conditions} in Appendix~\ref{app:assumptions} records where each is in force, and the ledger there states what changes when each is relaxed.

\begin{exmp}[Running instance]\label{ex:setup}
Feature $1$ is income, both improvable and fakeable; feature $2$ is a credential, improvable only; feature $3$ is self-reported assets, fakeable only. Thus $\Iset=\{1,2\}$, $\Gset=\{1,3\}$, $\Sset=\{1\}$, and $\Pset=\{3\}$. A low-audit education-like domain and a high-audit loan-like domain share costs, weights, and causal structure and differ only in $(p,u_{-})$.

Take \textnormal{(Lin)}, $g_0=1$, $w=(1,1,1)$, $u_+=1$, $u_-=-1$ (so $\Delta=2$), and costs $c^I_1=0.8$, $c^I_2=0.6$, $c^G_1=0.3$, $c^G_3=0.2$. Honest recourse costs $\PhiI=0.6$ by earning the credential. With $p_1=p_3=0.05$ the risk-adjusted gaming prices are $0.3+0.1=0.4$ on income and $0.2+0.1=0.3$ on assets, so the agent fakes the asset report. With $p_1=p_3=0.4$ the same prices are $1.1$ and $1.0$ while honest recourse still costs $0.6$, so the same agent improves.
\end{exmp}

\section{Agents' best responses}\label{sec:agent}
Fix a gaming support $S\subseteq\Gset$. Under support-based detection the probability of being caught, $P_S=1-\prod_{j\in S}(1-p_j)$, does not depend on how much is faked, so the audit is a support-level charge and the movement intensities solve an ordinary covering problem:
\begin{align}
\Phi(S)&:=\min\{C_I(\delta)+C_G(\gamma):\inn{w}{\delta+\gamma}\ge g_0,\ \mathrm{supp}(\gamma)\subseteq S\},\label{eq:PhiS}\\
\PhiM(\pi)&=\min_{S\subseteq\Gset}\big[\Phi(S)+P_S\Delta\big],\qquad\Phi(\varnothing)=\PhiI ,\label{eq:PhiM-decomp}
\end{align}
with $\delta,\gamma\ge0$ throughout.
The audit therefore determines which supports are worth using, while conditional on a support, the costs alone determine how much is moved. The following theorem formalizes this: audit acts on the extensive margin only, in that it switches gaming coordinates on or off, and under linear costs it does so through a single price per channel.

\begin{thm}[Risk-adjusted portfolio]\label{thm:portfolio}
Assume \textnormal{(W$>$0)} and separable differentiable convex costs. The agent's value is
\[
\Pi^\star=\max\{0,\ u_+-\PhiI,\ u_+-\PhiM(\pi)\},
\]
and the agent games iff the minimizing support in \eqref{eq:PhiM-decomp} is nonempty and $u_+-\PhiM(\pi)>\max\{0,u_+-\PhiI\}$. On any fixed optimal support the score constraint binds and every active coordinate satisfies $c^{I}_{j}\phi'_{j}(\delta_j^\star)=\lambda w_j$ and $c^{G}_{j}\psi'_{j}(\gamma_j^\star)=\lambda w_j$, with no audit term in the first-order conditions. Under \textnormal{(Lin)} there is an optimal gaming support that is empty or a singleton, and channel $j$ is ranked by its \emph{risk-adjusted price} $g_0\,c^{G}_{j}/w_j+p_j\Delta$.
\end{thm}
\begin{proof}[Proof sketch]
Any non-covering action is rejected, so abstention weakly dominates it; a covering action with gaming support $S$ has detection probability $P_S$ independent of the effort levels, hence value $u_+-P_S\Delta-\Phi(S)$, and maximizing over $S$ gives $u_+-\PhiM(\pi)$ by \eqref{eq:PhiM-decomp}. On a fixed support, $\Phi(S)$ is a convex covering program whose constraint binds, and the audit term is constant on $S$, so it does not enter the stationarity conditions. Under \textnormal{(Lin)}, $\Phi(S)$ is a linear program with one constraint, so an optimum loads one cheapest coordinate; deleting the unused coordinates of $S$ preserves the movement cost and weakly lowers detection, which yields a singleton support with price $g_0c^G_j/w_j+p_j\Delta$. The full proof is in Appendix~\ref{app:agent}.
\end{proof}

Intuitively, the audit adds $p_j\Delta$ to the cost-only price $g_0c^G_j/w_j$ of channel $j$, and the agent compares the resulting prices across channels and against the price of honest recourse. Note that detection and penalty enter the agent's problem only through their product $p_j\Delta$. This is a statement about behavior, not about welfare: the same expected penalty can be delivered by very different detection-penalty mixes, and we show in Section~\ref{sec:risk-tuning} that these mixes are not welfare-equivalent.

Under \textnormal{(Lin)}, let $j$ be the cheapest cost-only fake. The agent's fallback if it does not game is honest improvement when that is individually rational ($\PhiI\le u_+$) and abstention otherwise, so the value the fake must beat is $\min\{\PhiI,u_+\}$. Define the \emph{rent} of channel $j$ as
\begin{equation}\label{eq:rent}
D_j:=\min\{\PhiI,u_+\}-g_0c^G_j/w_j .
\end{equation}

\begin{lem}[Deterrence frontier]\label{lem:deterrence-frontier}
Under \textnormal{(Lin)} and \textnormal{(W$>$0)}, gaming on channel $j$ is deterred exactly when $p_j\Delta\ge D_j$. At the frontier one more unit of detection substitutes for $\Delta$ units of rent and one more unit of severity for $p_j$ units of rent.
\end{lem}
\begin{proof}
By Theorem~\ref{thm:portfolio}, the agent's non-gaming value is $u_+-\min\{\PhiI,u_+\}$ and its best value from gaming on $j$ is $u_+-g_0c^G_j/w_j-p_j\Delta$. Gaming on $j$ is (weakly) deterred exactly when the latter is at most the former, i.e., $g_0c^G_j/w_j+p_j\Delta\ge\min\{\PhiI,u_+\}$, which is $p_j\Delta\ge D_j$ by \eqref{eq:rent}. The substitution rates follow from differentiating $p_j\Delta=D_j$ along the frontier: $dD_j=\Delta\,dp_j+p_j\,d\Delta$.
\end{proof}

In other words, the rent is the saving from faking rather than earning one's way to the bar when honest recourse is individually rational (the standing case in Sections~\ref{sec:firm} and~\ref{sec:computation}), and it is the whole surplus from faking, $u_+-g_0c^G_j/w_j$, when it is not; in the latter case, deterrence pushes the agent out rather than toward improvement.

\textit{Shared coordinates.}
When a coordinate is both improvable and fakeable, strictly convex costs lead to an interior mix of improvement and faking with a closed-form fake-to-improve ratio (Appendix~\ref{app:agent}, Figure~\ref{fig:shared-coordinate-mixing}); the ``forge rather than earn'' split is strongest near linear costs and is diluted by curvature. We note that audit does not shrink this ratio gradually; rather, it eventually drops the fake coordinate from the support altogether.

\textit{The role of support-based detection.}
The singleton-support structure above relies on detection depending on \emph{which} coordinates are faked rather than on \emph{how much}. If instead detection depends on the magnitude of the fake, the gaming first-order conditions acquire an audit term, and a fake spread over several channels can strictly dominate a concentrated one (Proposition~\ref{prop:magnitude}, Appendix~\ref{app:agent}). That said, Proposition~\ref{prop:neardecouple} (Appendix~\ref{app:firm}) shows that the decomposition of Theorem~\ref{thm:rentsuff} below continues to hold approximately, with an error that vanishes as the magnitude sensitivity of detection vanishes.

\section{The firm's audit allocation}\label{sec:firm}
By Lemma~\ref{lem:deterrence-frontier}, a type with rent $D$ on its binding channel is deterred when $p\Delta\ge D$; audit allocation is therefore the problem of choosing how far to push this price frontier on each channel. In this section, we fix the classifier and study this problem under two objectives, both special cases of \eqref{eq:firm-objective}. The \emph{deterrence objective} credits the firm only for agents who stop gaming; this is the natural benchmark when a caught gamer is costly to process, or when the firm treats a caught fake as a cost rather than as a correct rejection. The \emph{outcome objective} additionally credits gamers who continue to game but are caught. We state our main characterizations under the deterrence objective, and then discuss how they change under the outcome objective.

\textit{Audit across channels.}
Channel $j$ has a continuum of would-be gamers whose binding channel is $j$ at the current design and whose rents have density $m_j$ on $[0,\bar D_j]$. A type with rent $d$ is deterred by audit intensity $p_j$ exactly when $d\le p_j\Delta$. Let $a_j(d)>0$ be the correct-allocation value of that type and $\ell_j(p_j)$ any reduced-form collateral loss from running the audit. Under the deterrence objective the value of auditing channel $j$ at intensity $p_j$ is
\begin{equation}\label{eq:cum-deterrence}
V_j(p_j)=\int_0^{\min\{p_j\Delta,\bar D_j\}} a_j(d)m_j(d)\,dd-\ell_j(p_j),
\end{equation}
for $0\le p_j\le \bar p_j:=\min\{\bar D_j/\Delta,1\}$, and with inspection cost $k_j$ the fixed-channel allocation problem is
\begin{equation}\label{eq:audit-allocation}
\max_{0\le p_j\le \bar p_j}\sum_{j\in\Gset}\big[V_j(p_j)-k_j(p_j)\big].
\end{equation}
Here, ``fixed-channel'' means that as $p_j$ increases with the other rates held fixed, the types assigned to channel $j$ either remain on $j$ until they are deterred or stop gaming; they do not switch to another channel over the audited range.

\begin{thm}[Marginal-value water-filling, fixed-channel regime]\label{thm:auditdesign}
Suppose $d\mapsto a_j(d)m_j(d)$ is nonincreasing and $\ell_j+k_j$ is convex for every $j$. Then \eqref{eq:audit-allocation} is a concave program, and every interior active channel satisfies
\begin{equation}\label{eq:waterfilling-foc}
\Delta\,a_j(p_j^\star\Delta)m_j(p_j^\star\Delta)
=\ell_j'(p_j^\star)+k_j'(p_j^\star).
\end{equation}
With a shared budget $\sum_jk_j(p_j)\le B_a$ in place of the cost term, interior active channels equalize net marginal deterrence value per marginal dollar (Appendix~\ref{app:firm}, \eqref{eq:waterfilling-budget}). The optimal value is unique. The optimal allocation on channel $j$ is unique whenever $a_jm_j$ is strictly decreasing or $\ell_j+k_j$ is strictly convex on that channel; if this holds on every active channel, the optimal allocation is unique.
\end{thm}
\begin{proof}[Proof sketch]
For $0<p_j<\bar p_j$, \eqref{eq:cum-deterrence} has derivative $V_j'(p_j)=\Delta\,a_j(p_j\Delta)m_j(p_j\Delta)-\ell_j'(p_j)$, which is nonincreasing by the assumptions, so each summand $V_j-k_j$ is concave and the KKT conditions are necessary and sufficient; \eqref{eq:waterfilling-foc} is stationarity, and channels with $V_j'(0)\le k_j'(0)$ are dark. The program separates across channels, which gives the uniqueness statements. Under a shared budget, the Lagrangian $\sum_jV_j(p_j)-\lambda\sum_jk_j(p_j)$ equalizes $V_j'/k_j'$ across interior channels (Appendix~\ref{app:firm}).
\end{proof}

Intuitively, the theorem states that the firm should audit where the next increase in $p_j$ crosses the densest and most valuable slice of rents, which need not be the channel with the largest possible manipulation. Uniqueness is decided channel by channel because the program separates across channels: strictness on one channel pins down that channel's rate, but says nothing about a flat channel elsewhere. The no-switching condition is the price of a global statement; without it, the same equations hold locally, with the rents recomputed after each marginal update, and as we show in Section~\ref{sec:computation}, the exact allocation becomes NP-hard once gaming technologies overlap.

Under the outcome objective, an undeterred type of rent $d$ on channel $j$ is caught with probability $p_j$ and then counts as a correct rejection, so the channel's value becomes
\begin{equation}\label{eq:cum-outcome}
V^{\mathrm{out}}_j(p_j)=\int_0^{p_j\Delta} a_jm_j\,dd\;+\;p_j\int_{p_j\Delta}^{\bar D_j} a_jm_j\,dd\;-\;\ell_j(p_j),
\end{equation}
again for $0\le p_j\le\bar p_j$.

\begin{cor}[Outcome objective]\label{cor:outcome}
Under the hypotheses of Theorem~\ref{thm:auditdesign}, replacing $V_j$ by $V^{\mathrm{out}}_j$ in \eqref{eq:audit-allocation} preserves concavity, and every interior active channel satisfies
\begin{equation}\label{eq:outcome-foc}
(1-p_j^\star)\,\Delta\,a_j(p_j^\star\Delta)m_j(p_j^\star\Delta)\;+\int_{p_j^\star\Delta}^{\bar D_j} a_jm_j\,dd\;=\;\ell_j'(p_j^\star)+k_j'(p_j^\star).
\end{equation}
If channel $j$ has an interior optimum $p_j^{\mathrm{det}}$ under the deterrence objective, then its outcome-objective optimum satisfies $p_j^{\mathrm{out}}\ge p_j^{\mathrm{det}}$ if and only if
\begin{equation}\label{eq:outcome-direction}
\int_{p_j^{\mathrm{det}}\Delta}^{\bar D_j} a_jm_j\,dd\ \ \ge\ \ p_j^{\mathrm{det}}\Delta\,a_j(p_j^{\mathrm{det}}\Delta)m_j(p_j^{\mathrm{det}}\Delta).
\end{equation}
For the uniform audit of Theorem~\ref{thm:quantile-audit} below, with $U^{\mathrm{out}}(p)=v[1-G(D/p)]+v\,p\,G(D/p)-c_ap$, the interior first-order condition is $v\big[(1-p)h(D/p)D/p^2+G(D/p)\big]=c_a$, and if $U^{\mathrm{out}}$ is quasi-concave the outcome optimum exceeds the deterrence optimum $p^\star$ iff $G(D/p^\star)\ge c_ap^\star/v$.
\end{cor}

Intuitively, crediting caught gamers changes the marginal value in two directions at once. The undeterred mass beyond the frontier now earns $a_j$ per unit of extra detection, which raises the marginal value; but the marginal type at the frontier was already being caught with probability $p_j$, so deterring it fully is worth only $(1-p_j)$ of what it was worth before, which lowers it. Condition \eqref{eq:outcome-direction} says which effect wins. Neither direction is generic: with a linearly decreasing rent density, the outcome optimum lies above the deterrence optimum when the latter is small and below it when the latter is large. The proof is in Appendix~\ref{app:firm}.

\textit{Heterogeneous risk and a uniform audit.}
Agents differ in how costly being caught is to them. Let $\Delta_t=u_+-u_{-,t}$ be drawn from a CDF $G$ with density $h$ on $[\Delta_{\min},\Delta_{\max}]$, independently of the binding rent $D$ (if risk tolerance is correlated with recourse costs, the deterred set is no longer an upper tail; we keep the independent case for its closed form). Suppose the firm can only use a single uniform audit rate $p$. Then, type $t$ is deterred iff $\Delta_t\ge D/p$; that is, a uniform audit deters the upper tail of the risk distribution. With value $v$ per deterred gamer and audit cost $c_ap$, the deterrence objective is
\begin{equation}\label{eq:Uquantile}
U(p)=v[1-G(D/p)]-c_ap ,
\end{equation}
and deterrence changes continuously only on the open interval $\mathcal P^\circ$ with endpoints $D/\Delta_{\max}$ and $\min\{1,D/\Delta_{\min}\}$. The following theorem shows that even this undiscriminating instrument targets, in that it optimally stops at a quantile of the risk distribution.

\begin{thm}[The optimal uniform audit is a quantile]\label{thm:quantile-audit}
Assume that on $\mathcal P^\circ$ the marginal value $p\mapsto v\,h(D/p)\,D/p^2$ crosses $c_a$ once from above. If the maximizer of \eqref{eq:Uquantile} is interior it is unique and satisfies $v\,h(D/p^\star)\,D/(p^\star)^2=c_a$; for the marginal deterred type $\Delta_m=D/p^\star$ this reads $h(\Delta_m)\Delta_m^2=c_aD/v$. The firm deters exactly the tail $\{\Delta_t\ge\Delta_m\}$ and tolerates the rest. If no interior root exists the optimum is a boundary: no audit, full feasible audit, or full deterrence. If the firm optimizes against an estimate $\hat G$, its regret under the true $G$ is at most $2v\|G-\hat G\|_\infty$.
\end{thm}
\begin{proof}[Proof sketch]
$U'(p)=v\,h(D/p)D/p^2-c_a$ crosses zero once from above on $\mathcal P^\circ$ by assumption, which gives uniqueness and the first-order condition; substituting $\Delta_m=D/p^\star$ gives the marginal-type form. For the regret bound, only $v[1-G(D/p)]$ depends on $G$, so $|U_G(p)-U_{\hat G}(p)|\le v\|G-\hat G\|_\infty$ for every $p$, and a standard two-sided comparison at $p^\star$ and $\hat p$ gives the factor $2$ (Appendix~\ref{app:firm}).
\end{proof}

Intuitively, one might expect the firm to audit until gaming is eliminated. The theorem shows that full deterrence is optimal only in the limit of costless audits: low-$\Delta$ agents lose less when caught, so that a given detection rate deters them last, and reaching the lower tail of the risk distribution may simply be too expensive with a uniform instrument. In other words, residual gaming at the optimum is not evidence that the audit technology has failed. We also note that the regret bound makes the rule operational, as it suffices to estimate the risk distribution in Kolmogorov distance. The full-deterrence limit also marks where the classifier stops carrying the whole anti-gaming burden: if the firm can set $p_j\ge D_j/\Delta_{\min}$ on the channels that would be gamed against a candidate classifier, no type games against it, and its fakeable predictive features can be priced as prediction features again. 

\textit{How the two instruments interact.}
We now turn to the interaction between the classifier and the audit. Write each strategic type's score deficit as $g_{0,t}=\theta-\inn{w}{x^0_t}$, its honest price as $\PhiI(t)$, its channel rents as $D_{j,t}:=\PhiI(t)-g_{0,t}c^G_j/w_j$, and its post-improvement label as $y^I_t$ (whether improving to the bar yields genuine qualification). We call the law of $\big(\PhiI(t),(D_{j,t})_{j\in\Gset},\Delta_t,y_t,y^I_t\big)$ under $\rho$ the \emph{rent profile} $\nu_{w,\theta}$.

\begin{thm}[Rent-sufficiency and two-stage design]\label{thm:rentsuff}
Assume \textnormal{(Lin)}, \textnormal{(W$>$0)}, the additively separable objective \eqref{eq:firm-objective}, and a non-degenerate population: the law of $\big(\PhiI(t),(D_{j,t})_j,\Delta_t\big)$ charges no agent-indifference hyperplane (any law absolutely continuous in those coordinates qualifies). Then the firm value depends on the classifier only through the rent profile, $\Pio(w,\theta,\pi)=\Psi(\nu_{w,\theta},\pi)$; two classifiers inducing the same $\nu$ are payoff-equivalent for every audit profile, and
\begin{equation}\label{eq:two-stage}
\max_{w,\theta,\pi}\Pio
=\max_{w,\theta}\ \underbrace{\max_{\pi}\Psi(\nu_{w,\theta},\pi)}_{\text{audit allocation}} .
\end{equation}
The two stages decouple exactly only along rent-invariant classifier directions, which are non-generic since $g_{0,t}$ moves with every weight and with $\theta$.
\end{thm}
\begin{proof}[Proof sketch]
Under \textnormal{(Lin)}, Theorem~\ref{thm:portfolio} expresses the agent's three branch values as $0$, $u_+-\PhiI$, and $(u_+-\PhiI)+\max_j(D_j-p_j\Delta)$, which depend on the type only through its rent coordinates $\xi_t$; the maximizing branch and channel are therefore measurable functions of $(\xi_t,\pi)$. The set of rent coordinates at which two branches or two channels tie is a finite union of hyperplanes, which is $\rho$-null by non-degeneracy, so the tie-breaking rule is immaterial. The firm's expected per-type payoff is then a bounded measurable function $f(\xi_t,\pi)$ (only the audit coin on the played channel is random), and additive separability gives $\Pio=\int f(\xi,\pi)\,d\nu_{w,\theta}(\xi)=:\Psi(\nu_{w,\theta},\pi)$. Full details are in Appendix~\ref{app:firm}.
\end{proof}

In words, the theorem states that the classifier matters for the audit problem only through the gaming rents it creates. The inner problem in \eqref{eq:two-stage} is the allocation problem studied above (water-filling under the rent-density representation of $\nu$, and the quantile rule under its risk-CDF representation). The two designs do not separate, but their coupling is funneled through a single object; this is what allows the choice of classifier to be viewed as an outer optimization over induced rent profiles, and it is what Section~\ref{sec:computation} makes algorithmic. We note that the theorem is stated under support-based detection. Proposition~\ref{prop:neardecouple} (Appendix~\ref{app:firm}) shows that under a detection technology with magnitude sensitivity $1/\beta$, the firm's value converges to $\Psi(\nu_{w,\theta},\pi)$ as $\beta\to\infty$; the mechanism is a detection floor (Lemma~\ref{lem:floor}): at a fixed gameable dimension, a fake cannot be diluted into undetectability, because covering the score gap forces some channel to a magnitude at which sharp audit catches it.

\section{Computational complexity of audit allocation}\label{sec:computation}
In this section, we study the computational complexity of audit allocation and of the joint design. The picture that emerges mirrors the economics, with the rent profile of Theorem~\ref{thm:rentsuff} as the interface between the two: audit allocation can be solved exactly when few channels are audited, is Densest-$k$-Subgraph-hard once agents can game through many channels with heterogeneous gaming technologies, and the full joint design is polynomial in fixed feature dimension.

\textit{The finite-type allocation problem.}
Fix a classifier and finitely many strategic types $t=1,\dots,N$ under \textnormal{(Lin)} with honest participation $\PhiI(t)\le u_+$. Type $t$ has deterrence value $a_t>0$, differential $\Delta_t$, gameable support $\Gset_t\subseteq\Gset$ with unit costs $c^G_{j,t}$, and channel rents $D_{j,t}=\PhiI(t)-g_{0,t}c^G_{j,t}/w_j$; write $J_t:=\{j\in\Gset_t:D_{j,t}>0\}$ for its profitable channels. By Theorem~\ref{thm:portfolio}, type $t$ is deterred iff $p_j\Delta_t\ge D_{j,t}$ for \emph{every} $j\in J_t$: deterrence is a covering condition with one inequality per profitable channel. The problem \textsc{Audit} is to choose $p\in[0,1]^{\Gset}$ with $\sum_j\kappa_jp_j\le B$ maximizing the deterrence value $\sum_{t\,\text{deterred}}a_t$; it is the finite-type form of \eqref{eq:cum-deterrence}.

\begin{prop}[Vertex rule: exact allocation for few channels]\label{prop:vertex}
Let $m=|\Gset|$ and let $\mathcal H$ be the arrangement in $[0,1]^m$ generated by the deterrence hyperplanes $p_j\Delta_t=D_{j,t}$, the switch hyperplanes $p_j\Delta_t-p_i\Delta_t=g_{0,t}\big(c^G_{i,t}/w_i-c^G_{j,t}/w_j\big)$, the budget hyperplane $\sum_j\kappa_jp_j=B$, and the box faces. Under leader-favorable tie-breaking an optimal solution of \textsc{Audit} lies at a vertex of $\mathcal H$; the same holds for the outcome objective in which an undeterred gamer on channel $j$ contributes $-(1-p_j)a_t$. Enumerating vertices solves either problem exactly in time $(Nm)^{O(m)}$, polynomial in $N$ for fixed $m$.
\end{prop}

The proof (Appendix~\ref{app:computation}) shows that the set of profiles at which a given action pattern is optimal is a polyhedron with facets on $\mathcal H$, on which the objective is affine. We note that the optimum can lie at an intersection at which a type is indifferent between two gaming channels, a point that per-channel deterrence thresholds alone would miss. The next result shows that the exponential dependence on $m$ cannot be removed in general.

\begin{thm}[Hardness of audit allocation across many channels]\label{thm:hardness}
\textsc{Audit} is NP-hard, by a reduction from the Densest $k$-Subgraph problem: given a graph $G=(V,E)$ and an integer $k$, build one channel per vertex with $w_j=1$ and $\kappa_j=1$, and one type $t_e$ per edge $e=\{u,v\}$ with gameable support $\{u,v\}$, unit gaming costs $c^G_{u,t_e}=c^G_{v,t_e}=1$, score deficit $g_{0,t_e}=1$, one improvable coordinate of unit weight and cost $2$ (so $\PhiI(t_e)=2$), acceptance utility $u_+=2$, differential $\Delta_{t_e}=2$, value $a_{t_e}=1$, and budget $B=k/2$. In this instance every edge-type has rent $1$ on both of its channels and is deterred iff $p_u\ge\tfrac12$ and $p_v\ge\tfrac12$, and the optimal value of \textsc{Audit} equals the maximum number of edges induced by a $k$-vertex subset of $G$. Hardness therefore holds already with unit values, a common differential, common rents, linear unit-cost audits, and two profitable channels per type.
\end{thm}
\begin{proof}
Let $(G,k,M)$ ask whether some $S\subseteq V$ with $|S|=k$ induces at least $M$ edges, and build the instance in the statement. By Theorem~\ref{thm:portfolio}, type $t_e$ games iff some channel $j\in\{u,v\}$ has $g_{0,t_e}c^G_{j,t_e}/w_j+p_j\Delta_{t_e}<\PhiI(t_e)$, i.e., iff $\min\{p_u,p_v\}<\tfrac12$. If $S$ induces $M'$ edges, the profile $p=\tfrac12\indic_S$ is feasible and deters exactly the edge-types inside $S$, so its value is $M'$. Conversely, given a feasible $p$ with value $M'$, let $S'=\{j:p_j\ge\tfrac12\}$; feasibility gives $|S'|\le 2B=k$, every deterred type is an edge inside $S'$, and padding $S'$ to size $k$ can only add induced edges. Hence the two optima are equal, and the decision problem has answer yes iff the optimum of \textsc{Audit} is at least $M$. Since Densest $k$-Subgraph contains Clique ($M=\binom k2$), \textsc{Audit} is NP-hard.
\end{proof}

\begin{cor}[Inapproximability]\label{cor:inapprox}
Since the reduction preserves objective values exactly, any polynomial-time $\rho$-approximation for \textsc{Audit} yields a $\rho$-approximation for Densest $k$-Subgraph. Hence \textsc{Audit} admits no PTAS unless NP has subexponential randomized algorithms \cite{khot2006ptas}, and, under the Exponential Time Hypothesis, admits no polynomial-time approximation within a factor $m^{1/(\log\log m)^{c}}$ in the number of channels $m$ \cite{manurangsi2017dks}. The reduction requires heterogeneous gameable supports: if all types share unit costs and weights, each profitable set $J_t$ is a prefix of the channels ordered by $c^G_j/w_j$, any two such sets are nested, and the reduction's instances cannot be realized.
\end{cor}

Intuitively, the theorem identifies where the difficulty comes from. With a single binding channel per type, allocation is the concave water-filling problem of Theorem~\ref{thm:auditdesign}, and with few channels the vertex rule is exact. What is hard is the regime in which overlapping gaming technologies force the auditor to choose which \emph{combinations} of channels to seal under a cap, since deterring a type pays off only if all of its escape routes are priced out at once. Whether the nested-prefix instances arising under a common gaming technology are polynomial-time solvable remains open; that problem reduces to selecting a maximum-value set of decreasing lines under a budget on the integral of their upper envelope.

The hardness above is for the deterrence objective. Under the outcome objective, a fractional audit earns partial credit for catching gamers it does not deter, and the reduction survives only when that credit is small.

\begin{prop}[Hardness with catch credit]\label{prop:catch-credit}
Let \textsc{Audit}$_\lambda$ credit an undeterred type $t$ gaming on channel $j(t)$ with $\lambda\,p_{j(t)}a_t$, $\lambda\in[0,1]$, so that $\lambda=0$ is \textsc{Audit} and $\lambda=1$ is the outcome objective. If $\lambda\le1/(2\sum_ta_t)$, the instances of Theorem~\ref{thm:hardness} still have deterrence value equal to the Densest $k$-Subgraph optimum at every optimal solution of \textsc{Audit}$_\lambda$, so \textsc{Audit}$_\lambda$ is NP-hard, and a $\rho$-approximation for it yields a $2\rho$-approximation for Densest $k$-Subgraph. For $\lambda$ bounded away from zero the optimizer can change (on small instances the outcome-optimal profile deters strictly fewer types than the Densest $k$-Subgraph optimum), and the complexity of \textsc{Audit}$_1$ is open.
\end{prop}
The proof is in Appendix~\ref{app:computation}.

\textit{Where the hardness comes from.}
The reduction above uses a hard cap on total inspection. Perhaps surprisingly, for the deterrence objective, replacing the cap by a per-unit cost removes the hardness altogether.

\begin{prop}[A soft budget makes allocation polynomial]\label{prop:softbudget}
Let \textsc{Audit}$^{\mathrm{cost}}$ be the problem of Theorem~\ref{thm:hardness} with the budget constraint replaced by a linear inspection cost, i.e., the problem of maximizing $\sum_{t\,\text{deterred}}a_t-\sum_j\kappa_jp_j$ over $p\in[0,1]^{\Gset}$. Then an optimal solution lies on the finite lattice of deterrence thresholds $L=\prod_j\big(\{0\}\cup\{D_{j,t}/\Delta_t\}_t\big)$, the objective is supermodular on $L$, and \textsc{Audit}$^{\mathrm{cost}}$ is solvable in time polynomial in $N$ and $m$ by submodular function minimization over a ring family with $O(Nm)$ elements. The same is false for the outcome objective, whose payoff is not supermodular on $L$ already with three channels.
\end{prop}
\begin{proof}
Let $\theta_{j,t}:=D_{j,t}/\Delta_t$ and $L_j:=\{0\}\cup\{\theta_{j,t}\le1\}$. For any $p$, rounding each $p_j$ down to $\underline p_j:=\max\{\ell\in L_j:\ell\le p_j\}$ preserves every deterrence inequality $p_j\ge\theta_{j,t}$ that held at $p$ and weakly lowers the cost, so an optimum lies on $L=\prod_jL_j$. Order $L$ coordinatewise and let $\chi_t(p)$ indicate that $t$ is deterred at $p$; $\chi_t$ is the indicator of an up-set of $L$, and a direct case check gives $\chi_t(p\vee p')+\chi_t(p\wedge p')\ge\chi_t(p)+\chi_t(p')$ (if both sides are deterred so are the join and the meet; if one is, so is the join). Hence $\sum_ta_t\chi_t$ is supermodular and the objective $F(p)=\sum_ta_t\chi_t(p)-\kappa\cdot p$ is supermodular on $L$. Encoding $p$ by the set $\{(j,\ell):\ell\in L_j,\ell\le p_j\}$ maps $L$ isomorphically onto a ring family over a ground set of size at most $Nm$, on which $-F$ is submodular; minimizing a submodular function over a ring family is strongly polynomial in the ground set size and the value-oracle cost \cite{schrijver2000sfm,iwata2001sfm,grotschel1988geometric}, and $F$ is evaluated in $O(Nm)$ time. The outcome-objective claim is shown by an instance in Appendix~\ref{app:computation}.
\end{proof}

The reason for this difference is that deterring a type is a covering condition, and covering conditions are complements: auditing channel $u$ is worth more once channel $v$ is also audited, because only then does a type with routes $\{u,v\}$ stop gaming. This complementarity is exactly supermodularity, and supermodular maximization over a lattice is tractable. What Densest $k$-Subgraph exploits is the cap, whose feasible set is not closed under coordinatewise maxima; indeed, on the reduction's own instances, the cost version reads $\max_S e(S)-\lambda|S|$, which is a minimum-cut problem. In other words, the hardness of Theorem~\ref{thm:hardness} is a statement about \emph{capacity-constrained} verification with overlapping routes, and not about overlapping routes alone.

\textit{Joint design.}
By Theorem~\ref{thm:rentsuff}, the classifier enters the audit problem through the rent profile. For finitely many types, this interface is finite: as $(w,\theta,p)$ vary, the population's action pattern (who abstains, who improves, and who games on which channel) changes only when the sign of one of polynomially many polynomial inequalities changes, and within a fixed pattern, the firm's value does not depend on $(w,\theta)$ at all.

\begin{thm}[Joint Stackelberg design in fixed dimension]\label{thm:fixeddim}
Under \textnormal{(Lin)} with linear audit costs, $N$ types in feature dimension $d$ (so $|\Gset|\le d$), and post-improvement qualification given by a bounded-degree semialgebraic rule, an exactly optimal joint design $(w^\star,\theta^\star,\pi^\star)$ over a box of positive weights is computable in time $(Nd)^{O(d)}$, polynomial in $N$ for fixed $d$. The algorithm enumerates the realizable sign conditions of a family of $O(Nd^2)$ bounded-degree polynomials in the $d+1+|\Gset|$ design variables \cite{basu2006algorithms}; on each cell the action pattern is constant, the objective is affine in $p$ and constant in $(w,\theta)$, and the cell optimum is a semialgebraic optimization in fixed dimension.
\end{thm}
\begin{proof}[Proof sketch]
Each type's action is determined by the signs of $O(d^2)$ polynomials of degree at most $3$ in $z=(w,\theta,p)$ (strategic or inert, cheapest improvable coordinate, participation, gaming versus improvement, and channel switches). By the Basu-Pollack-Roy bound, this family of $O(Nd^2)$ polynomials realizes $(Nd)^{O(d)}$ sign conditions, which can be listed in that time. On each cell, the action pattern is fixed, the objective is affine in $p$ and independent of $(w,\theta)$, and the cell optimum is a linear objective over a semialgebraic set in fixed dimension, solvable exactly in $(Nd)^{O(d)}$ time (Appendix~\ref{app:computation}).
\end{proof}

This result makes the decomposition \eqref{eq:two-stage} operational: the outer maximization ranges over $(Nd)^{O(d)}$ combinatorial rent patterns, with the allocation of Proposition~\ref{prop:vertex} inside. We note that the algorithm is an enumeration certificate rather than a practical procedure. Together, the results of this section show that auditing is computationally cheap in the regime in which strategic classification is usually studied (few features and few verifiable channels), and that its difficulty grows with the heterogeneity of gaming technologies rather than with the size of the population.

\section{Learning to audit}\label{sec:learning}
The firm in the previous sections knows the rent profile. By Theorem~\ref{thm:rentsuff}, this is all it needs to know, which raises the question of how costly this knowledge is to acquire when the only source of information is the audit itself. In this section, we answer this question for the audit allocation problem: the classifier is fixed, the type distribution is unknown, and the firm learns by inspecting.

\textit{Protocol.}
There are $T$ rounds. In round $\tau$ the firm commits to $p_\tau\in[0,1]^m$ and pays $\kappa\cdot p_\tau$ with $\kappa_j\le1$; a type $t_\tau$ is drawn independently from an unknown law $q$ over rent coordinates with $a_t\in[0,1]$; the agent best-responds to $p_\tau$ as in Theorem~\ref{thm:portfolio}; and the firm observes its realized payoff for that round, $r_\tau=a_{t_\tau}\indic\{t_\tau\text{ deterred at }p_\tau\}-\kappa\cdot p_\tau$, and nothing else. This is the deterrence objective under bandit feedback: the firm never observes what a deterred agent would have done, and learns about a profile only by playing it. Regret is measured against the best fixed profile, $\mathrm{Reg}_T:=T\max_pV(p)-\E\sum_\tau V(p_\tau)$ with $V(p):=\E_q[a_t\indic\{t\text{ deterred at }p\}]-\kappa\cdot p$.

The closest model to ours is the censored semi-bandit of \cite{verma2019censored}, in which each arm has a single unknown threshold above which its loss vanishes; that problem is equivalent to a multiple-play bandit and admits logarithmic regret. In our setting, a gaming population is a distribution of thresholds on every channel, and a type with overlapping routes is deterred only when every one of its channels clears its threshold. As we show next, both features change the rates.

\begin{thm}[Learning to audit]\label{thm:learning}
\textup{(i)} \emph{Upper bound.} For any $q$ and $T\ge m^{m+2}$, the exponential-weights bandit algorithm over the grid $\{w,2w,\dots,1\}^m$ with $w=T^{-1/(m+2)}$, rounding every profile up to the grid, achieves $\mathrm{Reg}_T=O\big(m\,T^{(m+1)/(m+2)}\sqrt{\log T}\big)$, with no dependence on the number of types or on the form of $q$.

\textup{(ii)} \emph{Lower bound.} For $m=1$ and every $T\ge2^{12}$ there is a population of $\lceil T^{1/3}\rceil$ types, with values and cost at most $1$, on which every algorithm has $\mathrm{Reg}_T\ge c\,T^{2/3}$ for an absolute constant $c>0$. The exponent in (i) is tight for one channel.

\textup{(iii)} \emph{Smooth single-route rents.} If instead each channel carries a continuum of types whose binding channel is fixed and whose rent density $a_jm_j$ is nonincreasing (the regime of Theorem~\ref{thm:auditdesign}), and the firm observes its realized payoff channel by channel, then the per-channel objective is concave and one-dimensional bandit convex optimization gives $\mathrm{Reg}_T=\tilde O(m\sqrt T)$; $\Omega(\sqrt T)$ is unavoidable even for $m=1$.

\textup{(iv)} \emph{Overlapping routes.} Concavity in (iii) is a single-route property: with two channels and types whose rents on both channels are independent and uniform on $[0,\Delta]$, the expected deterrence payoff is $p_1p_2-\kappa(p_1+p_2)$, which is not concave.
\end{thm}
\begin{proof}[Proof sketch]
(i) Deterrence is monotone in $p$, so rounding every profile up to the grid $G_w$ loses at most $mw$ per round in cost and nothing in deterrence; running Exp3 over the $|G_w|\le eT^{m/(m+2)}$ grid points with payoffs rescaled to $[0,1]$ gives regret $O(m\sqrt{T|G_w|\log|G_w|})$ against the best grid point, and $w=T^{-1/(m+2)}$ balances the two terms. (ii) Take $N=\lceil T^{1/3}\rceil$ equally likely types with rents $j/N$ and $\kappa=1$, so that $V(p)=\lfloor Np\rfloor/N-p\le0$ with equality exactly at the corners $j/N$; moving mass $\varepsilon=\Theta(\sqrt{N/T})$ from type $i+1$ to type $i$ makes corner $i$ the unique optimum by $\varepsilon$ and changes the observed Bernoulli at a single corner. Playing a non-corner point is dominated by the corner to its left, so the problem is an $N$-armed bandit with one $\varepsilon$-better arm, and the standard change-of-measure argument gives regret $\Omega(\sqrt{NT})=\Omega(T^{2/3})$. (iii) By Theorem~\ref{thm:auditdesign}, each channel's objective is concave and Lipschitz, so one-dimensional bandit convex optimization \cite{agarwal2011bandit} applies channel by channel; the lower bound is that of \cite{shamir2013complexity} for a concave quadratic with unknown shift. (iv) The deterrence probability is $p_1p_2$, whose Hessian has eigenvalues $\pm1$. Details are in Appendix~\ref{app:learning}.
\end{proof}

Intuitively, the theorem states that the shape of the rent distribution determines the learning rate. A single threshold per channel gives logarithmic regret \cite{verma2019censored}; a smooth, decreasing spread of thresholds gives $\sqrt T$, since the concavity of Theorem~\ref{thm:auditdesign} turns learning into one-dimensional convex bandit optimization; and atoms force $T^{2/3}$ already on a single channel, since each atom is a corner that the learner has to locate. Overlapping routes remove the concavity for the same reason that they create the hardness in Theorem~\ref{thm:hardness}: deterrence is a covering condition, and a covering condition is a product rather than a sum. We note that the upper bound in (i) is exponential in $m$, and its grid has $T^{m/(m+2)}$ points; since the offline cost problem is polynomial (Proposition~\ref{prop:softbudget}), the obstacle to a $\mathrm{poly}(m)$-time learner is exploration over an exponential grid rather than offline hardness, and whether the supermodular structure can be exploited online remains open. Finally, if the same agents recur across rounds rather than being drawn afresh, a long-lived population could under-game to mislead the learner; the reduction of \cite{haghtalab2022nonmyopic} from non-myopic to myopic agents through delayed feedback applies to (i) without change.

\section{Welfare analysis and domain separation}\label{sec:welfare}
We now ask what the audit instrument is worth. Let $\Fcost$ be the set of $(\mathrm{FP},\mathrm{TP})$ outcomes reachable by varying $(w,\theta)$ with $p\equiv0$, and $\Ftwo$ its analogue when the firm can also audit.

Containment $\Fcost\subseteq\Ftwo$ is trivial; the question is when it is strict, and what behavior creates the gap.

\begin{thm}[Perfect-audit benchmark]\label{thm:dominance}
Assume \textnormal{(Lin)} and perfect audits in the sense of no false positives. Fix a cost-only Pareto point at which a positive mass of strategic agents games on a binding channel. Let their rents be uniformly spread with density $\mu_j$ on $[0,D_{\max}]$, so that the total gaming mass is $\mu_g=\mu_jD_{\max}$, and suppose that $D_{\max}\le\Delta$. If each deterred gamer is worth one unit of correct-allocation welfare, auditing that channel up to $p^\star=D_{\max}/\Delta$ deters the entire gaming mass and raises welfare by exactly
\begin{equation}\label{eq:welfare-lift}
W'-W=\int_0^{p^\star}\mu_j\Delta\,dp
=\mu_g>0 .
\end{equation}
Moreover, under the regularity conditions \textnormal{(R1)}--\textnormal{(R4)} of Appendix~\ref{app:welfare} (local exposedness of the cost-only boundary and local openness of the design-to-rate map), $\Ftwo\setminus\Fcost$ contains a positive-measure region in the $(\mathrm{FP},\mathrm{TP})$ plane.
\end{thm}
\begin{proof}[Proof sketch]
Raising $p$ by $dp$ moves the deterrence frontier by $\Delta\,dp$ in rent space and flips mass $\mu_j\Delta\,dp$, each unit of which is worth one; integrating to $p^\star$ gives $\mu_g$. The positive-measure statement follows from local exposedness of $\Fcost$ and local openness of the design-to-rate map (Appendix~\ref{app:welfare}).
\end{proof}

The lift in \eqref{eq:welfare-lift} should be viewed as a benchmark: it is the most the audit instrument can buy when verification is free and never errs. Audit costs and false positives erode it, and as we show in Section~\ref{sec:risk-tuning}, welfare then becomes single-peaked in the audit intensity.

\textit{The price of cost-only design.}
How much welfare does a designer forgo by ignoring the audit instrument? Let $W^{\mathrm{two}}$ and $W^{\mathrm{cost}}$ be the optimal Stackelberg values of the problems with and without audit, net of audit costs and false-flag losses, and let the \emph{informational value of the gameable coordinates} be $V_{\Gset}:=W^{\mathrm{ub}}-W^{-\Gset}$, where $W^{\mathrm{ub}}$ is the pointwise utopia value (each agent assigned its best feasible outcome) and $W^{-\Gset}$ the honest-response value of the best classifier that zeroes every gameable coordinate.
For tightness, use the scalar family of \S\ref{sec:risk-tuning}: qualified mass $q\le\tfrac12$ sits at the cap of a bounded purely fakeable feature $x\in[0,1]$, unqualified agents sit at the bottom, gaming costs $c_G<u_+$, and unqualified agents have no honest recourse.

\begin{thm}[Price of cost-only design]\label{thm:price}
(i) In general, $0\le W^{\mathrm{two}}-W^{\mathrm{cost}}\le V_{\Gset}$.
(ii) The bound is tight and the ratio unbounded. In the scalar family above, every deterministic cost-only design has value $\max(0,2q-1)=0$, while auditing at the deterrence level $p_g=(u_+-c_G)/\Delta$ yields $W^{\mathrm{two}}\ge q(1-\alpha p_g)-c_ap_g$.

With a perfect audit, the gap equals $q=V_{\Gset}$, the entire attainable welfare.
\end{thm}
\begin{proof}[Proof sketch]
(i) $p\equiv0$ is feasible with audit, and every agent's contribution is at most its best feasible outcome, so $W^{\mathrm{cost}}\le W^{\mathrm{two}}\le W^{\mathrm{ub}}$; zeroing the gameable coordinates makes gaming useless, so $W^{\mathrm{cost}}\ge W^{-\Gset}$. (ii) In the scalar family, any cost-only threshold $\tau\le1$ is reached by every unqualified agent at cost $c_G\tau<u_+$, and any $\tau>1$ is reached by nobody, so the cost-only value is $\max(0,2q-1)$; auditing at $\tau=1$ with $p_g$ deters all gaming and yields the stated value. Details and the randomization claim are in Appendix~\ref{app:welfare}.
\end{proof}

We note that boundedness of the signal in (ii) is essential. With an unbounded signal, the cost-only firm can raise the threshold one gaming budget above the qualified position and restore Spence separation without any audit; moreover, even randomized acceptance rules that deter gaming forgo at least $q(1-c_G/u_+)$ of the gain from auditing (Remark~\ref{rmk:headroom}, Appendix~\ref{app:welfare}). In other words, the theorem does not say that every fakeable feature deserves to be audited: if the fakeable information is weak, the gain is small, and if an unaudited threshold can separate types through an unbounded signal, the gain collapses. Audit recovers information, rather than merely punishing gaming, for a signal that is predictive, cheap to imitate, and bounded enough that the firm cannot screen by raising the bar.

\textit{Domain separation.}
Our next result concerns identification, i.e., which cross-domain behaviors can and cannot be explained by a cost-only model. A domain is a population with a common classifier, cost structure, and acceptance utility $(w,\theta,\Iset,\Gset,C_I,C_G,u_+)$, and its own audit profile $(p,u_-)$. Holding the shared structure fixed fixes the gaming rent $D$.

\begin{thm}[Audit-induced domain separation]\label{thm:domainsep}
Consider two domains $A$ and $B$ with the same costs, classifier, causal structure, and acceptance utility, hence the same rent $D>0$ on the binding gaming channel. Assume honest improvement is feasible and individually rational, $u_+-\PhiI\ge0$. Suppose the observed behavior differs: in $A$ the agent games, while in $B$ the agent improves.
\begin{enumerate}
\item No cost-only model with these primitives can rationalize the pair: with $p\equiv0$, every domain with the same cost and causal primitives makes the same prediction.
\item Our model rationalizes the pair iff $p_A\Delta_A<D\le p_B\Delta_B$. If $p_A$ and $p_B$ are observed, the risk primitives are sharply partially identified by $\Delta_A<D/p_A$, $\Delta_B\ge D/p_B$. If $p_A=p_B$, the improving domain has a strictly worse caught state: $u_{-,B}<u_{-,A}$.
\end{enumerate}
\end{thm}
\begin{proof}
When $p\equiv0$, the condition $p\Delta\ge D$ of Lemma~\ref{lem:deterrence-frontier} fails in every domain, so the binding type behaves identically in $A$ and $B$, which proves (1). For (2), Lemma~\ref{lem:deterrence-frontier} gives that $A$ games iff $p_A\Delta_A<D$ and $B$ improves iff $p_B\Delta_B\ge D$. Every pair of differentials satisfying these inequalities reproduces the observed behavior and every pair violating them contradicts one observation, which is the sharp set. If $p_A=p_B=p$, then $\Delta_A<D/p\le\Delta_B$, and since $u_+$ is common, $u_{-,B}<u_{-,A}$.
\end{proof}

Intuitively, a loan office and an admissions office can induce different behavior without any difference in the cost of effort or in the causal meaning of a feature, but only in what happens when a fake is caught. We note that the rent $D$ is treated as known here; the maintained assumption that the domains share their costs, classifier, and causal structure is an identifying restriction that the two behaviors cannot themselves test, and without it, a split could also reflect unmodeled cost differences. Figure~\ref{fig:domain-audit-field} illustrates the split on a synthetic population; its construction is in Appendix~\ref{app:figure2}.

\section{Detection versus penalty}\label{sec:risk-tuning}
Verification has false positives: honest applicants can be wrongly flagged, and the penalty can fall on them. We show that welfare is then single-peaked in the audit intensity, and that the firm and the planner, while agreeing on how much deterrence to buy, disagree on the mix of detection and penalty that buys it.

\textit{False positives.}
In the scalar specialization, a qualified mass $q$ clears honestly and an unqualified mass $1-q$ can game at cost $c_G$. Audit intensity $p$ catches a gamer with probability $p$ and wrongly flags an honest applicant with probability $\alpha p$, $\alpha\in(0,1)$, so a qualified applicant's utility is $EU_{\mathrm{hon}}(p)=u_+-\alpha p\Delta$. Let $p_g:=\frac{u_+-c_G}{\Delta}$ and $p_q:=\frac{u_+}{\alpha\Delta}$. $p_g$ is the audit level that deters gaming; $p_q$ is the level at which qualified applicants are chilled out. Since $\alpha<1$, $p_g<p_q$.

\begin{prop}[Over-penalization trap]\label{prop:overpen}
Assume the deterrence jump dominates the cost-only corner, $1-q\alpha p_g>q$. With audit false positives, firm accuracy is single-peaked in $p$ and uniquely maximized at the deterrence-binding level $p^\star=p_g$: below $p_g$ gamers enter; on $[p_g,p_q)$ additional audit only wrongly flags honest applicants; beyond $p_q$ the qualified abstain and accuracy falls to its floor. The same holds for utilitarian welfare subtracting false-accusation harm, under the analogous corner condition.
\end{prop}
\begin{proof}
An unqualified agent games iff $u_+-p\Delta-c_G>0$, i.e., $p<p_g$; a qualified agent applies iff $u_+-\alpha p\Delta\ge0$, i.e., $p\le p_q$. Accuracy is therefore $A(p)=q(1-\alpha p)\indic\{p\le p_q\}+(1-q)[p\indic\{p<p_g\}+\indic\{p\ge p_g\}]$: it rises with slope $(1-q)-q\alpha$ or falls before $p_g$, jumps by $(1-q)(1-p_g)$ at $p_g$, falls with slope $-q\alpha$ on $[p_g,p_q)$, and drops to $1-q$ at $p_q$. The unique maximizer is $p_g$ whenever $A(p_g)=1-q\alpha p_g>q=A(0)$. A convex false-accusation harm only steepens the decline after $p_g$.
\end{proof}

Intuitively, above $p_g$ the marginal audit has no gamer left to deter, and only honest applicants left to harm; Figure~\ref{fig:overpenalization} in Appendix~\ref{app:risk} plots the resulting diagnostic. As we show next, however, agreement on the level of audit does not imply agreement on how to deliver it.

\textit{The detection-penalty mix.}
Recall that the agents' behavior depends only on the product $p\Delta$, so that every pair with $p\Delta=D$ deters the marginal gamer. These pairs are not welfare-equivalent, however, since detection costs resources while a harsher penalty imposes a convex harm on the wrongly accused. Let the detection cost be $c_ap$ and the false-accusation harm be $\kappa\varphi(\Delta)$, with $\varphi$ increasing and strictly convex, and restrict attention to deterrence-binding designs with $p\Delta=D$ and $\Delta\in[D,\Delta_{\max}]$, where $\Delta_{\max}$ is the penalty cap of Section~\ref{sec:model}. (Without the cap and without false-accusation harm, nothing would stop the firm from choosing $p\to0$ and $\Delta\to\infty$.)

Among such designs the firm objective $U_f(\Delta)=1-(q\alpha+c_a)D/\Delta$ is strictly increasing in $\Delta$, so a private firm chooses maximal admissible punishment and minimal detection, whereas utilitarian welfare has an interior optimum: for $\varphi(\Delta)=\Delta^2/2$, $\Delta^\star=\sqrt{2(q\alpha+c_a)/(\kappa q\alpha)}$, independent of the gaming rent $D$, so higher rents should raise detection $p^\star=D/\Delta^\star$, not severity (Proposition~\ref{prop:composition}, Appendix~\ref{app:risk}).

In other words, a firm that does not bear the harm of false accusations chooses the Becker corner \cite{becker1968crime} of severe penalties and little detection, whereas the planner, whose interior penalty is pinned down by the honest applicants flagged at rate $\alpha p$, audits more and penalizes less \cite{polinsky2000enforcement}. Both mixes deter the marginal gamer through the same product $p\Delta$, but they distribute the errors differently; where the falsely accused cannot be compensated, severity is therefore not a free substitute for detection. Figure~\ref{fig:risk-landscape} in Appendix~\ref{app:risk} plots the wedge.

\textit{Robustness.}
Our analysis has relied on three idealizations: audits that never flag the honest, magnitude-independent detection, and purely fakeable features. Figure~\ref{fig:robustness} relaxes each of these in a heterogeneous Stackelberg simulation and reports the \emph{audit dividend}, i.e., the welfare gap between the optimal design with audit and the best cost-only design. In these experiments, the dividend remains nonnegative: with audit false positives (panel (a)), the firm's value is single-peaked with an interior optimum at the deterrence level of Proposition~\ref{prop:overpen}, now in the full design rather than the scalar family, and the firm's optimum lies above the planner's because the firm does not bear the false-accusation harm; with magnitude-dependent detection $p_k(\gamma)=p\,(1-e^{-(\beta\gamma)^2})$ (panel (b)), fakes diversify from singleton to multi-channel as predicted by Proposition~\ref{prop:magnitude}, and the dividend vanishes only when even large fakes evade detection; and with causal leakage $\rho$ (panel (c)), pure fakeability is the worst case, and the dividend shrinks toward zero as gaming increasingly turns gamers into qualified agents. The construction and parameters are in Appendix~\ref{app:experiments}.

\begin{figure}[t]
\centering
\includegraphics[width=\linewidth]{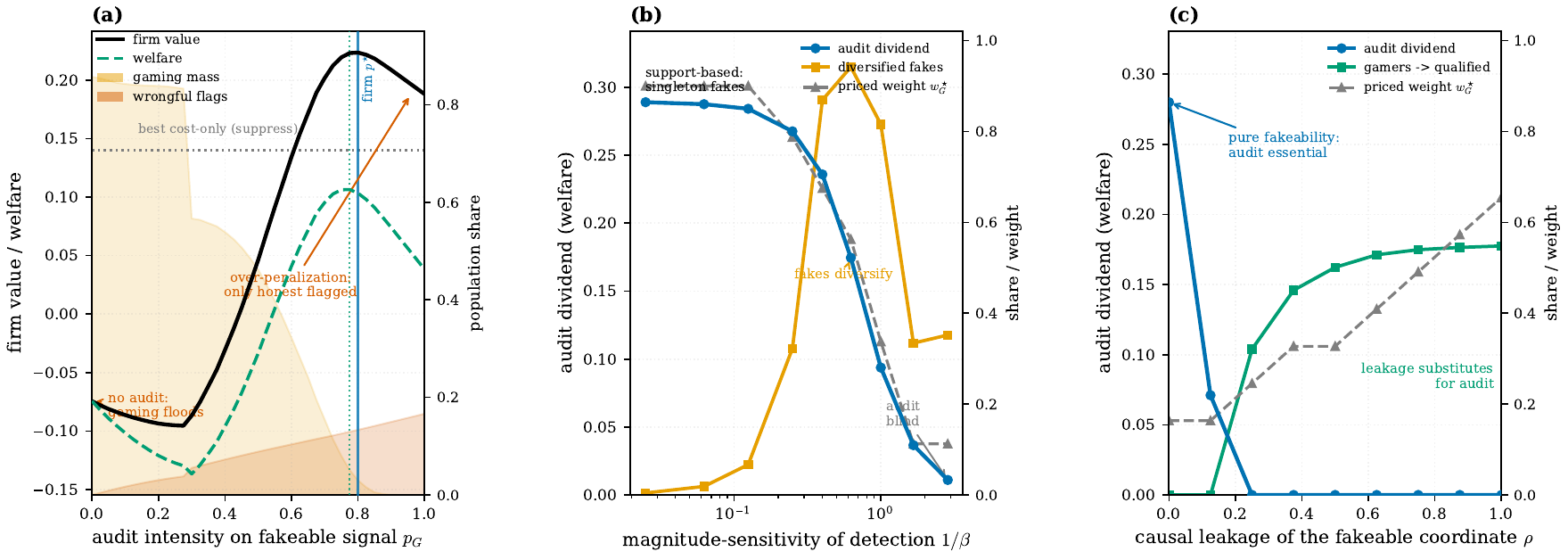}
\caption{Robustness of the price-and-audit result to the three idealizations the model isolates, measured by the audit dividend (optimal welfare with audit minus the best cost-only design) in a fixed-seed heterogeneous Stackelberg simulation. (a) Audits that wrongly flag honest applicants at rate $\alpha p$: firm value and welfare are single-peaked in audit intensity; the dotted line is the best cost-only design, the shaded areas are gaming mass and wrongful flags, and the firm's optimum sits above the welfare optimum. (b) Magnitude-dependent detection with sensitivity $1/\beta$: the dividend persists and fakes diversify from singleton (support-based regime, left) to multi-channel, collapsing only when detection is blind to large fakes (right). (c) Causal leakage $\rho$: the dividend is largest at pure fakeability and shrinks as faking the signal increasingly converts gamers into qualified agents, while the signal stays priced throughout.}
\label{fig:robustness}
\end{figure}

\section{Discussion and Conclusion}\label{sec:disc}
We have proposed a model of strategic classification in which the firm designs not only a classifier, but also an audit profile over the fakeable features, with detection probability and penalty as separate instruments. This model enabled us to ask how the two instruments interact, how audit should be allocated, what it costs to compute and to learn, and what it is worth. At a technical level, the key structural result is that the classifier affects the audit problem only through the distribution of gaming rents it induces (Theorem~\ref{thm:rentsuff}). This is what allowed us to hold the manipulation technology and the causal structure fixed while varying institutional risk alone, and it is also where the computational difficulty lives: audit allocation is exactly solvable for few channels (Proposition~\ref{prop:vertex}), is Densest-$k$-Subgraph-hard when gaming routes overlap under an inspection cap (Theorem~\ref{thm:hardness}), and the joint design is polynomial in fixed dimension (Theorem~\ref{thm:fixeddim}).

Our findings suggest a message that is narrower than ``audit more.'' Verification is valuable when it prices a rent that would otherwise force the classifier to discard a useful signal. Once that rent is priced out, additional audit buys no deterrence and falls on honest applicants, so that welfare is single-peaked in the audit intensity (Proposition~\ref{prop:overpen}); moreover, how a given level of deterrence is delivered matters, since the firm prefers severity and the planner prefers detection (Proposition~\ref{prop:composition}). In other words, the right target is the least verification that prices a rent out, rather than the most that the firm can afford.

We conclude with a discussion of our modeling assumptions and directions for future work. First, our model is a one-shot Stackelberg game against a known type distribution, in the mechanism-design strand of this literature rather than its learning strand \cite{dong2018revealed,zrnic2021leads,sundaram2021pac}. Theorem~\ref{thm:rentsuff} identifies what a learner would need to estimate, namely the rent profile, Theorem~\ref{thm:quantile-audit} bounds the regret in terms of the Kolmogorov distance of the estimated risk distribution, and Section~\ref{sec:learning} gives the bandit rates when the only feedback is the audit itself; repeated interaction with the same agents, audit reputation, and strategic delay remain outside the model. Second, we take the audit profile to be committed and observed. This is descriptively accurate where verification policy is disclosed or institutionally fixed (e.g., published audit criteria or mandated proctoring), and it is otherwise a commitment benchmark: a concealed audit makes the agent respond to a belief $\hat p_j\le p_j$, which weakly relaxes the frontier of Lemma~\ref{lem:deterrence-frontier} to $\hat p_j\Delta\ge D_j$, so that the priced rent $p_j\Delta$ bounds what any hidden-audit scheme buys at the same expected inspection cost. Concealment can help only through channels that our one-shot model excludes, such as correlated beliefs across agents or ambiguity aversion. Third, the extensive-margin structure of the agent's response relies on magnitude-independent detection (Propositions~\ref{prop:magnitude} and~\ref{prop:neardecouple}), and our welfare analysis does not model appeals, remediation, or legal limits on penalties beyond the cap $\Delta_{\max}$. Finally, the complexity of audit allocation under the full outcome objective, and of the nested-prefix instances that arise under a common gaming technology, remain open (Proposition~\ref{prop:catch-credit} and Corollary~\ref{cor:inapprox}). Extending the two-instrument decomposition to dynamic, appealable, and concealed enforcement, and understanding how it deforms there, is an interesting direction for future work.

\bibliographystyle{plainnat}
\bibliography{references}

@inproceedings{hardt2016strategic,
  title     = {Strategic Classification},
  author    = {Hardt, Moritz and Megiddo, Nimrod and Papadimitriou, Christos and Wootters, Mary},
  booktitle = {Proc.\ ACM Conf.\ on Innovations in Theoretical Computer Science (ITCS)},
  year      = {2016},
  pages = {111--122},
  publisher = {ACM},
  address = {New York, NY, USA},
}

@inproceedings{kleinberg2019induce,
  title     = {How Do Classifiers Induce Agents To Invest Effort Strategically?},
  author    = {Kleinberg, Jon and Raghavan, Manish},
  booktitle = {Proc.\ ACM Conf.\ on Economics and Computation (EC)},
  year      = {2019},
  note      = {Journal version: ACM TEAC 2020. arXiv:1807.05307},
  pages = {825--844},
  publisher = {ACM},
  address = {New York, NY, USA},
}

@inproceedings{miller2020causal,
  title     = {Strategic Classification is Causal Modeling in Disguise},
  author    = {Miller, John and Milli, Smitha and Hardt, Moritz},
  booktitle = {Proc.\ Int.\ Conf.\ on Machine Learning (ICML)},
  year      = {2020},
  note      = {arXiv:1910.10362},
  pages = {6917--6926},
  publisher = {PMLR},
  address = {Vienna, Austria},
}

@inproceedings{jagadeesan2021microfoundations,
  title     = {Alternative Microfoundations for Strategic Classification},
  author    = {Jagadeesan, Meena and Mendler-D{\"u}nner, Celestine and Hardt, Moritz},
  booktitle = {Proc.\ Int.\ Conf.\ on Machine Learning (ICML)},
  year      = {2021},
  note      = {arXiv:2106.12705},
  pages = {4687--4697},
  publisher = {PMLR},
  address = {Virtual Event},
}

@inproceedings{dong2018revealed,
  title     = {Strategic Classification from Revealed Preferences},
  author    = {Dong, Jinshuo and Roth, Aaron and Schutzman, Zachary and Waggoner, Bo and Wu, Zhiwei Steven},
  booktitle = {Proc.\ ACM Conf.\ on Economics and Computation (EC)},
  pages     = {55--70},
  year      = {2018},
  note      = {arXiv:1710.07887},
  publisher = {ACM},
  address = {New York, NY, USA},
}

@misc{chen2020strategyaware,
  title  = {Learning Strategy-Aware Linear Classifiers},
  author = {Chen, Yiling and Liu, Yang and Podimata, Chara},
  year   = {2020},
  note   = {NeurIPS 2020. arXiv:1911.04004}
}

@inproceedings{levanon2021practical,
  title     = {Strategic Classification Made Practical},
  author    = {Levanon, Sagi and Rosenfeld, Nir},
  booktitle = {Proc.\ Int.\ Conf.\ on Machine Learning (ICML)},
  year      = {2021},
  note      = {arXiv:2103.01826},
  pages = {6243--6253},
  publisher = {PMLR},
  address = {Virtual Event},
}

@inproceedings{shavit2020causal,
  title     = {Causal Strategic Linear Regression},
  author    = {Shavit, Yonadav and Edelman, Benjamin and Axelrod, Brian},
  booktitle = {Proc.\ Int.\ Conf.\ on Machine Learning (ICML)},
  year      = {2020},
  note      = {arXiv:2002.10066},
  pages = {8676--8686},
  publisher = {PMLR},
  address = {Vienna, Austria},
}

@inproceedings{harris2022strategic,
  title     = {Strategic Instrumental Variable Regression: Recovering Causal Relationships From Strategic Responses},
  author    = {Harris, Keegan and Ngo, Dung Daniel T. and Stapleton, Logan and Heidari, Hoda and Wu, Zhiwei Steven},
  booktitle = {Proc.\ Int.\ Conf.\ on Machine Learning (ICML)},
  year      = {2022},
  note      = {arXiv:2107.05762},
  pages = {8641--8661},
  publisher = {PMLR},
  address = {Baltimore, MD, USA},
}

@inproceedings{bechavod2021gaming,
  title     = {Gaming Helps! Learning from Strategic Interactions in Natural Dynamics},
  author    = {Bechavod, Yahav and Ligett, Katrina and Wu, Zhiwei Steven and Ziani, Juba},
  booktitle = {Proc.\ Int.\ Conf.\ on Artificial Intelligence and Statistics (AISTATS)},
  year      = {2021},
  note      = {arXiv:2002.07024},
  pages = {1234--1242},
  publisher = {PMLR},
  address = {Virtual Event},
}

@incollection{avenhaus2002inspection,
  title     = {Inspection Games},
  author    = {Avenhaus, Rudolf and von Stengel, Bernhard and Zamir, Shmuel},
  booktitle = {Handbook of Game Theory with Economic Applications},
  editor    = {Aumann, Robert J. and Hart, Sergiu},
  volume    = {3},
  pages     = {1947--1987},
  publisher = {Elsevier},
  year      = {2002},
  address = {Amsterdam, The Netherlands},
}

@inproceedings{blocki2013audit,
  title     = {Audit Games},
  author    = {Blocki, Jeremiah and Christin, Nicolas and Datta, Anupam and Procaccia, Ariel D. and Sinha, Arunesh},
  booktitle = {Proc.\ Int.\ Joint Conf.\ on Artificial Intelligence (IJCAI)},
  year      = {2013},
  pages = {41--47},
  publisher = {AAAI Press},
  address = {Beijing, China},
}

@article{becker1968crime,
  title   = {Crime and Punishment: An Economic Approach},
  author  = {Becker, Gary S.},
  journal = {Journal of Political Economy},
  volume  = {76},
  number  = {2},
  pages   = {169--217},
  year    = {1968}
}

@article{townsend1979csv,
  title   = {Optimal Contracts and Competitive Markets with Costly State Verification},
  author  = {Townsend, Robert M.},
  journal = {Journal of Economic Theory},
  volume  = {21},
  number  = {2},
  pages   = {265--293},
  year    = {1979}
}

@article{reinganum1985compliance,
  title   = {Income Tax Compliance in a Principal-Agent Framework},
  author  = {Reinganum, Jennifer F. and Wilde, Louis L.},
  journal = {Journal of Public Economics},
  volume  = {26},
  number  = {1},
  pages   = {1--18},
  year    = {1985}
}

@article{mookherjee1989auditing,
  title   = {Optimal Auditing, Insurance, and Redistribution},
  author  = {Mookherjee, Dilip and Png, Ivan},
  journal = {Quarterly Journal of Economics},
  volume  = {104},
  number  = {2},
  pages   = {399--415},
  year    = {1989}
}

@article{benporath2014verification,
  title   = {Optimal Allocation with Costly Verification},
  author  = {Ben-Porath, Elchanan and Dekel, Eddie and Lipman, Barton L.},
  journal = {American Economic Review},
  volume  = {104},
  number  = {12},
  pages   = {3779--3813},
  year    = {2014}
}

@inproceedings{perdomo2020performative,
  title     = {Performative Prediction},
  author    = {Perdomo, Juan C. and Zrnic, Tijana and Mendler-D{\"u}nner, Celestine and Hardt, Moritz},
  booktitle = {Proc.\ Int.\ Conf.\ on Machine Learning (ICML)},
  year      = {2020},
  note      = {arXiv:2002.06673},
  pages = {7599--7609},
  publisher = {PMLR},
  address = {Vienna, Austria},
}

@article{elzayn2025disparities,
  title   = {Measuring and Mitigating Racial Disparities in Tax Audits},
  author  = {Elzayn, Hadi and Smith, Evelyn and Hertz, Thomas and Guage, Cameron and Ramesh, Arun and Fisher, Robin and Ho, Daniel E. and Goldin, Jacob},
  journal = {Quarterly Journal of Economics},
  volume  = {140},
  number  = {1},
  pages   = {113--163},
  year    = {2025}
}

@book{eubanks2018automating,
  title     = {Automating Inequality: How High-Tech Tools Profile, Police, and Punish the Poor},
  author    = {Eubanks, Virginia},
  publisher = {St.\ Martin's Press},
  year      = {2018},
  address = {New York, NY, USA},
}

@inproceedings{ebrahimi2025double,
  title={The double-edged sword of behavioral responses in strategic classification: Theory and user studies},
  author={Ebrahimi, Raman and Vaccaro, Kristen and Naghizadeh, Parinaz},
  booktitle={Proceedings of the 2025 ACM Conference on Fairness, Accountability, and Transparency},
  pages={868--886},
  year={2025},
  publisher = {ACM},
  address = {New York, NY, USA},
}

@inproceedings{ahmadi2021perceptron,
  title     = {The Strategic Perceptron},
  author    = {Ahmadi, Saba and Beyhaghi, Hedyeh and Blum, Avrim and Naggita, Keziah},
  booktitle = {Proc.\ ACM Conf.\ on Economics and Computation (EC)},
  year      = {2021},
  pages = {6--25},
  publisher = {ACM},
  address = {New York, NY, USA},
}

@inproceedings{alon2020multiagent,
  title     = {Multiagent Evaluation Mechanisms},
  author    = {Alon, Tal and Dobson, Magdalen and Procaccia, Ariel D. and Talgam-Cohen, Inbal and Tucker-Foltz, Jamie},
  booktitle = {Proc.\ AAAI Conf.\ on Artificial Intelligence (AAAI)},
  year      = {2020},
  pages = {1774--1781},
  publisher = {AAAI Press},
  address = {Palo Alto, CA, USA},
}

@book{basu2006algorithms,
  title     = {Algorithms in Real Algebraic Geometry},
  author    = {Basu, Saugata and Pollack, Richard and Roy, Marie-Fran{\c c}oise},
  edition   = {2nd},
  publisher = {Springer},
  year      = {2006},
  address = {Berlin, Germany},
}

@inproceedings{braverman2020randomness,
  title     = {The Role of Randomness and Noise in Strategic Classification},
  author    = {Braverman, Mark and Garg, Sumegha},
  booktitle = {Proc.\ Symposium on Foundations of Responsible Computing (FORC)},
  year      = {2020},
  pages = {9:1--9:10},
  publisher = {Schloss Dagstuhl -- Leibniz-Zentrum f{\"u}r Informatik},
  address = {Dagstuhl, Germany},
}

@misc{burnat2026dpaudit,
  title  = {Differentially Private Auditing under Strategic Response},
  author = {Burnat, F. A. D.},
  year   = {2026},
  note   = {arXiv:2605.07674}
}

@inproceedings{ceppi2019partial,
  title     = {Partial Verification as a Substitute for Money},
  author    = {Ceppi, Sofia and Kash, Ian and Frongillo, Rafael},
  booktitle = {Proc.\ AAAI Conf.\ on Artificial Intelligence (AAAI)},
  year      = {2019},
  publisher = {AAAI Press},
  address = {Palo Alto, CA, USA},
  numpages = {8},
}

@inproceedings{dai2025audits,
  title     = {Non-Monetary Mechanism Design without Priors: Achieving Efficiency via Adaptive Costly Audits},
  author    = {Dai, Yan and Blanchard, Moise and Jaillet, Patrick},
  booktitle = {Proc.\ Conf.\ on Learning Theory (COLT)},
  year      = {2025},
  publisher = {PMLR},
  address = {Lyon, France},
  numpages = {40},
}

@inproceedings{ghalme2021dark,
  title     = {Strategic Classification in the Dark},
  author    = {Ghalme, Ganesh and Nair, Vineet and Eilat, Itay and Talgam-Cohen, Inbal and Rosenfeld, Nir},
  booktitle = {Proc.\ Int.\ Conf.\ on Machine Learning (ICML)},
  year      = {2021},
  pages = {3672--3681},
  publisher = {PMLR},
  address = {Virtual Event},
}

@article{green1986verifiable,
  title   = {Partially Verifiable Information and Mechanism Design},
  author  = {Green, Jerry R. and Laffont, Jean-Jacques},
  journal = {Review of Economic Studies},
  volume  = {53},
  number  = {3},
  pages   = {447--456},
  year    = {1986}
}

@inproceedings{haghtalab2020welfare,
  title     = {Maximizing Welfare with Incentive-Aware Evaluation Mechanisms},
  author    = {Haghtalab, Nika and Immorlica, Nicole and Lucier, Brendan and Wang, Jack Z.},
  booktitle = {Proc.\ Int.\ Joint Conf.\ on Artificial Intelligence (IJCAI)},
  year      = {2020},
  pages = {160--166},
  publisher = {ijcai.org},
  address = {Yokohama, Japan},
}

@article{hart2017evidence,
  title   = {Evidence Games: Truth and Commitment},
  author  = {Hart, Sergiu and Kremer, Ilan and Perry, Motty},
  journal = {American Economic Review},
  volume  = {107},
  number  = {3},
  pages   = {690--713},
  year    = {2017}
}

@inproceedings{hu2019disparate,
  title     = {The Disparate Effects of Strategic Manipulation},
  author    = {Hu, Lily and Immorlica, Nicole and Vaughan, Jennifer Wortman},
  booktitle = {Proc.\ ACM Conf.\ on Fairness, Accountability, and Transparency (FAccT)},
  year      = {2019},
  pages = {259--268},
  publisher = {ACM},
  address = {New York, NY, USA},
}

@article{khot2006ptas,
  title   = {Ruling Out PTAS for Graph Min-Bisection, Dense $k$-Subgraph, and Bipartite Clique},
  author  = {Khot, Subhash},
  journal = {SIAM Journal on Computing},
  volume  = {36},
  number  = {4},
  pages   = {1025--1071},
  year    = {2006}
}

@inproceedings{krishnaswamy2021withheld,
  title     = {Classification with Strategically Withheld Data},
  author    = {Krishnaswamy, Anilesh K. and Li, Haoming and Rein, David and Zhang, Hanrui and Conitzer, Vincent},
  booktitle = {Proc.\ AAAI Conf.\ on Artificial Intelligence (AAAI)},
  year      = {2021},
  pages = {5514--5522},
  publisher = {AAAI Press},
  address = {Palo Alto, CA, USA},
}

@inproceedings{manurangsi2017dks,
  title     = {Almost-Polynomial Ratio ETH-Hardness of Approximating Densest $k$-Subgraph},
  author    = {Manurangsi, Pasin},
  booktitle = {Proc.\ ACM Symposium on Theory of Computing (STOC)},
  year      = {2017},
  pages = {954--961},
  publisher = {ACM},
  address = {New York, NY, USA},
}

@inproceedings{milli2019social,
  title     = {The Social Cost of Strategic Classification},
  author    = {Milli, Smitha and Miller, John and Dragan, Anca D. and Hardt, Moritz},
  booktitle = {Proc.\ ACM Conf.\ on Fairness, Accountability, and Transparency (FAccT)},
  year      = {2019},
  pages = {230--239},
  publisher = {ACM},
  address = {New York, NY, USA},
}

@article{mylovanov2017expost,
  title   = {Optimal Allocation with Ex Post Verification and Limited Penalties},
  author  = {Mylovanov, Tymofiy and Zapechelnyuk, Andriy},
  journal = {American Economic Review},
  volume  = {107},
  number  = {9},
  pages   = {2666--2694},
  year    = {2017}
}

@article{polinsky2000enforcement,
  title   = {The Economic Theory of Public Enforcement of Law},
  author  = {Polinsky, A. Mitchell and Shavell, Steven},
  journal = {Journal of Economic Literature},
  volume  = {38},
  number  = {1},
  pages   = {45--76},
  year    = {2000}
}

@inproceedings{sinha2018security,
  title     = {Stackelberg Security Games: Looking Beyond a Decade of Success},
  author    = {Sinha, Arunesh and Fang, Fei and An, Bo and Kiekintveld, Christopher and Tambe, Milind},
  booktitle = {Proc.\ Int.\ Joint Conf.\ on Artificial Intelligence (IJCAI)},
  year      = {2018},
  pages = {5494--5501},
  publisher = {ijcai.org},
  address = {Stockholm, Sweden},
}

@inproceedings{sundaram2021pac,
  title     = {PAC-Learning for Strategic Classification},
  author    = {Sundaram, Ravi and Vullikanti, Anil and Xu, Haifeng and Yao, Fan},
  booktitle = {Proc.\ Int.\ Conf.\ on Machine Learning (ICML)},
  year      = {2021},
  pages = {9978--9988},
  publisher = {PMLR},
  address = {Virtual Event},
}

@inproceedings{zrnic2021leads,
  title     = {Who Leads and Who Follows in Strategic Classification?},
  author    = {Zrnic, Tijana and Mazumdar, Eric and Sastry, Shankar and Jordan, Michael I.},
  booktitle = {Proc.\ Advances in Neural Information Processing Systems (NeurIPS)},
  year      = {2021},
  pages = {15257--15269},
  publisher = {Curran Associates},
  address = {Red Hook, NY, USA},
}

@inproceedings{estornell2021audits,
  author    = {Estornell, Andrew and Das, Sanmay and Vorobeychik, Yevgeniy},
  title     = {Incentivizing Truthfulness Through Audits in Strategic Classification},
  booktitle = {Proceedings of the AAAI Conference on Artificial Intelligence},
  volume    = {35(6)},
  pages     = {5347--5354},
  year      = {2021},
  publisher = {AAAI Press},
  address = {Palo Alto, CA, USA},
}

@inproceedings{estornell2023recourse,
  author    = {Estornell, Andrew and Chen, Yatong and Das, Sanmay and Liu, Yang and Vorobeychik, Yevgeniy},
  title     = {Incentivizing Recourse through Auditing in Strategic Classification},
  booktitle = {Proceedings of the Thirty-Second International Joint Conference on Artificial Intelligence (IJCAI)},
  pages     = {400--408},
  year      = {2023},
  publisher = {ijcai.org},
  address = {Macao, China},
}

@article{cezar2020adversarial,
  author    = {Cezar, Asunur and Raghunathan, Srinivasan and Sarkar, Sumit},
  title     = {Adversarial Classification: Impact of Agents' Faking Cost on Firms and Agents},
  journal   = {Production and Operations Management},
  volume    = {29},
  number    = {12},
  pages     = {2789--2807},
  year      = {2020}
}

@inproceedings{fotakis2016selective,
  author    = {Fotakis, Dimitris and Tzamos, Christos and Zampetakis, Emmanouil},
  title     = {Mechanism Design with Selective Verification},
  booktitle = {Proceedings of the 2016 ACM Conference on Economics and Computation (EC)},
  pages     = {771--788},
  year      = {2016},
  publisher = {ACM},
  address = {New York, NY, USA},
}

@inproceedings{caragiannis2012probabilistic,
  author    = {Caragiannis, Ioannis and Elkind, Edith and Szegedy, Mario and Yu, Lan},
  title     = {Mechanism Design: From Partial to Probabilistic Verification},
  booktitle = {Proceedings of the 13th ACM Conference on Electronic Commerce (EC)},
  pages     = {266--283},
  year      = {2012},
  publisher = {ACM},
  address = {New York, NY, USA},
}

@inproceedings{ball2019probabilistic,
  author    = {Ball, Ian and Kattwinkel, Deniz},
  title     = {Probabilistic Verification in Mechanism Design},
  booktitle = {Proceedings of the 2019 ACM Conference on Economics and Computation (EC)},
  pages     = {389--390},
  year      = {2019},
  publisher = {ACM},
  address = {New York, NY, USA},
}

@inproceedings{kephart2016reporting,
  author    = {Kephart, Andrew and Conitzer, Vincent},
  title     = {The Revelation Principle for Mechanism Design with Reporting Costs},
  booktitle = {Proceedings of the 2016 ACM Conference on Economics and Computation (EC)},
  pages     = {85--102},
  year      = {2016},
  publisher = {ACM},
  address = {New York, NY, USA},
}

@article{auer2002nonstochastic,
  author  = {Auer, Peter and Cesa-Bianchi, Nicol{\`o} and Freund, Yoav and Schapire, Robert E.},
  title   = {The Nonstochastic Multiarmed Bandit Problem},
  journal = {SIAM Journal on Computing},
  volume  = {32},
  number  = {1},
  pages   = {48--77},
  year    = {2002}
}

@book{lattimore2020bandit,
  author    = {Lattimore, Tor and Szepesv{\'a}ri, Csaba},
  title     = {Bandit Algorithms},
  publisher = {Cambridge University Press},
  address   = {Cambridge, UK},
  year      = {2020}
}

@inproceedings{agarwal2011bandit,
  author    = {Agarwal, Alekh and Foster, Dean P. and Hsu, Daniel and Kakade, Sham M. and Rakhlin, Alexander},
  title     = {Stochastic Convex Optimization with Bandit Feedback},
  booktitle = {Advances in Neural Information Processing Systems 24 (NeurIPS)},
  pages     = {1035--1043},
  publisher = {Curran Associates},
  address   = {Red Hook, NY, USA},
  year      = {2011}
}

@inproceedings{shamir2013complexity,
  author    = {Shamir, Ohad},
  title     = {On the Complexity of Bandit and Derivative-Free Stochastic Convex Optimization},
  booktitle = {Proceedings of the 26th Annual Conference on Learning Theory (COLT)},
  pages     = {3--24},
  publisher = {PMLR},
  address   = {Princeton, NJ, USA},
  year      = {2013}
}

@inproceedings{verma2019censored,
  author    = {Verma, Arun and Hanawal, Manjesh K. and Rajkumar, Arun and Sankaran, Raman},
  title     = {Censored Semi-Bandits: A Framework for Resource Allocation with Censored Feedback},
  booktitle = {Advances in Neural Information Processing Systems 32 (NeurIPS)},
  year      = {2019},
  pages = {14526--14536},
  publisher = {Curran Associates},
  address = {Red Hook, NY, USA},
}

@inproceedings{haghtalab2022nonmyopic,
  author    = {Haghtalab, Nika and Lykouris, Thodoris and Nietert, Sloan and Wei, Alexander},
  title     = {Learning in Stackelberg Games with Non-myopic Agents},
  booktitle = {Proceedings of the 23rd ACM Conference on Economics and Computation (EC)},
  pages     = {917--918},
  year      = {2022},
  publisher = {ACM},
  address = {New York, NY, USA},
}

@article{schrijver2000sfm,
  author  = {Schrijver, Alexander},
  title   = {A Combinatorial Algorithm Minimizing Submodular Functions in Strongly Polynomial Time},
  journal = {Journal of Combinatorial Theory, Series B},
  volume  = {80},
  number  = {2},
  pages   = {346--355},
  year    = {2000}
}

@article{iwata2001sfm,
  author  = {Iwata, Satoru and Fleischer, Lisa and Fujishige, Satoru},
  title   = {A Combinatorial Strongly Polynomial Algorithm for Minimizing Submodular Functions},
  journal = {Journal of the ACM},
  volume  = {48},
  number  = {4},
  pages   = {761--777},
  year    = {2001}
}

@book{grotschel1988geometric,
  author    = {Gr{\"o}tschel, Martin and Lov{\'a}sz, L{\'a}szl{\'o} and Schrijver, Alexander},
  title     = {Geometric Algorithms and Combinatorial Optimization},
  publisher = {Springer},
  address   = {Berlin, Germany},
  year      = {1988}
}

\clearpage
\appendix
\section{Assumption ledger}\label{app:assumptions}

Table~\ref{tab:conditions} records where each condition is in force. The ledger below records what each one buys and what changes when it is relaxed.

\begin{table}[t]
\centering\small
\caption{Which conditions are in force where. A dot means the condition is assumed in every result of that section.}
\label{tab:conditions}
\setlength{\tabcolsep}{3pt}
\resizebox{\linewidth}{!}{%
\begin{tabular}{lccccc}
\toprule
 & \S\ref{sec:agent} agent & \S\ref{sec:firm} firm & \S\ref{sec:computation} comput. & \S\ref{sec:welfare} welfare & \S\ref{sec:risk-tuning} tuning \\
\midrule
Committed, observed, support-based audit & $\bullet$ & $\bullet$ & $\bullet$ & $\bullet$ & $\bullet$ \\
\textnormal{(Lin)} & Thm.~\ref{thm:portfolio} (last part), Lem.~\ref{lem:deterrence-frontier} & Thm.~\ref{thm:rentsuff} & $\bullet$ & $\bullet$ & $\bullet$ \\
\textnormal{(SC)} & shared-coordinate mix & -- & -- & -- & -- \\
\textnormal{(W$>$0)} & $\bullet$ & $\bullet$ & $\bullet$ & $\bullet$ & $\bullet$ \\
Finitely many types & -- & -- & $\bullet$ & -- & -- \\
Continuum of rents or risk types & -- & Thms.~\ref{thm:auditdesign}, \ref{thm:quantile-audit} & -- & Thm.~\ref{thm:dominance} & -- \\
Audit cost $K(p)$ explicit & -- & $\bullet$ & budget $B$ & Thm.~\ref{thm:price} & $c_ap$ \\
False positives on honest agents & -- & -- & -- & Thm.~\ref{thm:price} & $\bullet$ \\
\bottomrule
\end{tabular}}
\end{table}

\begin{description}
\item[Committed, observed audit.] Buys: deterrence is a Stackelberg threshold $p_j\Delta\ge D_j$, and the priced rent $p_j\Delta$ is the commitment-value upper bound on deterrence (Section~\ref{sec:disc}). Relaxed: hidden or uncommitted inspection becomes an inspection game with mixed strategies; agents respond to a belief $\hat p_j\le p_j$, which weakly relaxes the frontier to $\hat p_j\Delta\ge D_j$.
\item[Support-based detection.] Buys: audit is an extensive-margin charge and intensities solve ordinary covering programs (Theorem~\ref{thm:portfolio}). Relaxed: magnitude-dependent detection enters the first-order conditions and can make split gaming optimal (Proposition~\ref{prop:magnitude}); the factorization of Theorem~\ref{thm:rentsuff} then holds approximately, with error vanishing in the magnitude sensitivity (Proposition~\ref{prop:neardecouple}).
\item[Severity cap $\Delta\le\Delta_{\max}$.] Buys: a well-posed composition problem in Section~\ref{sec:risk-tuning}. Relaxed: a firm that bears no false-accusation harm drives $p\to0$ and $\Delta\to\infty$ along $p\Delta=D$ (Proposition~\ref{prop:composition}).
\item[\textnormal{(Lin)}.] Buys: a singleton gaming support and the rent vector $D_j$ of \eqref{eq:rent}. Relaxed: convex costs preserve the support/intensity separation under support-based audit, but the optimal support need not collapse to the same price list.
\item[Additive firm objective \eqref{eq:firm-objective}.] Buys: the rent profile is a sufficient statistic for joint design (Theorem~\ref{thm:rentsuff}). Relaxed: capacity, congestion, or quota objectives couple types and require a richer state than $\nu_{w,\theta}$.
\item[Fixed-channel regime.] Buys: water-filling is a global allocation rule over the audited interval (Theorem~\ref{thm:auditdesign}). Relaxed: with channel switching the same equations hold locally, with best responses and rent densities recomputed after each marginal update; with overlapping channels the exact problem is NP-hard (Theorem~\ref{thm:hardness}).
\item[$\Delta_t$ independent of $D$.] Buys: a uniform audit deters a clean upper tail of risk differentials (Theorem~\ref{thm:quantile-audit}). Relaxed: correlated risk and recourse costs require a joint law over $(D,\Delta)$ and a two-dimensional allocation rule.
\item[Perfect-audit benchmark.] Buys: the closed-form dividend $W'-W=\mu_g$ (Theorem~\ref{thm:dominance}). Relaxed: false positives and audit costs erode the dividend and create the over-penalization trap (Proposition~\ref{prop:overpen}).
\item[Shared domain primitives.] Buys: opposite observed behavior refutes every cost-only model with the same primitives (Theorem~\ref{thm:domainsep}). Relaxed: unmodeled cost or causal differences can also explain the split, and the identification claim disappears.
\item[Local exposedness and openness \textnormal{(R1)}-\textnormal{(R4)}.] Buys: a pointwise audit improvement becomes a positive-measure frontier gap (Theorem~\ref{thm:dominance}). Relaxed: the benchmark lift survives but not the open reachable region.
\end{description}

\section{Agent-side proofs}\label{app:agent}

\subsection{Proof of Theorem~\ref{thm:portfolio}}

Any non-covering action yields rejection and weakly negative utility, so abstention weakly dominates it. A covering action with no gaming has best value $u_+-\PhiI$. A covering action with gaming support $S$ has detection probability $P_S$ independent of the effort levels, hence best value $u_+-P_S\Delta-\Phi(S)$. Maximizing over $S$ gives $u_+-\PhiM(\pi)$ by \eqref{eq:PhiM-decomp}. Taking the maximum over abstention, honest cover, and cover-with-gaming proves the value expression and the gaming criterion.

On a fixed support, the optimization defining $\Phi(S)$ is a convex covering program. The score constraint binds, since otherwise some positive effort can be reduced. With multiplier $\lambda>0$, stationarity gives
\[
c^I_j\phi'_j(\delta_j^\star)=\lambda w_j,\qquad
c^G_j\psi'_j(\gamma_j^\star)=\lambda w_j
\]
on active coordinates; the audit term is constant on $S$, so it does not enter these first-order conditions.

Under \textnormal{(Lin)}, $\Phi(S)$ is a linear program with one covering constraint. There exists an optimum loading one cheapest active coordinate. If a support contains coordinates not used by such an optimum, deleting them preserves movement cost and weakly lowers detection risk. Thus there exists an optimal gaming support that is a singleton, up to ties and zero-audit redundant coordinates. For singleton $\{j\}$, the total risk-adjusted covering price is $g_0c^G_j/w_j+p_j\Delta$, proving the resulting equation. \qed

\subsection{Proof of Lemma~\ref{lem:deterrence-frontier}}

Under \textnormal{(Lin)} and \textnormal{(W$>$0)}, by Theorem~\ref{thm:portfolio} the agent's non-gaming value is $\max\{0,u_+-\PhiI\}=u_+-\min\{\PhiI,u_+\}$, and its best value from gaming on $j$ is $u_+-g_0c^G_j/w_j-p_j\Delta$. Gaming on $j$ is weakly deterred exactly when the latter is at most the former,
\[
g_0c^G_j/w_j+p_j\Delta\ \ge\ \min\{\PhiI,u_+\},
\]
which rearranges to $p_j\Delta\ge D_j$ with $D_j$ as in \eqref{eq:rent}. When $\PhiI\le u_+$ the fallback is honest improvement and $D_j=\PhiI-g_0c^G_j/w_j$; otherwise the fallback is abstention and $D_j=u_+-g_0c^G_j/w_j$. The substitution rates follow from differentiating $p_j\Delta=D_j$ along the frontier: $dD_j=\Delta\,dp_j+p_j\,d\Delta$. \qed

\subsection{Shared-coordinate formula}

For a shared active coordinate under costs $\phi_j(t)=\psi_j(t)=t^a/a$ with $a>1$, the fake-to-improve ratio on an active support is
\begin{equation}\label{eq:mix-power}
\frac{\gamma^\star_j}{\delta^\star_j}=\left(\frac{c^I_j}{c^G_j}\right)^{1/(a-1)} .
\end{equation}
The fixed-support first-order conditions give
\[
c^I_j(\delta_j^\star)^{a-1}=\lambda w_j,\qquad
c^G_j(\gamma_j^\star)^{a-1}=\lambda w_j .
\]
Dividing and taking the $(a-1)$st root gives \eqref{eq:mix-power}. If $c^G_j<c^I_j$, then the fake share is increasing in $(c^I_j/c^G_j)^{1/(a-1)}$ and this expression decreases to $1$ as $a\to\infty$, giving convergence to an even split. As $a\to1^+$ it diverges, giving the near-linear gaming corner.

\begin{figure}[t]
\centering
\includegraphics[width=\linewidth]{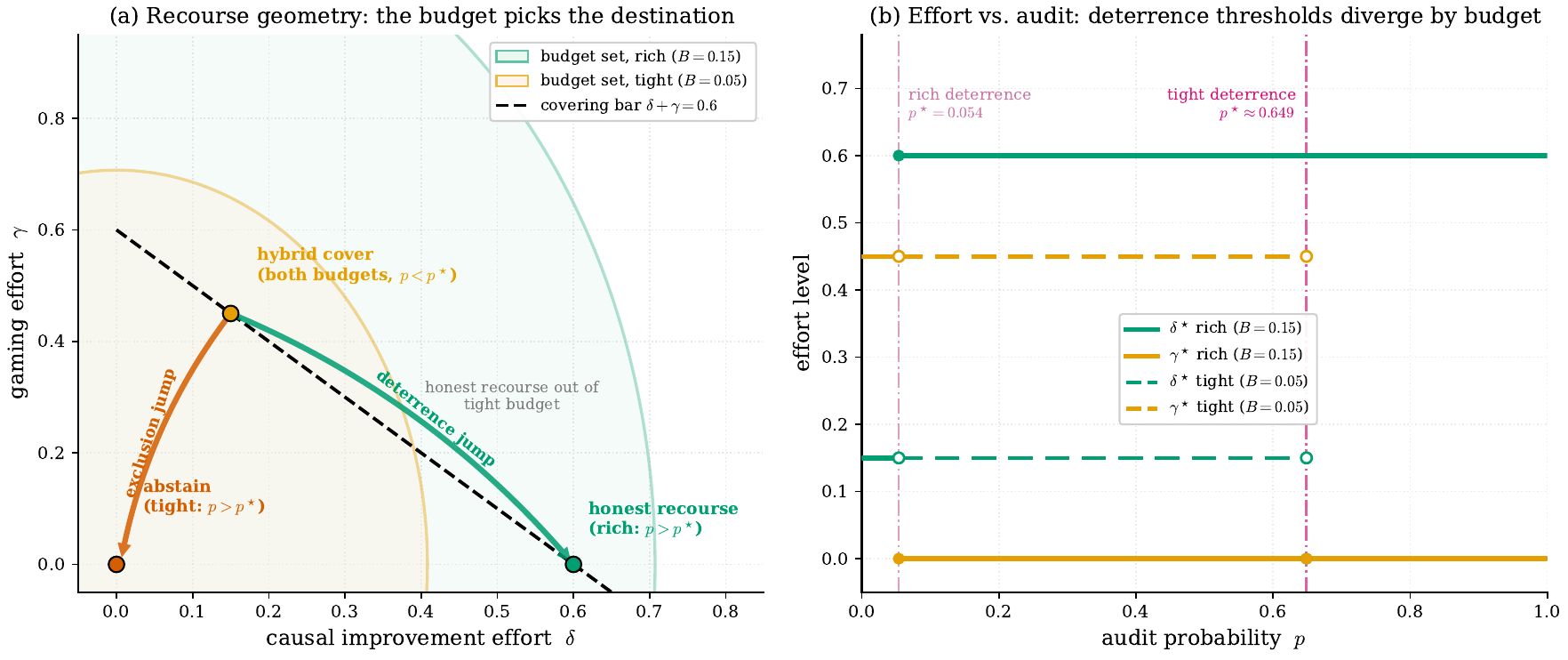}
\Description{See caption.}
\caption{Shared-coordinate mixing. On an active shared coordinate, convex costs create an interior improve/fake mix; audit changes the active support and can drop the fake coordinate discontinuously.}
\label{fig:shared-coordinate-mixing}
\end{figure}

\subsection{Generation protocol for Figure~\ref{fig:shared-coordinate-mixing}}

Figure~\ref{fig:shared-coordinate-mixing} is generated by \texttt{fig\_g\_shared\_coordinate\_mixing.py}. It is a deterministic calculation, not a Monte Carlo simulation. The script fixes one shared coordinate with quadratic improvement and gaming costs, $C_I(\delta)=\frac12c_I\delta^2$ and $C_G(\gamma)=\frac12c_G\gamma^2$, using $c_I=0.6$, $c_G=0.2$, score gap $g=0.6$, acceptance utility $u_+=1$, and caught-state differential $\Delta=1.5$. The unconstrained hybrid cover is computed from the first-order conditions,
\[
\delta^\star=\frac{gc_G}{c_I+c_G}=0.15,\qquad
\gamma^\star=\frac{gc_I}{c_I+c_G}=0.45.
\]
Its movement cost is $C_G^{\mathrm{hyb}}=0.027$, while the all-honest cover $\delta=g$ costs $C_I^{\mathrm{hon}}=0.108$.

The figure overlays two budget regimes. The ``rich'' agent has budget $B=0.15$, so honest recourse is feasible once audit deters the fake component; its deterrence threshold is $(C_I^{\mathrm{hon}}-C_G^{\mathrm{hyb}})/\Delta=0.054$. The ``tight'' agent has budget $B=0.05$, so honest recourse is infeasible after gaming is deterred; its threshold is instead $(u_+-C_G^{\mathrm{hyb}})/\Delta\simeq0.649$, after which the agent abstains. Panel (a) draws the two quadratic budget ellipses, the covering line $\delta+\gamma=0.6$, the common hybrid point, and the post-deterrence jump destinations. Panel (b) plots the piecewise-constant effort trajectories in $p$: solid lines for the rich agent, dashed lines for the tight agent, green for causal effort, and orange for fake effort. The purpose is to isolate the extensive-margin role of audit: the hybrid intensity is fixed by costs, while audit decides when the fake component exits the support.

\subsection{Magnitude-dependent audit}

\begin{prop}[Magnitude-dependent audit couples the intensive margin]\label{prop:magnitude}
Suppose detection is magnitude-dependent, $P(\gamma)=1-\prod_{j\in S(\gamma)}(1-p_j(\gamma_j))$ with each $p_j$ nondecreasing, differentiable on $\gamma_j>0$, and $p_j(0)=0$. On a fixed active support the gaming first-order condition acquires an audit term,
\begin{equation}\label{eq:magnitude-foc}
c^{G}_{j}\psi'_{j}(\gamma_j^\star)+\Delta\,p_j'(\gamma_j^\star)\!\!\prod_{i\in S,\,i\ne j}\!\!\big(1-p_i(\gamma_i^\star)\big)=\lambda w_j ,
\end{equation}
so the support/intensity separation of Theorem~\ref{thm:portfolio} - and hence the factorization of Theorem~\ref{thm:rentsuff} - is preserved in general only by magnitude-independent detection ($p_j'\equiv0$ on the interior). Moreover, even under \textnormal{(Lin)} the optimal gaming support need not be a singleton: with two equal-cost channels and a symmetric convex $p(\cdot)$, splitting the fake strictly dominates concentrating it.
\end{prop}

\subsubsection*{Proof of Proposition~\ref{prop:magnitude}}

\emph{First-order condition.} Fix a support $S$ and consider $\min_{\delta,\gamma\ge0}C_I(\delta)+C_G(\gamma)+\Delta P(\gamma)$ subject to $\inn{w}{\delta+\gamma}\ge g_0$ and $\mathrm{supp}(\gamma)\subseteq S$. On the interior $P(\gamma)=1-\prod_{j\in S}(1-p_j(\gamma_j))$ is differentiable, and the product rule gives
\[
\frac{\partial P}{\partial\gamma_j}=p_j'(\gamma_j)\!\!\prod_{i\in S,\,i\ne j}\!\!\big(1-p_i(\gamma_i)\big)\ge0 .
\]
With multiplier $\lambda$ on the binding score constraint, stationarity on an active gaming coordinate is $c^G_j\psi'_j(\gamma_j^\star)+\Delta\,\partial P/\partial\gamma_j=\lambda w_j$, which is \eqref{eq:magnitude-foc}; the improvement coordinates retain $c^I_j\phi'_j(\delta_j^\star)=\lambda w_j$. If $p_j$ is the step $\bar p_j\indic\{\cdot>0\}$ then $p_j'\equiv0$ on $\gamma_j>0$, the audit term vanishes, and the support-level charge of \eqref{eq:PhiS}--\eqref{eq:PhiM-decomp} is recovered; otherwise the audit enters the intensities and the cost-only covering FOCs no longer characterize the mix.

\emph{Non-singleton optimum.} Take two purely-gameable channels with $w_1=w_2=1$, $g_0=1$, equal linear unit cost $c^G_1=c^G_2=c$, and symmetric $p_1(\cdot)=p_2(\cdot)=p$ with $p(\gamma)=\min(a\gamma^2,1)$, $a\in(0,1]$. Since movement cost is linear it equals $c(\gamma_1+\gamma_2)$, minimized on the binding constraint $\gamma_1+\gamma_2=1$ (raising the total only adds cost and detection). On that segment detection is $1-(1-p(\gamma_1))(1-p(\gamma_2))$, so minimizing detection maximizes $(1-p(\gamma_1))(1-p(\gamma_2))$. Because $p$ is convex, $1-p$ is concave, $\log(1-p)$ is concave, and $\gamma\mapsto\log(1-p(\gamma))$ is concave; the symmetric point $\gamma_1=\gamma_2=\tfrac12$ maximizes the sum of logs on the segment. There the detection is $1-(1-a/4)^2=\tfrac{a}{2}-\tfrac{a^2}{16}<a=p(1)$, the singleton detection at equal movement cost. Hence the split strictly dominates and the optimal support is $\{1,2\}$. The gap is structural: $1-\prod_j(1-p_j)=\sum_j p_j-\sum_{i<j}p_ip_j+\cdots$ is submodular in the $p_j$, so spreading detection across channels is discounted. \qed

\section{Firm-side proofs}\label{app:firm}

\subsection{Proof of Theorem~\ref{thm:auditdesign}}

The fixed-channel assumption means the channel-$j$ rent density in \eqref{eq:cum-deterrence} is a valid state variable over the stated audit interval: raising $p_j$ moves the deterrence cutoff through that density rather than inducing reassignment to other gaming channels. Under that maintained regime, the proof is a concave resource-allocation calculation.

For $0<p_j<\bar p_j$, the value function \eqref{eq:cum-deterrence} has derivative
\[
V_j'(p_j)=\Delta\,a_j(p_j\Delta)m_j(p_j\Delta)-\ell_j'(p_j),
\]
where the derivative is one-sided at the cap. Since $d\mapsto a_j(d)m_j(d)$ is nonincreasing and $\ell_j+k_j$ is convex, each summand $V_j(p_j)-k_j(p_j)$ is concave on its feasible interval. Hence KKT conditions are necessary and sufficient, and the optimal value is unique. The program separates across channels, so uniqueness of the allocation is decided channel by channel: if $a_jm_j$ is strictly decreasing or $\ell_j+k_j$ is strictly convex, the summand for channel $j$ is strictly concave and its maximizer is unique; if channel $j$ has a flat marginal value and linear costs on an interval, every point of that interval is optimal for that channel regardless of what holds elsewhere. Strictness on one channel therefore does not deliver uniqueness of the whole allocation, which requires strictness on every active channel.

For an interior active channel, stationarity is $V'_j(p_j^\star)-k'_j(p_j^\star)=0$, giving \eqref{eq:waterfilling-foc}. If $V'_j(0)\le k'_j(0)$, the channel is dark; if $V'_j(\bar p_j)\ge k'_j(\bar p_j)$, it saturates. For the shared-budget variant, maximizing $\sum_j V_j(p_j)$ subject to $\sum_j k_j(p_j)\le B_a$, the Lagrangian is
\[
\sum_j V_j(p_j)-\lambda\sum_j k_j(p_j)+\lambda B_a ,
\]
so interior active channels satisfy $V'_j(p_j^\star)=\lambda k'_j(p_j^\star)$; equivalently, they equalize net marginal deterrence value per marginal dollar,
\begin{equation}\label{eq:waterfilling-budget}
\frac{\Delta\,a_j(p_j^\star\Delta)m_j(p_j^\star\Delta)-\ell_j'(p_j^\star)}
{k_j'(p_j^\star)}=\lambda .
\end{equation} If $k_j(p)=\beta_jp$, the marginal value per dollar is $[\Delta a_j(p_j\Delta)m_j(p_j\Delta)-\ell_j'(p_j)]/\beta_j$, so an optimal greedy allocation spends each next dollar on the channel with highest current net marginal value per dollar. Because this value declines as $p_j$ rises, several channels can be partially funded; the KKT equalization condition remains the correct characterization. \qed

\subsection{Proof of Corollary~\ref{cor:outcome}}

\begin{proof}
Write $\mu_j(d):=a_j(d)m_j(d)$ and $M_j(x):=\int_x^{\bar D_j}\mu_j(d)\,dd$. From \eqref{eq:cum-outcome}, for $0<p<\bar p_j$,
\begin{align*}
\frac{d}{dp}V^{\mathrm{out}}_j(p)&=\Delta\mu_j(p\Delta)+M_j(p\Delta)-p\Delta\mu_j(p\Delta)-\ell_j'(p)\\
&=(1-p)\Delta\mu_j(p\Delta)+M_j(p\Delta)-\ell_j'(p),
\end{align*}
which is \eqref{eq:outcome-foc} once set equal to $k_j'(p)$. Where $\mu_j$ is differentiable,
\begin{align*}
\frac{d^2}{dp^2}V^{\mathrm{out}}_j(p)&=-\Delta\mu_j(p\Delta)+(1-p)\Delta^2\mu_j'(p\Delta)-\Delta\mu_j(p\Delta)-\ell_j''(p)\\
&=-2\Delta\mu_j(p\Delta)+(1-p)\Delta^2\mu_j'(p\Delta)-\ell_j''(p)\le0,
\end{align*}
since $\mu_j\ge0$, $\mu_j'\le0$, $p\le1$, and $\ell_j$ is convex; at a downward jump of $\mu_j$ the first derivative drops, which preserves concavity. So $V^{\mathrm{out}}_j-k_j$ is concave and the KKT conditions are again necessary and sufficient.

For the comparison, both $V_j-k_j$ and $V^{\mathrm{out}}_j-k_j$ are concave in $p$, so the outcome optimum lies weakly above an interior deterrence optimum $p^{\mathrm{det}}$ iff the outcome objective is nondecreasing at $p^{\mathrm{det}}$. At that point the deterrence first-order condition gives $\Delta\mu_j(p^{\mathrm{det}}\Delta)=\ell_j'+k_j'$, so
\[
\frac{d}{dp}\big[V^{\mathrm{out}}_j-k_j\big](p^{\mathrm{det}})
=M_j(p^{\mathrm{det}}\Delta)-p^{\mathrm{det}}\Delta\mu_j(p^{\mathrm{det}}\Delta),
\]
whose sign is \eqref{eq:outcome-direction}. The uniform-audit statement is the same computation with $U^{\mathrm{out}}(p)=v[1-G(D/p)]+vpG(D/p)-c_ap$: differentiating,
\[
U^{\mathrm{out}\prime}(p)=v(1-p)h(D/p)D/p^2+vG(D/p)-c_a ,
\]
and at the deterrence optimum $p^\star$, where $vh(D/p^\star)D/(p^\star)^2=c_a$, this equals $vG(D/p^\star)-c_ap^\star$.
\end{proof}

\subsection{Proof of Theorem~\ref{thm:quantile-audit}}

For $p>0$, type $t$ is deterred iff $p\Delta_t\ge D$, equivalently $\Delta_t\ge D/p$. Thus the deterred mass is $1-G(D/p)$ and the objective is \eqref{eq:Uquantile} on the feasible interval. Differentiating gives
\[
U'(p)=v h(D/p)\frac{D}{p^2}-c_a .
\]
The stated single-crossing assumption makes this derivative cross zero once from positive to negative on the interior interval, so any interior maximizer is unique and satisfies the first-order condition. If the root falls outside the feasible interval, quasiconcavity puts the optimum at the corresponding boundary. Setting $\Delta_m=D/p^\star$ gives the equivalent marginal-type equation.

For robustness, only the term $v[1-G(D/p)]$ depends on the distribution. Therefore for every $p$,
\[
|U_G(p)-U_{\hat G}(p)|\le v\|G-\hat G\|_\infty .
\]
Let $p^\star$ maximize the true objective and $\hat p$ maximize the estimated objective. Then
\begin{align*}
U_G(p^\star)-U_G(\hat p)
&\le [U_G(p^\star)-U_{\hat G}(p^\star)]
+[U_{\hat G}(\hat p)-U_G(\hat p)]\\
&\le 2v\|G-\hat G\|_\infty ,
\end{align*}
where optimality of $\hat p$ makes the omitted middle term nonpositive. \qed

\subsection{Proof of Theorem~\ref{thm:rentsuff}}

Fix $(w,\theta,\pi)$ and write $\xi_t=\big(\PhiI(t),(D_{j,t})_{j\in\Gset},\Delta_t,y_t,y^I_t\big)\in\Xi$ for the type's rent coordinates, whose law under $\rho$ is the rent profile $\nu_{w,\theta}$. A type with $g_{0,t}\le0$ is accepted and inert, contributing a label-only constant absorbed into $\Psi$; fix a strategic type $g_{0,t}>0$.

\emph{Step 1 (best response is a function of $\xi_t$ and $\pi$).} By Theorem~\ref{thm:portfolio} under \textnormal{(Lin)} the agent's value is the largest of three branch values, which we write directly in the rent coordinates:
\[
A:=0,\ \ H(\xi):=u_+-\PhiI,\ \ G(\xi,\pi):=(u_+-\PhiI)+\max_{j\in\Gset}\big(D_j-p_j\Delta\big),
\]
using $g_{0,t}c^G_j/w_j=\PhiI(t)-D_{j,t}$ to rewrite the gaming value as $(u_+-\PhiI)+(D_j-p_j\Delta)$. Each of $A,H,G$ is continuous - indeed piecewise affine - in $\xi$ for fixed $\pi$, hence Borel measurable on $\Xi$, and the maximizing branch (and, within $G$, the maximizing channel $j^\star$) is a measurable selection. Thus the agent's action class $\alpha(\xi,\pi)\in\{\mathrm{abstain},\mathrm{honest},\,\mathrm{game}\text{-}j^\star\}$ depends on the classifier only through $\xi$.

\emph{Step 2 (the indifference set is null under a non-degenerate population).} Let $T_\pi$ be the set of $\xi\in\Xi$ at which two of $A,H,G$ attain the maximum, or two channels attain the inner $\max_{j}(D_j-p_j\Delta)$.
$T_\pi$ is a finite union of the affine slices $\{A=H\}$, $\{A=G\}$, $\{H=G\}$, and $\{D_i-p_i\Delta=D_j-p_j\Delta\}$ for $i\ne j$, each of codimension at least one in the $\big(\PhiI,(D_j)_j,\Delta\big)$ coordinates. We assume the population is \emph{non-degenerate}: $\nu_{w,\theta}$ assigns zero mass to every such indifference hyperplane (equivalently $\rho(\{t:\xi_t\in T_\pi\})=0$). We flag explicitly that mere atomlessness of $\rho$ is \emph{not} sufficient for this - an atomless law can still concentrate on a lower-dimensional affine set - whereas non-degeneracy holds whenever the law of $\big(\PhiI(t),(D_{j,t})_j,\Delta_t\big)$ is absolutely continuous in those coordinates, the maintained regularity in our running specifications. On $\Xi\setminus T_\pi$ the maximizing branch and channel are unique, so $\alpha(\cdot,\pi)$ is single-valued there.

\emph{Step 3 (tie-breaking is immaterial).} On the null set $T_\pi$ we fix the Stackelberg leader-favorable selection: among the agent's maximizers, $\alpha$ returns the action the firm most prefers. This is a measurable selection (a fixed priority order on the finite action set, intersected with the measurable $\arg\max$), and it makes the leader's supremum attained. Because $\rho(T_\pi)=0$, any other measurable tie-break agrees with it $\rho$-a.e.\ and yields the same value of every integral below; the leader-favorable choice is for definiteness only.

\emph{Step 4 (the per-type kernel is bounded and measurable).} The firm's realized per-type payoff is determined by the action class together with the labels carried in $\xi$: an abstainer is scored by $y_t$, an honest improver by the post-improvement label $y^I_t$ (two classifiers with equal $\PhiI(t)$ may load different improvable coordinates and so differ in $y^I_t$, which is exactly why $y^I_t$ is carried in $\xi_t$ and is audit-independent), an accepted gamer is a false positive (gaming leaves the true label unchanged), and a caught gamer is a correct rejection. Only the audit coin is stochastic, and only on the played channel $j^\star$: gaming yields the correct-rejection payoff with probability $p_{j^\star}$ and the false-positive payoff with probability $1-p_{j^\star}$. Writing the finitely many outcome payoffs as bounded constants, define
\[
f(\xi,\pi):=
\begin{cases}
F_{\mathrm{abs}}(y) & \alpha(\xi,\pi)=\mathrm{abstain},\\
F_{\mathrm{imp}}(y^I) & \alpha(\xi,\pi)=\mathrm{honest},\\
p_{j^\star}\,r_{\mathrm{caught}}+(1-p_{j^\star})\,r_{\mathrm{FP}} & \alpha(\xi,\pi)=\mathrm{game}\text{-}j^\star .
\end{cases}
\]
As a finite sum $\sum_{a}\indic\{\alpha(\xi,\pi)=a\}\cdot(\text{branch value})$ of products of Borel functions (indicators of the measurable action cells times continuous branch values), $f(\cdot,\pi)$ is Borel; and since all outcome payoffs lie in a bounded range $[\underline r,\overline r]$ with $\bar a:=\overline r-\underline r<\infty$, $f$ is bounded, $|f|\le\max(|\underline r|,|\overline r|)$. By construction the audit coin is the only randomness left, so $\E\big[F(\mathrm{out}_t,y_t)\mid t\big]=f(\xi_t,\pi)$ for every strategic type, and the same identity holds for inert types with the label-only branch. Therefore $\E[F(\mathrm{out}_t,y_t)]=f(\xi_t,\pi)$ for $\rho$-a.e.\ $t$.

\emph{Step 5 (integration).} By additive separability of the firm objective,
\[
\Pio(w,\theta,\pi)=\int f(\xi_t,\pi)\,d\rho(t)=\int f(\xi,\pi)\,d\nu_{w,\theta}(\xi)=:\Psi(\nu_{w,\theta},\pi),
\]
since the integrand depends on $t$ only through $\xi_t$, whose law is $\nu_{w,\theta}$. This is rent-sufficiency, and \eqref{eq:two-stage} is the chained maximization. Exact decoupling requires $\arg\max_\pi\Psi(\nu_{w,\theta},\pi)$ to be invariant in $(w,\theta)$, i.e. $\nu_{w,\theta}$ constant; but $\nu_{w,\theta}$ is built from $g_{0,t}=\theta-\inn{w}{x^0_t}$, which moves with every weight and with $\theta$, so rent-invariant directions are non-generic. \qed

\subsection{Approximate rent-sufficiency under near-support detection}\label{app:neardecouple}

The factorization of Theorem~\ref{thm:rentsuff} is stated under support-based detection, and Proposition~\ref{prop:magnitude} shows it can fail under magnitude-dependent detection. We now show the failure is graceful: near the support-based regime the factorization degrades continuously, and the firm value tracks its rent-profile surrogate to first order in the magnitude sensitivity.

Index magnitude-dependent detection by a sharpness $\beta\in(0,\infty]$ and write
\begin{equation}\label{eq:sharp-family}
p_j(\gamma_j;\beta)=\bar p_j\,\sigma(\beta\, w_j\gamma_j/g_0),\qquad \bar p_j:=p_j(g_0/w_j;\infty),
\end{equation}
where $\sigma:[0,\infty)\to[0,1]$ is nondecreasing with $\sigma(0)=0$, $\sigma(x)\uparrow1$ as $x\to\infty$, and $\sigma(x)=1$ for $x\ge1$ (full single-channel cover is detected at the nameplate rate $\bar p_j$). The exemplar is $\sigma(x)=1-e^{-x^2}$ of Figure~\ref{fig:robustness}(b). The magnitude sensitivity is $1/\beta$: as $\beta\to\infty$, $p_j(\gamma_j;\beta)\to\bar p_j$ for every fixed $\gamma_j>0$ and \eqref{eq:sharp-family} recovers the support-based charge $P_S=1-\prod_{j\in S}(1-\bar p_j)$. Let $d:=|\Gset|$ be the gameable dimension.

\begin{lem}[No vanishing fake; detection floor]\label{lem:floor}
If a fake covers score mass $\inn{w}{\gamma}\ge\tau>0$, then some coordinate $k$ has $w_k\gamma_k\ge\tau/d$, and its detection satisfies $P_S(\gamma)\ge p_k(\gamma_k;\beta)\ge\bar p_k\,\sigma(\beta\tau/(dg_0))$. In particular, a fake that covers the whole gap on its own (the pure-gaming case $\tau=g_0$) places a coordinate at normalized magnitude $\ge1/d$, so
\begin{equation}\label{eq:detfloor}
P_S(\gamma)\ \ge\ \bar p_k\,\sigma(\beta/d)\ \xrightarrow[\beta\to\infty]{}\ \bar p_k .
\end{equation}
\end{lem}
\begin{proof}
If every coordinate had $w_j\gamma_j<\tau/d$, then $\inn{w}{\gamma}=\sum_{j}w_j\gamma_j<d\cdot\tau/d=\tau$, a contradiction. For the floor channel, $w_k\gamma_k/g_0\ge\tau/(dg_0)$ and $\sigma$ nondecreasing give $p_k(\gamma_k;\beta)=\bar p_k\sigma(\beta w_k\gamma_k/g_0)\ge\bar p_k\sigma(\beta\tau/(dg_0))$; since $P_S(\gamma)=1-\prod_j(1-p_j)\ge p_k(\gamma_k;\beta)$, the bounds follow, and $\tau=g_0$ gives \eqref{eq:detfloor}.
\end{proof}

The lemma is the mechanism behind the robustness: at a \emph{fixed} gameable dimension $d$, an agent cannot dilute a fake into undetectability by spreading it, because covering the score gap forces at least one coordinate to a magnitude $\ge g_0/(dw_k)$, at which sharp detection is near the nameplate rate. Diversification is the only channel through which magnitude dependence breaks Theorem~\ref{thm:portfolio} (Proposition~\ref{prop:magnitude}), and \eqref{eq:detfloor} caps the gain it can deliver.

\begin{prop}[Near-decoupling]\label{prop:neardecouple}
Assume \textnormal{(Lin)}, \textnormal{(W$>$0)}, an additively separable firm objective, a non-degenerate population in the sense of Theorem~\ref{thm:rentsuff} (the law of the rent coordinates charges no agent-indifference hyperplane; atomlessness alone is not enough, since behavior can jump on such a hyperplane and an atomless law can still concentrate mass there), and the detection family \eqref{eq:sharp-family}. Let $\Pio^{\beta}(w,\theta,\pi)$ be the firm value under sensitivity $1/\beta$ and let $\Psi(\nu_{w,\theta},\pi)$ be the support-based surrogate of Theorem~\ref{thm:rentsuff}. Then for $\rho$-a.e. type the gaming value $V_\beta(t)$ converges up to its support-based value $V_\infty(t)$ as $\beta\to\infty$, and
\begin{equation}\label{eq:neardecouple}
\Pio^{\beta}(w,\theta,\pi)\ \xrightarrow[\beta\to\infty]{}\ \Psi(\nu_{w,\theta},\pi)\qquad\text{for every }(w,\theta,\pi).
\end{equation}
Equivalently, rent-sufficiency holds approximately, with an error that vanishes in the magnitude sensitivity $1/\beta$; exact rent-sufficiency (Theorem~\ref{thm:rentsuff}) is the $\beta=\infty$ specialization.
\end{prop}
\begin{proof}
Fix a strategic type $t$ and work on its compact feasible set $K_t=\{(\delta,\gamma)\ge0:\inn{w}{\delta+\gamma}\ge g_{0,t}\}$ intersected with the box $\{\text{movement cost}\le u_+\}$ (beyond which abstention dominates), which is compact under \textnormal{(Lin)}. The objective $J_\beta(\delta,\gamma)=C_I(\delta)+C_G(\gamma)+\Delta_t P_S(\gamma;\beta)$ has $V_\beta(t)=u_+-\min_{K_t}J_\beta$ on the gaming branch.

\emph{Monotone direction.} Since $\sigma\le1$, $p_j(\gamma_j;\beta)\le\bar p_j$ for $\gamma_j>0$, so $P_S(\cdot;\beta)\le P_S(\cdot;\infty)$ pointwise and $J_\beta\le J_\infty$, giving $V_\beta(t)\ge V_\infty(t)$.

\emph{Upper limit.} Suppose, for contradiction, $\limsup_\beta V_\beta(t)\ge V_\infty(t)+2\eta$ for some $\eta>0$ along $\beta_n\to\infty$ with minimizers $a_n=(\delta_n,\gamma_n)\in K_t$. By compactness pass to $a_n\to a^\star=(\delta^\star,\gamma^\star)\in K_t$, feasible by closedness. Let $S^\star=\mathrm{supp}(\gamma^\star)$, and bound detection below by restricting the product to the non-vanishing coordinates (omitted factors are at most $1$):
\[
P_{S}(\gamma_n;\beta_n)=1-\prod_{j}\big(1-p_j(\gamma_{n,j};\beta_n)\big)\ \ge\ 1-\prod_{j\in S^\star}\big(1-p_j(\gamma_{n,j};\beta_n)\big).
\]
For $j\in S^\star$, $\gamma_{n,j}\to\gamma^\star_j>0$, so $p_j(\gamma_{n,j};\beta_n)=\bar p_j\sigma(\beta_n w_j\gamma_{n,j}/g_{0})\to\bar p_j$, and the right side $\to 1-\prod_{j\in S^\star}(1-\bar p_j)=P_S(\gamma^\star;\infty)$, with the convention $P_\varnothing=0$ when $S^\star=\varnothing$ (the case in which the fake vanishes in the limit and $a^\star$ is a pure improvement/abstention action). Either way $\liminf_n P_S(\gamma_n;\beta_n)\ge P_S(\gamma^\star;\infty)$. With continuity of $C_I,C_G$,
\begin{align*}
\limsup_n V_{\beta_n}(t)&=u_+-\liminf_n J_{\beta_n}(a_n)\\
&\le u_+-\big(C_I(\delta^\star)+C_G(\gamma^\star)+\Delta_t P_S(\gamma^\star;\infty)\big)\\
&=u_+-J_\infty(a^\star)\le V_\infty(t),
\end{align*}
the last step since $a^\star\in K_t$ is feasible for the support-based problem. This contradicts the assumption, so $V_\beta(t)\to V_\infty(t)$ for every strategic type (and trivially for inert and honest-optimal types, whose values do not depend on $\beta$). Convergence of values does not by itself give convergence of actions: a type whose support-based gaming value ties its honest value can switch behavior discontinuously at every finite $\beta$. This is where non-degeneracy enters. Off the indifference hyperplanes the support-based maximizer is unique and strict, so for $\beta$ large the $\beta$-maximizer lies in the same action class, and the per-type payoff converges; the indifference set is $\rho$-null by assumption. Lemma~\ref{lem:floor} sharpens the rate: whenever the agent games to cover a non-vanishing share $\tau$ of its gap, some coordinate carries detection $\ge\bar p_k\sigma(\beta\tau/(dg_0))$, so the gaming value exceeds its support-based counterpart by at most $\Delta_t\,\bar p_{\max}\big(1-\sigma(\beta\tau/(dg_0))\big)$.

\emph{Integration.} Per-type firm payoffs are bounded by the payoff range $\bar a<\infty$, and the played channel's detection enters the payoff continuously; thus the per-type expected payoff $f_\beta(\xi_t,\pi)\to f_\infty(\xi_t,\pi)=f(\xi_t,\pi)$ for $\rho$-a.e.\ $t$. By bounded convergence,
$\Pio^{\beta}=\int f_\beta\,d\rho\to\int f\,d\rho=\Psi(\nu_{w,\theta},\pi)$ for each fixed classifier. (Convergence is pointwise in $(w,\theta,\pi)$, not uniform: a type that improves to just below the bar and fakes a vanishing residual is caught with probability $\to\bar p_j$ only as $\beta\to\infty$, and the mass of such near-bar gamers depends on the rent profile the classifier induces. The error is nonetheless controlled on the deterrence-relevant mass: a type that fakes to cover a share $\tau$ of its gap contributes a detection discrepancy at most $\bar p_{\max}(1-\sigma(\beta\tau/(dg_0)))$ by Lemma~\ref{lem:floor}.)
\end{proof}

The convergence is fast and the regime where it bites is narrow. In the deterministic simulations (\texttt{fig\_near\_decoupling.py}, with the worst-case heterogeneous instances found by random search) the firm-value gap decays \emph{super-linearly} in $1/\beta$ and the singleton becomes \emph{exactly} optimal at a finite sharpness ($\beta\approx12$ for $d\in\{2,3\}$): once full-cover fakes are reliably detected, no split survives. The bound degrades only when $\beta=O(1)$ - when even full-magnitude manipulations evade detection - which is precisely the regime in which Figure~\ref{fig:robustness}(b) shows the audit dividend itself collapsing. The two-stage design of Theorem~\ref{thm:rentsuff} is therefore not an isolated knife-edge but the $\beta=\infty$ end of a continuum on which it holds approximately. \qed

\section{Computational proofs}\label{app:computation}

Throughout this appendix, \emph{leader-favorable tie-breaking} means that whenever an agent is indifferent among optimal actions (including indifference between two gaming channels), the realized action is one maximizing the firm's payoff; this is the standard Stackelberg selection and is the convention under which weak deterrence inequalities ($p_j\Delta_t\ge D_{j,t}$) deter.

\subsection{Proof of Proposition~\ref{prop:vertex}}

\emph{Covering characterization.} By Theorem~\ref{thm:portfolio} under \textnormal{(Lin)}, type $t$'s best gaming value is
\[
u_+-\min_{j\in\Gset_t}\big(g_{0,t}c^G_{j,t}/w_j+p_j\Delta_t\big),
\]
and its best honest value is $u_+-\PhiI(t)\ge0$ by honest participation. Gaming is (weakly) suboptimal iff $g_{0,t}c^G_{j,t}/w_j+p_j\Delta_t\ge\PhiI(t)$ for every $j$, i.e.\ $p_j\Delta_t\ge D_{j,t}$ for every $j\in J_t$ (channels outside $J_t$ have $D_{j,t}\le0$ and are never strictly profitable).

\emph{Action patterns are constant on cells.} An \emph{action pattern} $\sigma$ assigns each type an action in $\{\text{improve},\text{abstain}\}\cup\{\text{game-}j:j\in\Gset_t\}$. By honest participation, improving is weakly optimal whenever abstaining is, at every $p$; the improve-vs-abstain comparison and the firm's preference between the two are $p$-independent, so we may restrict attention to patterns in which each type's non-gaming action is fixed once and for all (leader-favorably), and the gaming-vs-abstain comparison never binds before the gaming-vs-improve one. The set of audit profiles at which such a pattern $\sigma$ is weakly optimal for every type is
\[
R_\sigma=\bigcap_t\{p:\text{$\sigma(t)$'s value}\ge\text{every alternative's value at }p\},
\]
an intersection of weak inequalities each of which is, in $p$, either trivial, a deterrence half-space $p_j\Delta_t\gtrless D_{j,t}$, or a switch half-space $p_j\Delta_t-p_i\Delta_t\gtrless g_{0,t}(c^G_{i,t}/w_i-c^G_{j,t}/w_j)$. Hence each $R_\sigma$ is a closed polyhedron whose facets lie on hyperplanes of $\mathcal H$, intersected with the box and the budget half-space.

\emph{Affine objective on each $R_\sigma$.} Fix $\sigma$. The deterrence-value objective $\sum_{t:\sigma(t)\ne\text{game}}a_t$ is constant on $R_\sigma$. The outcome objective is
\[
V_\sigma(p)=\sum_{t:\sigma(t)\,\text{honest}}\text{const}_t
\;-\!\!\sum_{t:\sigma(t)=\text{game-}j(t)}\!\!(1-p_{j(t)})\,a_t ,
\]
affine in $p$. (The budget enters as a constraint; with an audit cost in the objective instead, subtract $\sum_j\kappa_jp_j$, which preserves affineness.)

\emph{Optimum at a vertex.} Let $V(p)$ denote the firm value under leader-favorable selection. For any $p$, the selected pattern $\sigma'$ is weakly optimal at $p$, so $p\in R_{\sigma'}$ and $V(p)=V_{\sigma'}(p)$; conversely for any $\sigma$ and $p\in R_\sigma$, $\sigma$ is weakly optimal at $p$, so the leader-favorable value satisfies $V(p)\ge V_\sigma(p)$. Therefore
\[
\sup_{p\ \text{feasible}}V(p)=\max_\sigma\ \max_{p\in R_\sigma\cap\text{feasible}}V_\sigma(p),
\]
a finite maximum of affine functions over compact polytopes; each inner maximum is attained at a vertex of $R_\sigma\cap\{\sum_j\kappa_jp_j\le B\}\cap[0,1]^m$. Every facet of this polytope lies on a hyperplane of $\mathcal H$ (including budget and box), so every vertex is a vertex of the arrangement $\mathcal H$.

\emph{Complexity.} $\mathcal H$ contains at most $Nm$ deterrence, $N\binom m2$ switch, $2m$ box, and one budget hyperplane: $H=O(Nm^2)$. Each vertex is the solution of $m$ independent linear equations chosen among them, so there are at most $\binom Hm=(Nm)^{O(m)}$ candidates; each is checked for feasibility and evaluated (computing every type's best response and the objective) in $O(Nm+m^3)$ time. Enumeration therefore solves \textsc{Audit} exactly in $(Nm)^{O(m)}$ time. \qed

\subsection{Proof of Theorem~\ref{thm:hardness}}

\emph{Reduction.} Let $(G=(V,E),k,M)$ be a Densest $k$-Subgraph instance: does some $S\subseteq V$ with $|S|=k$ induce at least $M$ edges? Build an \textsc{Audit} instance with one channel per vertex ($m=|V|$, $w_j=1$, $\kappa_j=1$) and one type $t_e$ per edge $e=\{u,v\}$: gameable support $\Gset_{t_e}=\{u,v\}$ with unit costs $c^G_{u,t_e}=c^G_{v,t_e}=1$, score deficit $g_{0,t_e}=1$, honest price $\PhiI(t_e)=2$ (one improvable coordinate of unit weight and cost $2$), acceptance utility $u_+=2$ (honest participation holds), differential $\Delta_{t_e}=2$, and value $a_{t_e}=1$. Then $D_{j,t_e}=2-1=1$ on both channels, so by the covering characterization $t_e$ is deterred iff $p_u\ge\tfrac12$ and $p_v\ge\tfrac12$. Set the budget $B=k/2$.

\emph{Value preservation.} If $S$ induces $M'$ edges, the profile $p=\tfrac12\indic_S$ is feasible ($\sum_jp_j=k/2$) and deters exactly the types of edges inside $S$: value $M'$. Conversely, let $p$ be feasible with deterrence value $M'$ and set $S'=\{j:p_j\ge\tfrac12\}$; feasibility gives $|S'|\le B/(1/2)=k$, every deterred type is an edge inside $S'$, and padding $S'$ to size $k$ only adds induced edges. Hence the optimal values of the two instances are \emph{equal}, and the answer to the decision problem is yes iff \textsc{Audit}'s optimum is at least $M$. NP-hardness follows since Densest $k$-Subgraph contains Clique ($M=\binom k2$). \qed

\subsection{Proof of Corollary~\ref{cor:inapprox}}

\emph{Approximation transfer.} Any $\rho$-approximate solution to \textsc{Audit} on the reduced instance converts, by the rounding $S'$ above, into a $k$-set with at least $\mathrm{OPT}/\rho$ induced edges, since the rounding never loses deterred types and the two optima coincide. The PTAS exclusion \cite{khot2006ptas} and the ETH-based ratio \cite{manurangsi2017dks} transfer verbatim; the ratio is stated in the number of channels $m=|V|$.

\emph{Nestedness under a common gaming technology.} Suppose all types share unit costs and weights and let $\tilde c_j:=c^G_j/w_j$. Then $D_{j,t}=\PhiI(t)-g_{0,t}\tilde c_j$ is strictly decreasing in $\tilde c_j$ (as $g_{0,t}>0$), so $J_t=\{j:\tilde c_j<\PhiI(t)/g_{0,t}\}$ is a prefix of the channels sorted by $\tilde c_j$, and any two prefixes are nested. Two distinct edges $\{u,v\}\ne\{u',v'\}$ would require two distinct two-element profitable sets, which are nested only if equal; hence the reduction's instances need heterogeneous gameable supports (or heterogeneous unit costs). \qed

\subsection{Proof of Proposition~\ref{prop:catch-credit}}

\begin{proof}
Take the instance of Theorem~\ref{thm:hardness}, in which $\sum_ta_t=|E|=:N$, and let $\lambda\le1/(2N)$. In that instance both channels of an edge-type have equal cost, so an undeterred edge-type $t_e$ games on the channel with the lower rate and earns credit $\lambda\min\{p_u,p_v\}\le\lambda$. For any feasible $p$ with deterrence value $M'$, the \textsc{Audit}$_\lambda$ value therefore lies in $[M',\,M'+\lambda N]\subseteq[M',\,M'+\tfrac12]$.

Let $\mathrm{OPT}$ be the Densest $k$-Subgraph optimum and $p^\lambda$ an optimal solution of \textsc{Audit}$_\lambda$ with deterrence value $M^\lambda$. The profile $\tfrac12\indic_S$ for an optimal $k$-set $S$ has \textsc{Audit}$_\lambda$ value at least $\mathrm{OPT}$, so $M^\lambda+\tfrac12\ge\mathrm{OPT}$, i.e.\ $M^\lambda\ge\mathrm{OPT}-\tfrac12$, and since both are integers $M^\lambda\ge\mathrm{OPT}$. The rounding $S'=\{j:p^\lambda_j\ge\tfrac12\}$ of the value-preservation step gives $M^\lambda\le\mathrm{OPT}$, so $M^\lambda=\mathrm{OPT}$ and the decision problem reduces as before: NP-hardness of \textsc{Audit}$_\lambda$ follows.

For approximation, a solution with \textsc{Audit}$_\lambda$ value at least $\mathrm{OPT}^\lambda/\rho\ge\mathrm{OPT}/\rho$ has deterrence value at least $\mathrm{OPT}/\rho-\tfrac12$, and rounding yields a $k$-set with that many induced edges. If $\mathrm{OPT}/\rho\ge1$ this is at least $\mathrm{OPT}/(2\rho)$; if $\mathrm{OPT}/\rho<1$, any single edge (which a $k$-set with $k\ge2$ can contain whenever $E\ne\varnothing$) already achieves ratio $\rho$. Hence a $\rho$-approximation for \textsc{Audit}$_\lambda$ gives a $2\rho$-approximation for Densest $k$-Subgraph, which preserves the ETH-based ratio of Corollary~\ref{cor:inapprox} up to the constant.

For $\lambda$ bounded away from zero the argument fails at the first step, since the catch credit can exceed the value of a deterred type. Exhaustive search on random graphs with at most six vertices confirms that at $\lambda=1$ the optimal profile can spread the budget so as to deter no type while the Densest $k$-Subgraph optimum is positive; the complexity of \textsc{Audit}$_1$ on general instances is open.
\end{proof}

\subsection{Proof of Proposition~\ref{prop:softbudget}}

\emph{Lattice reduction.} Write $\theta_{j,t}:=D_{j,t}/\Delta_t$ for $j\in J_t$ and let $L_j:=\{0\}\cup\{\theta_{j,t}:t\text{ with }j\in J_t,\ \theta_{j,t}\le1\}$. For any $p\in[0,1]^m$ let $\underline p_j:=\max\{\ell\in L_j:\ell\le p_j\}$. Every deterrence inequality $p_j\ge\theta_{j,t}$ that holds at $p$ still holds at $\underline p$, because $\theta_{j,t}\in L_j$ and $\theta_{j,t}\le p_j$ imply $\theta_{j,t}\le\underline p_j$; and $\kappa\cdot\underline p\le\kappa\cdot p$. So $\underline p$ is weakly better, and an optimum lies on $L:=\prod_jL_j$, which has $|L_j|\le N+1$ levels per coordinate.

\emph{Supermodularity.} Order $L$ coordinatewise, with $p\vee p'$ and $p\wedge p'$ the coordinatewise maximum and minimum. For a type $t$ with profitable set $J_t$ and threshold vector $\theta_t$, put $\chi_t(p):=\indic\{p_j\ge\theta_{j,t}\ \forall j\in J_t\}$, the indicator of an up-set of $L$. For any $p,p'$: if $\chi_t(p)=\chi_t(p')=1$ then $p\vee p'$ and $p\wedge p'$ both dominate $\theta_t$ on $J_t$, so both sides of $\chi_t(p\vee p')+\chi_t(p\wedge p')\ge\chi_t(p)+\chi_t(p')$ equal $2$; if exactly one of $\chi_t(p),\chi_t(p')$ is $1$ then $p\vee p'$ dominates $\theta_t$ and the left side is at least $1$; if both are $0$ the right side is $0$. Hence each $\chi_t$ is supermodular, so is $\sum_ta_t\chi_t$ with $a_t>0$, and the linear cost is modular. The objective $F(p)=\sum_ta_t\chi_t(p)-\kappa\cdot p$ is supermodular on $L$.

\emph{Polynomial-time maximization.} Encode $p\in L$ by the set $S(p):=\{(j,\ell):\ell\in L_j,\ \ell\le p_j\}$ over the ground set $E:=\{(j,\ell):j\in\Gset,\ \ell\in L_j\setminus\{0\}\}$, $|E|\le Nm$. The map $p\mapsto S(p)$ is a lattice isomorphism onto the family $\mathcal R$ of sets that are down-closed within each chain $L_j$: $S(p\vee p')=S(p)\cup S(p')$ and $S(p\wedge p')=S(p)\cap S(p')$. $\mathcal R$ is closed under union and intersection, i.e.\ a ring family, and $G(S(p)):=-F(p)$ is submodular on it by the previous step. Minimizing a submodular function over a ring family is strongly polynomial in $|E|$ and the cost of a value oracle \cite{schrijver2000sfm,iwata2001sfm,grotschel1988geometric}, and evaluating $F$ costs $O(Nm)$. Hence \textsc{Audit}$^{\mathrm{cost}}$ is solvable in time polynomial in $N$ and $m$.

\emph{Outcome objective.} Under the outcome objective an undeterred type contributes $p_{j(t)}a_t$, where $j(t)$ is the channel of lowest risk-adjusted price. Supermodularity fails: exhaustive checking of random three-channel instances with six types finds lattice pairs $(p,p')$ with $F^{\mathrm{out}}(p\vee p')+F^{\mathrm{out}}(p\wedge p')<F^{\mathrm{out}}(p)+F^{\mathrm{out}}(p')$, because raising one channel's rate can move a gamer onto a less audited channel and lower the catch term discontinuously. The complexity of the cost version under the outcome objective is open. \qed

\subsection{Proof of Theorem~\ref{thm:fixeddim}}

Write $z=(w,\theta,p)$ ranging over the compact box
\[
Z=[w_{\min},w_{\max}]^d\times[\theta_{\min},\theta_{\max}]\times[0,1]^{m},
\]
with $w_{\min}>0$ and $m=|\Gset|\le d$, so $\dim Z=d+1+m\le 2d+1$.

\emph{A polynomial family controlling behavior.} We assume the post-improvement label is given by a semialgebraic qualification rule of constant degree: whether type $t$ improving on coordinate $i$ to the bar becomes qualified is the sign of a polynomial $q_{t,i}(w,\theta)$ of bounded degree (e.g.\ a linear qualification threshold on true features gives, after clearing $w_i>0$, the degree-$2$ polynomial $\inn{v}{x^0_t}w_i+v_ig_{0,t}-\eta w_i$). For each type $t$ consider the polynomials in $z$:
(i) $g_{0,t}(w,\theta)=\theta-\inn w{x^0_t}$ (strategic vs.\ inert; degree $1$);
(ii) for $i,i'\in\Iset_t$: $c^I_{i,t}w_{i'}-c^I_{i',t}w_i$ (which improvable coordinate is cheapest; degree $1$);
(iii) for $i\in\Iset_t$: $u_+w_i-g_{0,t}c^I_{i,t}$ (participation; degree $1$);
(iv) for $i\in\Iset_t$, $j\in\Gset_t$: $g_{0,t}c^I_{i,t}w_j-g_{0,t}c^G_{j,t}w_i-p_j\Delta_t w_iw_j$ (gaming on $j$ vs.\ honest improvement via $i$, after clearing the positive denominators $w_iw_j$; degree $3$);
(v) for $j,j'\in\Gset_t$: $g_{0,t}c^G_{j,t}w_{j'}-g_{0,t}c^G_{j',t}w_j+(p_j-p_{j'})\Delta_t w_jw_{j'}$ -- after clearing denominators, the channel-switch comparison (degree $3$);
(vi) for $j\in\Gset_t$: $u_+w_j-g_{0,t}c^G_{j,t}-p_j\Delta_tw_j$ (gaming vs.\ abstention, needed when participation fails; degree $2$);
(vii) the qualification polynomials $q_{t,i}$.
This family $\mathcal P$ has $P=O(Nd^2)$ polynomials of bounded degree in $\le2d+1$ variables. The sign vector of $\mathcal P$ at $z$ determines, for every type, whether it is strategic, its cheapest improvable coordinate (hence $\PhiI(t)$ and, via (vii), the post-improvement label $y^I_t$), participation, whether gaming beats honest recourse and abstention on each channel, and the ordering of gaming channels - hence the entire population action pattern under any fixed tie-breaking, and in particular the rent profile's combinatorial type in Theorem~\ref{thm:rentsuff}.

\emph{Cell enumeration.} By the Basu-Pollack-Roy bound and algorithm for realizable sign conditions \cite{basu2006algorithms}, the polynomials in $\mathcal P$ realize at most $(P\cdot3)^{O(2d+1)}=(Nd)^{O(d)}$ sign conditions on $Z$, and a list of all realizable sign conditions can be computed in time $(Nd)^{O(d)}$.

\emph{Per-cell optimization.} Fix a realizable sign condition $\sigma$ and let $\bar S_\sigma\subseteq Z$ be its weak closure. On $\bar S_\sigma$ the action pattern is fixed (weakly optimal, by continuity of the comparisons), so as in the proof of Proposition~\ref{prop:vertex} the firm objective equals
\[
V_\sigma(p)=\text{const}_\sigma-\!\!\sum_{t\ \text{games on }j(t)}\!\!(1-p_{j(t)})a_t-\sum_j\kappa_jp_j ,
\]
affine in $p$ and independent of $(w,\theta)$: outcome values depend only on actions, labels (including $y^I_t$, fixed by $\sigma$), and detection probabilities. The cell optimum $\max_{z\in\bar S_\sigma}V_\sigma(p)$ is a linear optimization over a semialgebraic set described by $P$ polynomial inequalities of degree $\le3$ in $\le2d+1$ variables, solvable exactly (optimal value as an algebraic number together with a sample point) in time $P^{O(d)}=(Nd)^{O(d)}$ by the critical-point/optimization algorithms of \cite{basu2006algorithms}.

\emph{Correctness.} Under leader-favorable selection, for every $z$ the realized value is $V_{\sigma'}(p)$ for some pattern $\sigma'$ weakly optimal at $z$, whence $z\in\bar S_{\sigma'}$; conversely on each $\bar S_\sigma$ the realized value is at least $V_\sigma(p)$. Therefore $\sup_ZV=\max_\sigma\max_{\bar S_\sigma}V_\sigma$, the maximum over the enumerated cells of the per-cell optima, and it is attained because each $\bar S_\sigma$ is compact. Multiplying the cell count by the per-cell cost gives total time $(Nd)^{O(d)}$, polynomial in $N$ for fixed $d$. \qed

\paragraph{Remark (what is and is not claimed).}
Theorem~\ref{thm:fixeddim} is a fixed-dimension result; the exponent's dependence on $d$ is unavoidable in this generality, since Theorem~\ref{thm:hardness} already makes the inner allocation NP-hard when the number of channels grows. The algorithm is an enumeration certificate, not a practical procedure; its role is to show that the rent-profile factorization \eqref{eq:two-stage} reduces joint design to finitely many combinatorial rent patterns, each carrying a one-dimensional-per-channel allocation problem.

\section{Learning proofs}\label{app:learning}

Throughout, $V(p)=\E_q\big[a_t\indic\{t\text{ deterred at }p\}\big]-\kappa\cdot p$ with $a_t\in[0,1]$ and $\kappa_j\le1$, so realized payoffs lie in $[-m,1]$. Deterrence at $p$ means $p_j\Delta_t\ge D_{j,t}$ for every profitable channel $j\in J_t$ (Theorem~\ref{thm:portfolio}), a condition that is monotone in $p$: if $t$ is deterred at $p$ and $p'\ge p$ coordinatewise, $t$ is deterred at $p'$.

\subsection{Proof of Theorem~\ref{thm:learning}(i)}

Let $L:=\lceil T^{1/(m+2)}\rceil$, $w:=1/L\le T^{-1/(m+2)}$, and $G_w:=\{w,2w,\dots,1\}^m$, so $K:=|G_w|=L^m\le\big(T^{1/(m+2)}+1\big)^m\le e^{\,mT^{-1/(m+2)}}\,T^{m/(m+2)}\le e\,T^{m/(m+2)}$, the last step using $T\ge m^{m+2}$. For $p\in[0,1]^m$ let $\lceil p\rceil_w\in G_w$ round each coordinate up to the grid. By monotonicity every type deterred at $p$ is deterred at $\lceil p\rceil_w$, and the cost rises by at most $\sum_j\kappa_j(\lceil p_j\rceil_w-p_j)\le mw$; hence
\[
\max_{g\in G_w}V(g)\ \ge\ \max_{p\in[0,1]^m}V(p)-mw .
\]
Run Exp3 \cite{auer2002nonstochastic} over the $K$ grid points, feeding it the realized payoff of the point played, rescaled from $[-m,1]$ to $[0,1]$. Its expected regret against the best grid point is at most $(m+1)\cdot2\sqrt{(e-1)TK\ln K}$. Adding the discretization loss,
\begin{align*}
\mathrm{Reg}_T\ &\le\ Tmw+2(m+1)\sqrt{(e-1)\,T\,K\ln K}\\
&\le\ mT^{\frac{m+1}{m+2}}+2(m+1)\sqrt{e(e-1)}\;\sqrt{T\cdot T^{\frac{m}{m+2}}\cdot\ln\!\big(eT^{\frac{m}{m+2}}\big)} ,
\end{align*}
which is $O\big(m\,T^{(m+1)/(m+2)}\sqrt{\log T}\big)$. No step used the number of types or the form of $q$. \qed

\subsection{Proof of Theorem~\ref{thm:learning}(ii)}

\emph{Instance family.} Take $m=1$, $\Delta_t\equiv1$, $\kappa=1$, and $N:=\lceil T^{1/3}\rceil$ types with rents $D_j=j/N$ and values $a_j=1$, $j=1,\dots,N$; type $j$ is deterred iff $p\ge c_j:=j/N$. Under the base law $q^0$ (uniform on the $N$ types),
\[
V^0(p)=\tfrac1N\big\lfloor Np\big\rfloor-p\ \le\ 0,
\]
with equality exactly on the corners $\{c_0,c_1,\dots,c_N\}$, $c_0:=0$. For $i$ in the index set $I:=\{i:\ N/4\le i\le3N/4-1\}$ define $q^i$ by moving mass $\varepsilon$ from type $i+1$ to type $i$, where $\varepsilon:=\tfrac14\sqrt{|I|/T}$; this is a probability vector since $\varepsilon\le\tfrac14\sqrt{N/T}\le1/N$ once $N^3\le16T$, which holds for $T\ge8$, and $\varepsilon\le1/8$. Under $q^i$,
\[
V^i(p)=V^0(p)+\varepsilon\,\indic\{c_i\le p<c_{i+1}\},
\]
so $c_i$ is the unique maximizer with value $\varepsilon$, every other corner has value $0$, and every non-corner point is strictly worse than the corner to its left.

\emph{Reduction to finitely many arms.} Fix any $p\in[c_j,c_{j+1})$. The deterrence indicator observed at $p$ is Bernoulli with parameter $Q_j:=\sum_{k\le j}q_k$, the same as at $c_j$, and the realized payoff at $p$ is that at $c_j$ minus the known constant $p-c_j\ge0$. An algorithm playing $p$ can therefore be simulated by one playing $c_j$ with the same information and weakly smaller regret. It suffices to lower-bound regret over algorithms that play corners, i.e.\ an $(N+1)$-armed bandit in which arm $j$ pays $B_j-c_j$ with $B_j\sim\mathrm{Bernoulli}(Q_j)$.

\emph{Prefix sums.} Under $q^i$ the prefix sums satisfy $Q^i_j=Q^0_j$ for all $j\ne i$ and $Q^i_i=Q^0_i+\varepsilon$: moving mass between types $i$ and $i+1$ changes only the $i$-th prefix. So the $|I|$ laws $q^i$ differ from $q^0$ in the reward distribution of one arm each, by a Bernoulli shift of $\varepsilon$ at a parameter $Q^0_i\in[\tfrac14,\tfrac34]$, where $\mathrm{KL}\big(\mathrm{Ber}(Q)\,\|\,\mathrm{Ber}(Q+\varepsilon)\big)\le\varepsilon^2/\big(Q(1-Q-\varepsilon)\big)\le 8\varepsilon^2$ for $\varepsilon\le1/8$.

\emph{Change of measure.} This is the standard multi-armed lower-bound family: $|I|$ candidate arms, all with mean $0$ except one with mean $\varepsilon$, the identity of which is unknown. By the argument of \cite[Theorem 5.1]{auer2002nonstochastic} (see also \cite[Chapter 15]{lattimore2020bandit}), with the KL bound above in place of the one for parameter $\tfrac12$, any algorithm suffers expected regret at least $c_0\sqrt{|I|\,T}$ against a law $q^i$ chosen uniformly from $I$, for an absolute constant $c_0>0$ (the constant $1/20$ of \cite{auer2002nonstochastic} becomes $c_0=1/(20\sqrt8)$ after the KL rescaling). For $T\ge2^{12}$ we have $N\ge16$, so $|I|\ge N/2-2\ge N/4\ge T^{1/3}/4$, and $\mathrm{Reg}_T\ge c\,T^{2/3}$ with $c=c_0/2$. The values and the cost are at most $1$, and the instance has $\lceil T^{1/3}\rceil$ types. \qed

\subsection{Proof of Theorem~\ref{thm:learning}(iii)}

In the fixed-channel regime with nonincreasing bounded rent densities, channel $j$'s expected payoff $U_j(p_j):=V_j(p_j)-\kappa_jp_j$, with $V_j$ from \eqref{eq:cum-deterrence} and $\ell_j\equiv0$, is concave on $[0,\bar p_j]$ by Theorem~\ref{thm:auditdesign} and Lipschitz with constant $\Delta\sup_d a_jm_j(d)+\kappa_j$. The objective is additively separable, $\sum_jU_j(p_j)$, and the per-channel realized payoff is an unbiased, bounded observation of $U_j(p_j)$. Running a one-dimensional stochastic bandit convex optimization algorithm independently on each channel, for instance the algorithm of \cite[Section 3]{agarwal2011bandit}, gives regret $\tilde O(\sqrt T)$ per channel and $\tilde O(m\sqrt T)$ in total.

For the lower bound take $m=1$ and a linearly decreasing rent density, $a\,m(d)=\alpha-\beta d$ on $[0,\alpha/\beta]$. Then $U(p)=\alpha\Delta p-\tfrac12\beta\Delta^2p^2-\kappa p$, a concave quadratic whose maximizer $(\alpha\Delta-\kappa)/(\beta\Delta^2)$ moves with $\kappa$ at fixed curvature. Bandit optimization of a quadratic with unknown shift from noisy bounded evaluations has minimax regret $\Omega(\sqrt T)$ \cite{shamir2013complexity}; the change-of-measure argument there uses only that the noisy observations at a point $p$ under two shifts have KL divergence $O((U(p)-U'(p))^2)$, which holds for our bounded payoff noise as it does for Gaussian noise. \qed

\subsection{Proof of Theorem~\ref{thm:learning}(iv)}

Let every type have profitable set $\{1,2\}$, value $1$, differential $\Delta$, and rents $(D_1,D_2)$ independent and uniform on $[0,\Delta]$. A type is deterred iff $p_1\Delta\ge D_1$ and $p_2\Delta\ge D_2$, an event of probability $p_1p_2$, so $V(p)=p_1p_2-\kappa(p_1+p_2)$. Its Hessian is $\big(\begin{smallmatrix}0&1\\1&0\end{smallmatrix}\big)$ with eigenvalues $\pm1$, so $V$ is neither concave nor convex; for instance $V(\tfrac12,\tfrac12)<\tfrac12\big(V(1,1)+V(0,0)\big)$ when $\kappa<\tfrac12$. \qed

\section{Welfare proofs}\label{app:welfare}

\subsection{Proof of Theorem~\ref{thm:dominance}}

Containment $\Fcost\subseteq\Ftwo$ holds because $p\equiv0$ is feasible in the two-lever model. For the lift, as audit rises from $p$ to $p+dp$, the deterrence frontier moves by $\Delta\,dp$ in rent space. With rent density $\mu_j$, this flips mass $\mu_j\Delta\,dp$. In the perfect-audit benchmark each newly flipped gamer contributes one unit of correct-allocation welfare. Integrating up to $p^\star=D_{\max}/\Delta$ gives
\[
W'-W=\int_0^{p^\star}\mu_j\Delta\,dp
=\mu_jD_{\max}
=\mu_g .
\]
This is positive whenever $\mu_g>0$.

For the positive-measure reachability gap we make the regularity of the design-to-rate map explicit. Let $R:(w,\theta,\pi)\mapsto(\mathrm{FP},\mathrm{TP})\in[0,1]^2$ denote the population false-positive and true-positive rates induced by a design, so $\Fcost=R(\cdot,\cdot,0)$ and $\Ftwo=R(\cdot,\cdot,\cdot)$. Assume: \textnormal{(R1)} $R$ is continuous; \textnormal{(R2)} at the base cost-only outcome $z_0=R(w_0,\theta_0,0)$ the set $\Fcost$ is locally supported, i.e. there is a line $\ell$ through $z_0$ with $\Fcost$ contained in one closed half-plane $H^-$ in a neighborhood of $z_0$ (local exposedness); \textnormal{(R3)} applying the deterring audit $\pi^\star$ of the theorem at the fixed classifier yields $z_1=R(w_0,\theta_0,\pi^\star)$ that strictly improves one frontier coordinate and weakly improves the other, placing $z_1$ in the open complementary half-plane $\mathrm{int}\,H^+$; \textnormal{(R4)} $R$ is locally open at $(w_0,\theta_0,\pi^\star)$, for example because two smooth design directions have a full-rank Jacobian for $(\mathrm{FP},\mathrm{TP})$. Under \textnormal{(R3)}, $z_1\notin\Fcost$ near $z_0$. By \textnormal{(R4)}, a neighborhood of $(w_0,\theta_0,\pi^\star)$ maps to a two-dimensional neighborhood of $z_1$ inside $\mathrm{int}\,H^+$ after shrinking if necessary. By \textnormal{(R2)}, this neighborhood is disjoint from the local cost-only frontier. Hence $\Ftwo\setminus\Fcost$ has positive area. The conclusion is conditional on \textnormal{(R1)}--\textnormal{(R4)}; (R1) holds whenever acceptance rates vary continuously with the design (e.g. atomless features), while (R2)--(R4) are the substantive geometric hypotheses that a deterring audit makes a locally two-dimensional Pareto-improving move at $z_0$. \qed

\subsection{Generation protocol for Figure~\ref{fig:domain-audit-field}}\label{app:figure2}

Figure~\ref{fig:domain-audit-field} illustrates the theorem on a fixed-seed synthetic population. The classifier, costs, budgets, causal equation, and initial distribution are identical across the two domains. Only the audit product $p\Delta$ changes. In the weak-audit domain, rejected agents can profitably move the fakeable signal to the acceptance boundary, producing a visible gaming ridge and a high false-positive rate. In the strong-audit domain, the same movement is no longer worth its risk; agents who can afford causal improvement move in the causal direction, while the rest exit.

\begin{figure}[ht]
\centering
\includegraphics[width=0.8\linewidth]{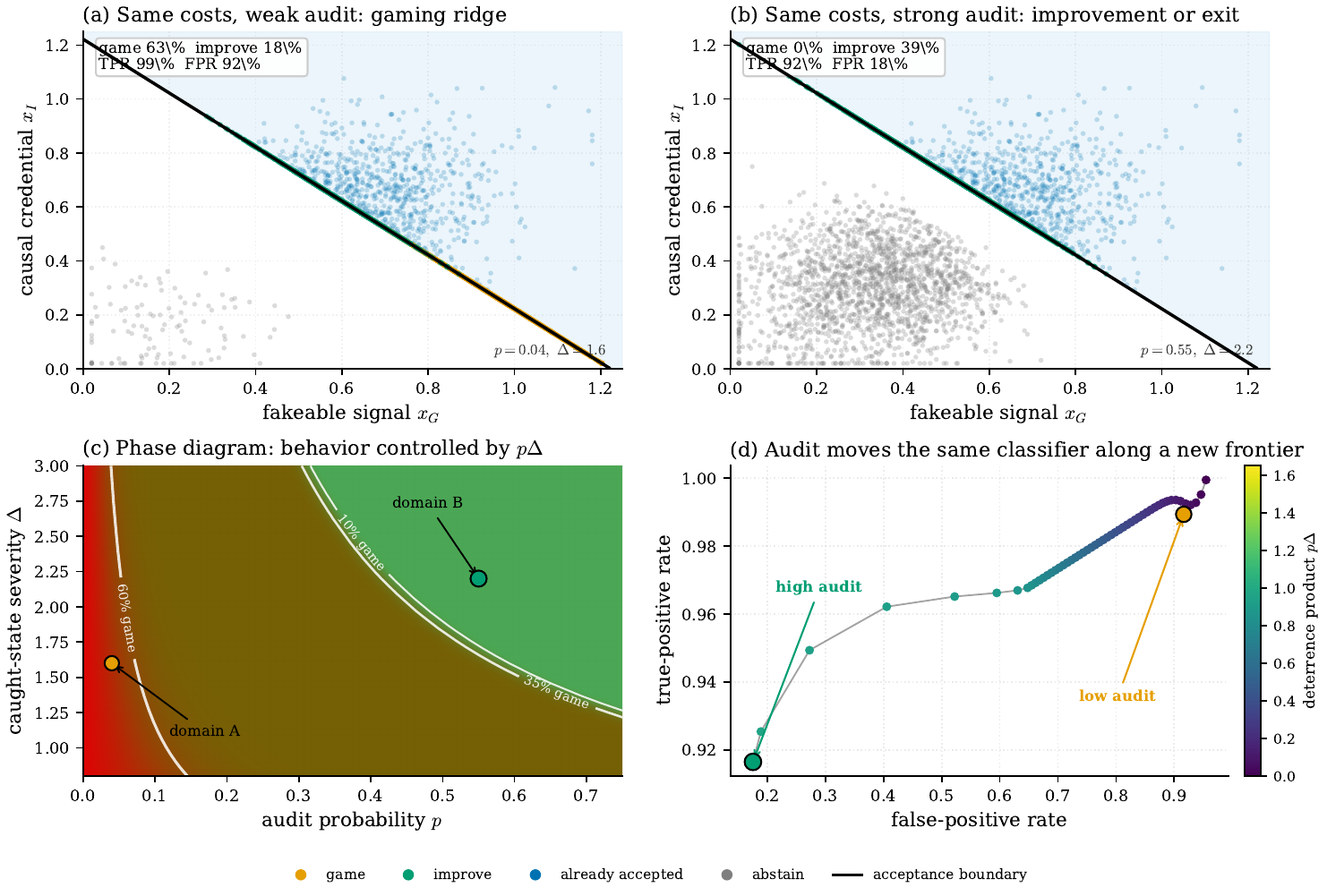}
\Description{See caption.}
\caption{Audit-induced domain separation in a heterogeneous population. The classifier, costs, budgets, and causal structure are fixed. Low audit creates a dense gaming ridge on the acceptance boundary; high audit removes gaming and leaves honest improvement plus exclusion. The lower panels show that the behavioral split is governed by the deterrence product $p\Delta$, and that changing audit moves the same classifier along a new TP/FP frontier.}
\label{fig:domain-audit-field}
\end{figure}

The lower panels are the identification diagnostic. Panel (c) varies only $p$ and $\Delta$ and shows that the aggregate action mix changes along the same deterrence frontier as Lemma~\ref{lem:deterrence-frontier}; panel (d) holds the classifier fixed and traces the induced movement in $(\mathrm{FP},\mathrm{TP})$ space as audit changes. A cost-only model has no parameter left to explain this split once costs and causal structure are held fixed. The two-lever model explains it by a single inequality, $p_A\Delta_A<D\le p_B\Delta_B$.

The script \texttt{fig\_domain\_audit\_field.py} generates Figure~\ref{fig:domain-audit-field}. The random seed is fixed at $20260607$ and the population size is $N=70{,}000$. The script first draws an ability variable $z\sim N(0,1)$ and two observed coordinates,
\begin{align*}
x_I&=\mathrm{clip}(0.43+0.15z+\varepsilon_I,0.02,1.18),\\
x_G&=\mathrm{clip}(0.46+0.13z+\varepsilon_G,0.02,1.18),
\end{align*}
with $\varepsilon_I\sim N(0,0.13^2)$ and $\varepsilon_G\sim N(0,0.15^2)$. The initial label is generated from the causal score
\[
y^0=\mathbf 1\{0.92x_I+0.28z+\eta>0.55\},\qquad \eta\sim N(0,0.16^2).
\]
Budgets are heterogeneous,
\[
B=\mathrm{clip}\{0.04+\mathrm{LogNormal}(\log 0.12,0.55),0.025,0.55\}.
\]
The displayed classifier is fixed throughout the figure: $s(x)=0.72x_I+0.72x_G$ with threshold $\theta=0.88$. Gaming and improvement costs are linear with $c_G=0.13$ and $c_I=0.46$, and acceptance utility is $u_+=1$.

For each agent below the threshold, the script computes the effort needed to cross the boundary through the fakeable coordinate and through the causal coordinate. Gaming has utility $u_+-c_Gg_0/0.72-p\Delta$ if it is within budget; honest improvement has utility $u_+-c_Ig_0/0.72$ if it is within budget; abstention has utility $0$. The chosen action is the utility-maximizing one, with accepted agents who already clear the threshold left unchanged. Gaming changes only the observed fakeable coordinate, while improvement changes the causal coordinate and recomputes post-improvement qualification using the same structural rule without the label noise term, $0.92x_I+0.28z>0.55$.

Panels (a) and (b) evaluate two domains, $(p,\Delta)=(0.04,1.60)$ and $(0.55,2.20)$, using the same population draw and the same classifier. The plotted points are a fixed subsample of $4{,}200$ agents; color encodes abstention, already-accepted status, gaming, or honest improvement after best response. Panel (c) recomputes the same best-response rule on a $125\times125$ grid with $p\in[0,0.75]$ and $\Delta\in[0.8,3.0]$, recording the shares of initially rejected agents who game, improve, or abstain. Panel (d) fixes $\Delta=2.20$, sweeps $p$ over $85$ values, and plots the deterministic TP/FP rates induced by the same classifier as $p\Delta$ changes. No fitted model is used in the plot; it is a deterministic evaluation of the structural best-response rule on a fixed synthetic population.

\subsection{Proof of Theorem~\ref{thm:domainsep}}

The shared structure fixes the rent $D$ and, by assumption, honest improvement is feasible and individually rational. In any cost-only model, $p=0$, so the deterrence condition $p\Delta\ge D$ fails whenever $D>0$. Therefore the binding type's behavior cannot differ across domains with the same cost and causal primitives; opposite observed behavior refutes the cost-only specification.

In the two-lever model, Lemma~\ref{lem:deterrence-frontier} gives: domain $A$ games iff $p_A\Delta_A<D$, and domain $B$ improves iff $p_B\Delta_B\ge D$. This is exactly the condition stated. If audit probabilities are observed, these inequalities imply the sharp set stated in the theorem. Sharpness holds because every pair of differentials satisfying the inequalities reproduces the behavior, and every pair violating them contradicts one of the two observations. If $p_A=p_B=p$, then $\Delta_A<D/p\le\Delta_B$. Since $\Delta=u_+-u_-$ and $u_+$ is common, this implies $u_{-,B}<u_{-,A}$. \qed

\subsection{Proof of Theorem~\ref{thm:price} and Remark~\ref{rmk:headroom}}

\begin{rmk}[Audit substitutes for signal headroom]\label{rmk:headroom}
Boundedness in (ii) is not cosmetic. If the same feature is unbounded, the cost-only firm can raise the threshold one gaming-budget above the qualified position: the qualified fake up cheaply from their head start, the unqualified hit their participation constraint, and Spence separation restores the full value $q$ with no audit. Randomization \cite{braverman2020randomness} narrows but cannot close the capped gap: any random acceptance schedule that deters gaming accepts the qualified with probability at most $c_G/u_+$ (apply the deterrence constraint at the cap), so it forgoes at least $q(1-c_G/u_+)$ of the audit dividend. The audit lever is most valuable exactly where screening-by-signaling fails - bounded credentials, capped scores, saturated reports - which is where verification is used in practice.
\end{rmk}

\emph{(i)} The lower bound holds because $p\equiv0$ is feasible in the two-lever problem. For the upper bound, under any design and any realized responses each agent's contribution to the firm objective is at most its best feasible outcome value, and audit costs are nonnegative, so $W^{\mathrm{two}}\le W^{\mathrm{ub}}$. Conversely the cost-only designer may zero every gameable coordinate; gaming then leaves the score unchanged, so no agent games (any positive gaming cost is wasted) and the design's value is its honest-response value, whence $W^{\mathrm{cost}}\ge W^{-\Gset}$. Subtracting gives $W^{\mathrm{two}}-W^{\mathrm{cost}}\le W^{\mathrm{ub}}-W^{-\Gset}=V_{\Gset}$.

\emph{(ii)} In the family, observed reports are capped: $x\in[0,1]$, qualified mass $q$ at $x^0=1$ with label $1$, unqualified mass $1-q$ at $x^0=0$ with label $0$, gaming cost $c_G<u_+$ per unit under \textnormal{(Lin)}, no improvable coordinates, common $u_+$ and $\Delta$.

\emph{Cost-only.} A design is $(w,\theta)$ on the single feature. If $w>0$, write $\tau=\theta/w$. For $\tau\le0$ everyone is accepted: value $2q-1$. For $\tau\in(0,1]$ the qualified are accepted at no cost while an unqualified agent fakes $\tau\le1$ units (feasible under the cap) at cost $c_G\tau\le c_G<u_+$, so all unqualified game and are accepted: value $2q-1$. For $\tau>1$ no observed report reaches $\tau$: value $0$. If $w\le0$, gaming ($\gamma\ge0$) cannot help and the acceptance set contains the unqualified position whenever it contains the qualified one, so the value is again at most $\max(0,2q-1)$. Hence $W^{\mathrm{cost}}=\max(0,2q-1)$, which is $0$ for $q\le\tfrac12$.

\emph{Two-lever.} Keep $\tau=1$ and audit the feature at intensity $p$. With no honest recourse, an unqualified agent games iff $c_G+p\Delta<u_+$, so $p_g=(u_+-c_G)/\Delta$ deters all gaming (weak inequality, leader-favorable). Wrongly flagged honest applicants lose acceptance at rate $\alpha p_g$ as in \S\ref{sec:risk-tuning}, and detection costs $c_ap_g$:
\[
W^{\mathrm{two}}\ \ge\ q(1-\alpha p_g)-c_ap_g .
\]
With a perfect audit ($\alpha=0$, $c_a=0$) this is $q$; since $W^{\mathrm{ub}}=q$ (the unqualified have no path to qualification), $W^{\mathrm{two}}=q$ exactly. Finally $W^{-\Gset}=\max(0,2q-1)=0$ because zeroing the only feature leaves a constant score, so $V_{\Gset}=q-0=q$ and the gap $W^{\mathrm{two}}-W^{\mathrm{cost}}=q$ attains the bound of (i). The ratio is unbounded: $W^{\mathrm{cost}}=0<q=W^{\mathrm{two}}$.

\emph{Remark~\ref{rmk:headroom}.} Remove the cap. Any $\tau^\star\in[u_+/c_G,\,u_+/c_G+1)$ is nonempty as an interval and gives: the qualified fake $\tau^\star-1$ units at cost $c_G(\tau^\star-1)<u_+$ and are accepted, while an unqualified agent would pay $c_G\tau^\star\ge u_+$ and abstains (at equality, abstention is the leader-favorable tie). The cost-only value is then $q$ with no audit: unbounded signal headroom restores Spence separation, and the gap of (ii) collapses.

For the randomization claim, restore the cap and let $P(\tau)$ be the acceptance probability of an observed report $\tau\in[0,1]$ under a random cost-only rule. If the rule deters gaming, then $u_+P(\tau)-c_G\tau\le0$ for every $\tau$ the unqualified could fake; taking $\tau=1$ gives $P(1)\le c_G/u_+$, so the qualified (whose report is $1$) are accepted with probability at most $c_G/u_+$ and the rule's value is at most $qc_G/u_+$, which is the stated shortfall relative to the perfect-audit value $q$. \qed

\section{Risk-tuning proofs}\label{app:risk}

\subsection{Proof of Proposition~\ref{prop:overpen}}

An unqualified agent games iff $u_+-p\Delta-c_G>0$, i.e. iff $p<p_g=(u_+-c_G)/\Delta$. A qualified agent applies iff $u_+-\alpha p\Delta\ge0$, i.e. iff $p\le p_q=u_+/(\alpha\Delta)$. Since $\alpha<1$ and $c_G>0$, $p_g<p_q$.

Accuracy is
\[
A(p)=q(1-\alpha p)\indic\{p\le p_q\}+(1-q)\big[p\,\indic\{p<p_g\}+\indic\{p\ge p_g\}\big].
\]
For $p<p_g$, gaming is active and the unqualified are correctly rejected only when caught. At $p_g$, the unqualified term jumps from $p_g$ to $1$. For $p_g\le p<p_q$, the unqualified term is already $1$ and the qualified term has slope $-q\alpha<0$. For $p\ge p_q$, the qualified abstain and accuracy is the floor $1-q$. Thus the unique maximizer is $p_g$ whenever the deterrence jump dominates the cost-only corner, $A(p_g^+)=1-q\alpha p_g>q=A(0)$, exactly the condition stated in the body. Adding a linear or convex false-accusation harm to wrongly flagged qualified agents only makes the post-$p_g$ slope more negative, so utilitarian welfare peaks at the same deterrence-binding point under the analogous corner-dominance condition. \qed

\subsection{Detection-penalty composition}

\begin{prop}[Detection-penalty composition]\label{prop:composition}
Among feasible deterrence-binding designs ($p\Delta=D$, $\Delta\in[D,\Delta_{\max}]$, detection cost $c_ap$, false-accusation harm $\kappa\varphi(\Delta)$ with $\varphi$ increasing and strictly convex), the firm objective $U_f(\Delta)=1-(q\alpha+c_a)D/\Delta$ is strictly increasing in $\Delta$, so a private firm chooses maximal admissible punishment and minimal detection. Utilitarian welfare $W(\Delta)=U_f(\Delta)-\kappa q\alpha(D/\Delta)\varphi(\Delta)$ has at most one interior maximizer; if it lies in $(D,\Delta_{\max})$, it is the unique social optimum and solves
\begin{equation}\label{eq:composition-optimum}
\kappa q\alpha\,[\Delta^\star\varphi'(\Delta^\star)-\varphi(\Delta^\star)]
=q\alpha+c_a .
\end{equation}
For $\varphi(\Delta)=\Delta^2/2$, $\Delta^\star=\sqrt{2(q\alpha+c_a)/(\kappa q\alpha)}$, independent of the gaming rent $D$; higher rents should raise detection $p^\star=D/\Delta^\star$, not severity.
\end{prop}

\noindent On the deterrence-binding frontier $p\Delta=D$, so $p=D/\Delta$. With gaming deterred, accuracy is $1-q\alpha p=1-q\alpha D/\Delta$. The firm subtracts detection cost $c_ap$, hence
\[
U_f(\Delta)=1-(q\alpha+c_a)D/\Delta .
\]
Its derivative is $(q\alpha+c_a)D/\Delta^2>0$, so the firm chooses maximal allowed severity.

Social welfare subtracts false-accusation harm on mass $q\alpha p$:
\[
W(\Delta)=1-(q\alpha+c_a)D/\Delta-\kappa q\alpha(D/\Delta)\varphi(\Delta).
\]
Differentiating,
\[
W'(\Delta)=\frac{D}{\Delta^2}\Big[(q\alpha+c_a)-\kappa q\alpha(\Delta\varphi'(\Delta)-\varphi(\Delta))\Big].
\]
The function $g(\Delta)=\Delta\varphi'(\Delta)-\varphi(\Delta)$ has derivative $g'(\Delta)=\Delta\varphi''(\Delta)>0$ under strict convexity, so the bracket crosses zero at most once. If the root lies in the feasible interval $(D,\Delta_{\max})$, the derivative changes from positive to negative there and this root is the unique constrained social optimum, giving \eqref{eq:composition-optimum}. Otherwise the constrained optimum is the appropriate endpoint. For $\varphi(\Delta)=\Delta^2/2$, $g(\Delta)=\Delta^2/2$, yielding the displayed closed form. \qed

\subsection{Over-penalization diagnostic}

\begin{figure}[ht]
\centering
\includegraphics[width=.92\linewidth]{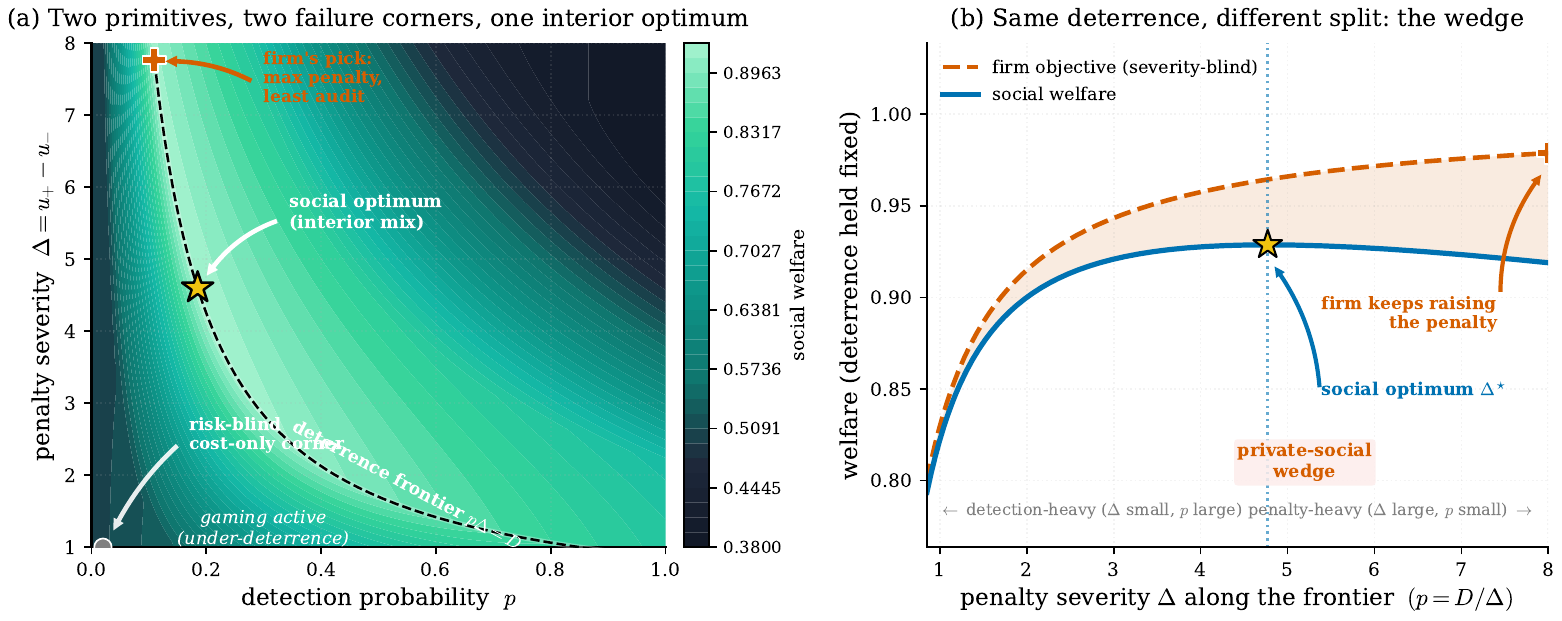}
\Description{See caption.}
\caption{Detection and punishment are behavioral substitutes but welfare complements. Along the frontier $p\Delta=D$, deterrence is fixed; private firm value drifts toward maximal punishment, while social welfare peaks at an interior mix.}
\label{fig:risk-landscape}
\end{figure}

\begin{figure}[ht]
\centering
\includegraphics[width=.92\linewidth]{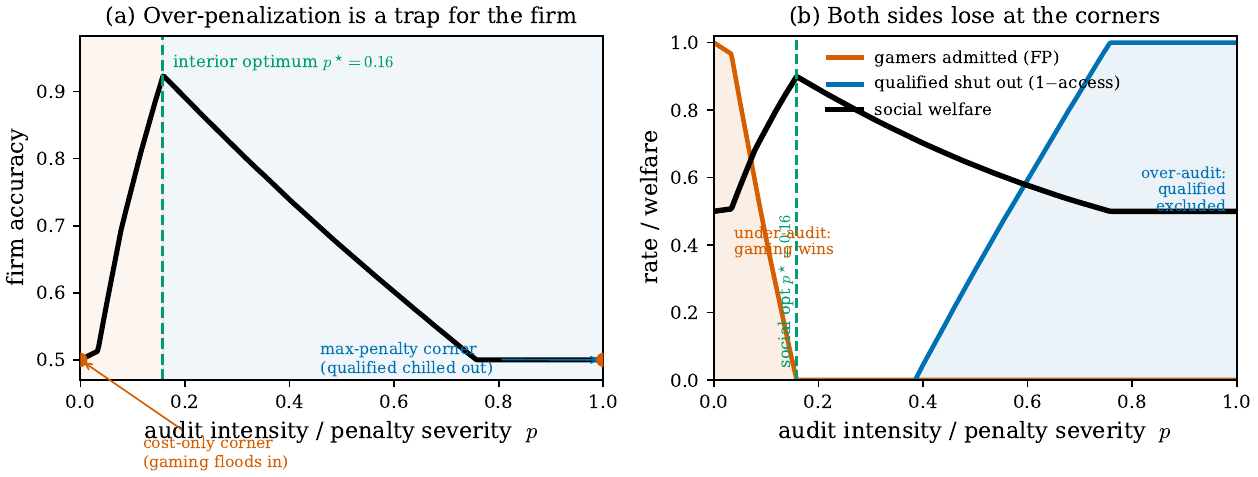}
\Description{See caption.}
\caption{Diagnostic over-penalization trap in the scalar model. Accuracy and welfare peak at the deterrence-binding audit. Below that point, gamers enter; above that point, there are no additional gamers to deter and false positives begin to chill qualified agents.}
\label{fig:overpenalization}
\end{figure}

\subsection{Generation protocols for Figures~\ref{fig:risk-landscape} and~\ref{fig:overpenalization}}

Figure~\ref{fig:risk-landscape} is generated by \texttt{fig\_risk\_landscape.py}. The seed is fixed at $11$ and the population size is $N=120{,}000$, split evenly between qualified and unqualified agents. Unqualified agents have heterogeneous gaming rents $D\sim\mathrm{Unif}[0.15,0.85]$ and are deterred when $p\Delta\ge D$. Qualified agents have outside options $o\sim\mathrm{Unif}[0,0.65]$ and apply when the honest expected utility $u_+-\alpha p\Delta$ exceeds $o$, with $u_+=1$ and $\alpha=0.16$. The firm pays audit operating cost $c_ap$ with $c_a=0.12$. Social welfare additionally subtracts false-accusation harm $\kappa q\alpha p\,\varphi(\Delta)$ on applying qualified agents, with $\kappa=0.22$ and $\varphi(\Delta)=\Delta^2/2$.

The script evaluates both the private firm objective and social welfare on a $240\times240$ grid with $p\in[0.001,1]$ and $\Delta\in[1,8]$. Panel (a) is the social-welfare surface on this grid, with the deterrence frontier for the largest rent, $p\Delta=0.85$, overlaid. The plotted markers are the global social optimum, the firm's severity-blind optimum, and the risk-blind low-audit corner. Panel (b) restricts attention to the deterrence frontier $p=0.85/\Delta$, recomputes firm value and social welfare along that curve, and shades the private-social wedge. Thus the panel separates a behavioral equivalence - the same $p\Delta$ deters the same marginal gamer - from a welfare nonequivalence: detection probability and penalty severity have different side effects.

Figure~\ref{fig:overpenalization} is generated by \texttt{fig\_overpenalization.py}. The seed is fixed at $3$ and the population size is $N=200{,}000$, again with qualified mass $q=0.5$. The caught-state differential is $\Delta=6$ from $u_+=1$ and $u_-=-5$, and the audit false-positive rate for honest applicants is $\alpha p$ with $\alpha=0.22$. Qualified outside options are drawn from $\mathrm{Unif}[0,0.9]$, so qualified agents apply iff $u_+-\alpha p\Delta\ge o$. Unqualified gaming costs are drawn from $\mathrm{Unif}[0.05,0.8]$, so an unqualified agent games iff $u_+-p\Delta-c>0$. The script evaluates $400$ audit levels $p\in[0,1]$, recording firm accuracy, qualified access, false-positive admission of gamers, and social welfare after subtracting false-accusation harm with weight $\kappa=1.6$. Panel (a) plots firm accuracy and marks the interior audit that maximizes it; panel (b) plots the two harms that flank the optimum - admitted gamers under low audit and excluded qualified agents under high audit - together with social welfare. The figure is a numerical diagnostic for Proposition~\ref{prop:overpen}, not an additional assumption in its proof.

Figure~\ref{fig:overpenalization} is a numerical illustration of Proposition~\ref{prop:overpen} in the scalar model. It separates two mistakes that are easy to conflate: the cost-only corner under-deters, while the surveillance corner over-penalizes. The optimum is the boundary between the two, not maximal audit.

\section{Robustness experiment: construction and parameters}\label{app:experiments}

This appendix documents the three-panel robustness experiment of Figure~\ref{fig:robustness}. All runs are deterministic with fixed seeds and are reproducible from \texttt{fig\_robustness\_stress\_test.py} (simulation) and \texttt{plot\_robustness\_stress\_test.py} (rendering). Each panel relaxes one idealization of the body and reports the \emph{audit dividend}, the optimal two-lever welfare minus the optimal cost-only welfare, so a positive value certifies that pricing-and-auditing a fakeable signal strictly beats suppressing it.

\subsection{Shared economy}

A population carries a causal coordinate $x_I$ and a purely fakeable coordinate $x_G$. Initial labels are $y^0\sim\mathrm{Bernoulli}(0.45)$, and conditional on $y^0$ the features are drawn from class-dependent normals clipped to $[0,1.2]$,
\[
\begin{array}{ll}
x_I\mid y^0{=}1\sim N(0.58,0.15^2), & x_I\mid y^0{=}0\sim N(0.45,0.16^2),\\
x_G\mid y^0{=}1\sim N(0.88,0.11^2), & x_G\mid y^0{=}0\sim N(0.24,0.13^2),
\end{array}
\]
so the fakeable coordinate is genuinely predictive in the honest population (this is what makes suppressing it costly). Budgets and the caught-state loss component are
\[
\begin{aligned}
B&=\mathrm{clip}\{0.02+\mathrm{LogNormal}(\log 0.08,0.55),\,0.015,\,0.40\},\\
L&=\mathrm{clip}\{\mathrm{LogNormal}(\log 0.55,0.60),\,0.05,\,8.0\},
\end{aligned}
\]
with stakes-risk differential $\Delta=u_++L$. The remaining constants are $u_+=1$, causal-improvement cost $c_I=0.95$, gaming cost $c_G=0.05$, post-improvement qualification threshold $q_\theta=0.82$, false-positive penalty $1.25$, recourse credit $0.03$, and audit operating cost $0.04$. For a design $(w_G,\theta,p)$ the score is $s(x)=(1-w_G)x_I+w_Gx_G$; an initially rejected agent abstains, improves the causal coordinate when $c_Ig_0/(1-w_G)\le B$, or games the fakeable coordinate when $c_Gg_0/w_G\le B$, choosing the option of highest expected utility. Improvement updates the true label through $x_I+\delta\ge q_\theta$; gaming leaves the true label unchanged and risks the audit penalty. The firm objective is
\begin{equation}\label{eq:exp-objective}
U=\mathrm{TP}-1.25\,\mathrm{FP}+0.03\,\mathrm{converted}-0.04\,p ,
\end{equation}
with rates taken as population averages after best response. Welfare $W$ additionally subtracts false-accusation harm where audits err (panel (a)). Every optimum below is a deterministic grid search in which agents are re-optimized at each candidate design.

\subsection{Panel (a): audits that err on the honest}

Population size $N=12{,}000$, seed $20260628$, single fakeable channel. An audit at intensity $p$ wrongly flags a genuinely accepted applicant with probability $\alpha p\,\sigma$, where $\sigma=w_Gx_G/s(x)\in[0,1]$ is the fakeable share of the applicant's score and $\alpha=0.85$; a flagged honest applicant becomes a false rejection (an accuracy loss) and inflicts social harm $0.9$ (a welfare loss). The two-lever optimum maximizes \eqref{eq:exp-objective} over $w_G\in[0,0.95]$ ($40$ points), $p\in[0,1]$ ($41$ points), and $\theta\in[0.32,0.95]$ ($40$ points); the cost-only optimum fixes $p=0$. The panel sweeps $p$ at the two-lever-optimal weight $w_G^\star=0.39$ and plots firm value, welfare, gaming mass, and wrongful-flag mass, together with the best cost-only value as a horizontal reference. Firm value peaks at $p^\star\approx 0.80$ and welfare at $p\approx 0.75$: the privately optimal audit exceeds the welfare-optimal one because the firm does not bear false-accusation harm, the finite-population image of Proposition~\ref{prop:composition}.

\subsection{Panel (b): magnitude-dependent detection}

Population size $N=5{,}000$, seed $20260629$. The fakeable feature can be raised through two manipulation channels with unit costs $c_A=0.050$ and $c_B=0.062$. A fake of total size $g=g_0/w_G$ is split as $\gamma_A=sg,\ \gamma_B=(1-s)g$ over $s\in\{0,0.1,\dots,1\}$, and detection on channel $k$ is convex in the magnitude routed through it,
\[
p_k(\gamma)=p\,\big(1-e^{-(\beta\gamma)^2}\big),\qquad
P_{\text{catch}}=1-(1-p_A(\gamma_A))(1-p_B(\gamma_B)),
\]
with audit intensity $p$ the detection ceiling. Large $\beta$ recovers support-based detection (any positive fake is caught at rate $p$, so a singleton fake is optimal because a second channel only raises $P_{\text{catch}}$); small $\beta$ makes small fakes evasive, so splitting a fake lowers total detection and a diversified fake can dominate. The agent picks the split maximizing gaming utility; a fake is counted diversified when the optimal $s\in(0.12,0.88)$. For each $\beta\in\{40,16,8,4,2.5,1.6,1.0,0.6,0.35\}$ both the two-lever optimum (over $w_G$ on $9$ points, $\theta$ on $12$, $p$ on $8$) and the cost-only optimum are recomputed. The panel plots the dividend, the diversified-fake share, and the two-lever-optimal weight against magnitude-sensitivity $1/\beta$. The dividend stays positive and the optimal weight stays on across the range in which detection is meaningfully magnitude-sensitive; only when even large fakes evade detection does the firm revert to suppression.

\subsection{Panel (c): partial causal leakage}

Population size $N=5{,}000$, seed $20260629$, single fakeable channel. The fakeable coordinate carries a causal fraction $\rho\in[0,1]$: a gamer raising the observed feature by $\gamma$ raises the true causal index by $\rho\gamma$, so $y^{\text{post}}=y^0\vee(x_I+\rho\gamma\ge q_\theta)$. Pure fakeability is $\rho=0$; $\rho=1$ makes gaming indistinguishable from improvement. For each $\rho$ on $9$ points the two-lever optimum (over $w_G$ on $12$ points, $\theta$ on $18$, $p$ on $10$) and the cost-only optimum are recomputed. The panel plots the dividend, the share of would-be gamers converted into genuinely qualified agents, and the cost-only-optimal weight on the signal. The dividend is largest at $\rho=0$ and falls to zero as conversion rises, while the optimal weight stays strictly positive throughout: causal leakage substitutes for audit but never makes suppression optimal, which is why purely fakeable coordinates are the regime in which the audit lever is most valuable.

\end{document}